\documentclass[12pt]{article}
\usepackage{amsmath}
\usepackage{graphicx,psfrag,epsf}
\usepackage{enumerate}
\usepackage{natbib}
\usepackage{url} 

\newcommand{\blind}{1}

\usepackage{amssymb}
\usepackage{dsfont}
\usepackage{subcaption}
\usepackage{amsthm}

\usepackage{thmtools,thm-restate}

\newtheorem*{remark}{Remark}

\newcommand{\T}{\intercal}

\newtheorem{theorem}{Theorem}[section]
\newtheorem{proposition}[theorem]{Proposition}

\begin{document}

\def\spacingset#1{\renewcommand{\baselinestretch}%
{#1}\small\normalsize} \spacingset{1}


\if1\blind
{
  \title{\bf Variable selection using pseudo-variables}
  \author{Wenhao Hu, Eric Laber, Leonard Stefanski \\
    Department of Statistics, North Carolina State University}
  \maketitle
} \fi

\if0\blind
{
  \bigskip
  \bigskip
  \bigskip
  \begin{center}
    {\LARGE\bf Variable selection using pseudo-variables}
\end{center}
  \medskip
} \fi

\bigskip
\begin{abstract}
Penalized regression has become a standard tool for model building
across a wide range of application domains. Common practice is to tune 
the amount of penalization to tradeoff bias and variance or to optimize
some other measure of performance of the estimated model. An advantage
of such automated model-building procedures is that their operating
characteristics are well-defined, i.e., completely data-driven, 
and thereby they can be systematically studied. However, in many
applications it is desirable to incorporate domain knowledge into 
the model building process; one way to do this is to characterize each model along the solution path of a penalized regression estimator 
in terms of an operating characteristic that is meaningful within a 
domain context and then to allow domain experts to choose from among 
these models using these operating characteristics as well as other factors not available to the estimation algorithm. We derive an 
estimator of the false selection rate for each model along the solution path using a novel variable addition method. The proposed estimator applies to both fixed and random designs and allows for $p \gg n$. 
The proposed estimator can be used to estimate a model with a
pre-specified false selection rate or can be overlaid on the 
solution path to facilitate interactive model exploration. We 
characterize the asymptotic behavior of the proposed estimator 
in the case of a linear model under a fixed design; however, 
simulation experiments show that the proposed estimator provides consistently more accurate estimates of the false selection rate 
than competing methods across a wide range of models.

\
\end{abstract}

\noindent%
\vfill

\newpage
\spacingset{1.45} 

\section{Introduction}
\label{sec:intro}
Penalized regression is now a primary tool for model building across a
wide range of application domains.  The operating characteristics of
penalized regression estimators can depend critically on tuning
parameters which govern the amount of penalization.  Accordingly,
there is an extensive literature on tuning parameter selection
including information-based criteria \citep{chen2008extended,
  wang2009shrinkage, zhang2010regularization, fan2013tuning,
  hui2015tuning}, resampling methods \citep{hall2009bootstrap,
  meinshausen2010stability, feng2013consistent, sun2013consistent,
  shah2013variable, sabourin2015permutation}, and variable addition
methods \citep{wu2007controlling, barber2015controlling,
  barber2016knockoff}.  However, these methods are typically used to
facilitate black-box estimation wherein model selection and fitting
are completely automated, i.e., data-driven, so as to produce a single
estimated model.  Complete automation is desirable in some contexts,
e.g., benchmarking or online estimation and prediction, and some level
of automation in model-building is unavoidable except in very small
problems.  However, it is often desirable to incorporate
domain knowledge into the model building process; one way to do this
is to characterize each candidate model along the solution path of a
penalized regression estimator in terms of its operating
characteristics and then to use these operating characteristics to
choose among candidate models.

We derive an estimator of the false selection rate for each model along
the solution path using a novel variable addition method.
The proposed estimator applies
to both fixed and random designs and allows for
$p \gg n$. The proposed estimator can be used to estimate a model
with a pre-specified false selection rate or can be overlaid on the solution
path to facilitate interactive model exploration.
Figure \ref{fig:prostate_inter} shows an example of such a solution
path using data from a study on prostate cancer \citep{stamey1989prostate};
this figure is a screen capture from the software provided in the Supplemental Materials
that allows the analyst to mouse-over any point on the solution path
and examine the estimated coefficient values as well as the estimated
false selection rate.  In this example, the selected point on the solution
path corresponds to a model with three selected variables,
log cancer volume (lcavol); log weight (lweight); and seminal vesicale
invasion (svi).   The estimated false selection rate corresponding to
this model is 0.10
(additional details are provided in Section \ref{sec:real}.)

\begin{figure}
    \centering
    \includegraphics[scale =0.3]{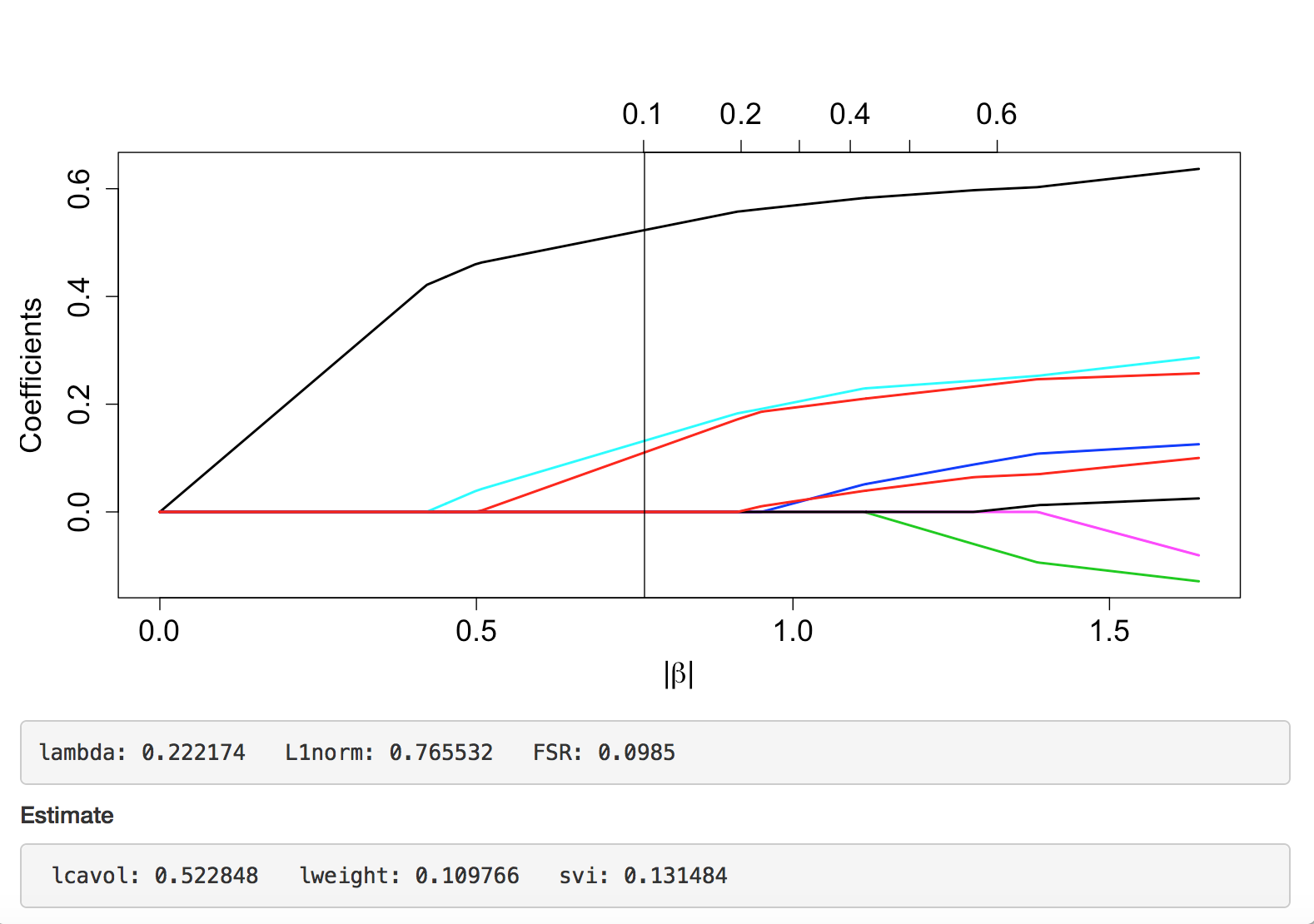}
    \caption{Lasso solution path for prostate cancer data.
     FSR and coefficient estimates are designed to be shown interactively.}
    \label{fig:prostate_inter}
\end{figure}

The proposed estimator of the false selection rate depends on the
generation of pseudo-variables that
are conditionally independent of the response
given the important  variables in the model.  As the true important
variables are unknown in practice, our estimator consists of three steps: (i) initial
variable screening to estimate the set of important variables;
(ii) generation of pseudo-variables so that
the covariance structure between the pseudo-variables and those
selected in the screening step mimics the covariance structure between the
not-selected and selected variables in the screening step; and (iii)
fitting the penalized estimator and using the proportion of selected
pseudo-variables to construct an estimator of the false selection
rate.
The proposed methodology is an example of a noise-variable or
knock-off variable method.  Such methods have been applied
to control the false selection rate in forward selection
\citep{wu2007controlling} and for the Lasso
\citep{barber2015controlling, barber2016knockoff}. A primary contribution of this
work is an estimator of the false selection rates for a sequence of
tuning parameter values
$\lambda_{(1)}, \lambda_{(2)}, \ldots, \lambda_{(m)}$ along the solution path
that applies when $p \gg n$.
When the proposed method is used to tune the amount of
penalization so as to achieve a target false selection rate, it provides
better empirical
performance than alternatives in simulation experiments.
Our theoretical and methodological
developments focus on a linear model
estimated using the Lasso \citep{tibshirani1996regression} under a
fixed design; however,
simulation experiments illustrate broader applicability.
To facilitate the interactive model building, we have
implemented the proposed methods in an R package and
a shiny web application both of which are contained in the
Supplemental Materials.

In Section \ref{sec:meth}, we establish notation, describe the proposed
estimator, and state some of its theoretical properties. In Section
\ref{sec:sim}, we demonstrate the finite-sample performance of the proposed
method in a suite of simulation experiments. In Section \ref{sec:real}, we
illustrate application of the proposed method using the data from prostate
cancer study and leukemia cancer study. Concluding remarks are made in Section
\ref{sec:conc}.

\section{Methods}
\label{sec:meth}
\subsection{Setup and notation}
We consider data from a linear model under a fixed design. The observed data
are $\left\lbrace (X_i, Y_i)\right\rbrace_{i=1}^{n}$ and it is assumed that
$Y_i = X_i^{\T}\beta_0 + \epsilon_i$, where $\epsilon_1,\ldots, \epsilon_n
\buildrel{iid}\over{\sim} \mathrm{Normal}(0,\sigma^2)$, and $\beta_0 =
(\beta_{0,1},\ldots, \beta_{0,p})^{\T} \in \mathbb{R}^p$. Define $\mathbb{X} =
(X_1,\ldots, X_n)^{\T} \in \mathbb{R}^{n\times p}$
 and $\mathbb{Y} = (Y_1,\ldots, Y_n)^{\T}$.  Given tuning parameter
$\lambda > 0$, the Lasso estimator of $\beta_0$ is
\begin{equation*}
\widehat{\beta}_{n}(\lambda; \mathbb{Y}, \mathbb{X}) =
\arg\min_{\beta \in\mathbb{R}^p}\left\lbrace
\frac{1}{2n}||\mathbb{Y} - \mathbb{X}\beta||^2 + \lambda
\sum_{j=1}^{p}|\beta_j|
\right\rbrace.
\end{equation*}
Define $A_{0} = \left\lbrace j\,:\, \beta_{0,j} \ne 0\right\rbrace$ to
be the index set of nonzero coefficients in the true model
and let
$\widehat{A}_{n}(\lambda) = \left\lbrace j\,:\,
\widehat{\beta}_{n,j}(\lambda; \mathbb{Y}, \mathbb{X}) \ne 0
\right\rbrace$ denote the active set at $\lambda$.   For
any $S \subseteq \left\lbrace 1,\ldots, p\right\rbrace$,
write $\mathbb{X}_{S}$ to denote the design matrix composed
of variables indexed by $S$; let $S^c$ denote the complement of $S$
and $\mathcal{N}(S)$ the number of elements in $S$.
Define $\boldsymbol{\Sigma} = n^{-1}\left(\mathbb{X}_{A_0}, \mathbb{X}_{A_0^C}
\right)^{T}\left(\mathbb{X}_{A_0}, \mathbb{X}_{A_0^C}
\right)$; $\widehat{I}_{n}(\lambda) = \mathcal{N}\left
\lbrace \widehat{A}_{n}(\lambda) \bigcap
  A_0\right\rbrace$;
and $\widehat{U}_{n}(\lambda) = \mathcal{N}\left\lbrace \widehat{A}_{n}(\lambda)\setminus
  A_0\right\rbrace$.
Thus, the false selection rate at $\lambda$ is
$p_n(\lambda) = E\left[ \widehat{U}_{n}(\lambda)/\max\left\lbrace
\widehat{I}_{n}(\lambda)+\widehat{U}_{n}(\lambda), 1
\right\rbrace\right]$.

\subsection{Estimating the false selection rate}
In this section, we provide a description of our estimator of the
false selection rate for each model along the Lasso solution path and
provide theoretical justification;
 details of the implementation are deferred to the subsequent
section.  The proposed estimator is constructed in three stages: (S1)
apply screening to form a preliminary estimator of the set of nonzero
coefficients, $A_0$; (S2) generate pseudo-variables
that mimic the unimportant variables, i.e., those in $A_0^c$;
and (S3) apply the Lasso to a dataset
composed of the selected variables from the screening step and the
generated pseudo-variables; the proportion of pseudo-variables in the active
set, $\widehat{A}_{n}(\lambda)$, is the estimated false selection rate
at tuning parameter value $\lambda$.

Let $r = rank(\mathbf{\mathbb{X}})$ and for any square matrix, $U$, write
$U^{-}$ to denote a pseudo-inverse.   For any non-empty
subset $S$ of $\left\lbrace 1,\ldots, p \right\rbrace$, define
$Q_{11}(S) = n^{-1}\mathbb{X}_{S}^{\T}\mathbb{X}_{S}$,
$Q_{12}(S) = n^{-1}\mathbb{X}_{S}^{\T}\mathbb{X}_{S^c}$,
$Q_{21}(S) = n^{-1}\mathbb{X}_{S^c}^{\T}\mathbb{X}_{S}$,
 $Q_{22}(S) = n^{-1}\mathbb{X}_{S^c}^{\T}\mathbb{X}_{S^c}$,
and $P_{\mathbb{X}_{S}} =
\mathbb{X}_{S}(\mathbb{X}_{S}^{\T}\mathbb{X}_{S})^{-}\mathbb{X}_{S}^{\T}$.
The estimator $\widehat{p}_{n}(\lambda)$
of $p_n(\lambda)$ is constructed as follows.
\begin{itemize}
\item[] Step 1 (Screening): For the full data
  $(\mathbb{X}, \mathbb{Y})$, apply a viable variable selection
  method to construct a preliminary estimator, $\widehat{A}_{0,n}$, of
  the set of nonzero coefficients $A_0$. Let $\widehat{r}_{0}$ denote
  the rank of $\mathbb{X}_{\widehat{A}_{0,n}}$.
\item[] Step 2 (Pseudo-variable generation): Let
  $\Omega(\widehat{A}_{0,n}) \in \mathbb{R}^{(r-\widehat{r}_{0})\times
    \left\lbrace p - \mathcal{N}(\widehat{A}_{0,n})\right\rbrace}$
satisfy
\begin{equation*}
 \Omega(\widehat{A}_{0,n})^{\T}\Omega(\widehat{A}_{0,n})
= Q_{22}(\widehat{A}_{0,n}) -
  Q_{21}(\widehat{A}_{0,n})Q_{11}^{-}(\widehat{A}_{0,n})
  Q_{12}(\widehat{A}_{0,n}), 
\end{equation*}
and
let $V(\widehat{A}_{0,n}) \in \mathbb{R}^{n\times (r-\widehat{r}_{0})}$ be any
orthonormal matrix that is orthogonal to the column
space of  $\mathbb{X}_{\widehat{A}_{0,n}}$.
Pseudo-variables have the form
    \begin{align}
    \label{eq:phony}
    \mathbb{X}_{\mathrm{pseudo}} = P_{\mathbb{X}_{\widehat{A}_{0,n}}}
          \mathbb{X}_{\widehat{A}_{0,n}^c} + \sqrt{n}
          V(\widehat{A}_{0,n})\Omega(\widehat{A}_{0,n}).
    \end{align}
In Section \ref{sec:comput}, we describe how to calculate
$\Omega(\widehat{A}_{0,n})$ and generate $V(\widehat{A}_{0,n})$
randomly thereby allowing for generating replicate random pseudo-variables.
\end{itemize}


\begin{itemize}

      \item[] Step 3 (Error rate estimation): Fit the Lasso estimator using
        $\mathbb{X}_{\mathrm{new}} = (\mathbb{X}_{\widehat{A}_{0,n}},
        \mathbb{X}_{\mathrm{pseudo}})$,
        and calculate
        $\widehat{A}_{n}^{\mathrm{new}}(\lambda) = \{j:
        \widehat{\beta}_{n, j}(\lambda; \mathbb{Y},
        \mathbb{X}_{\mathrm{new}}) \neq 0 \}$, and subsequently
    \begin{align}
    \widehat{p}_{n}(\lambda) =
          \frac{\widehat{U}_{n}^{\mathrm{new}}(\lambda)
          }
          {
          \max\left\lbrace
          \widehat{I}_{n}^{\mathrm{new}}(\lambda) +
          \widehat{U}_{n}^{\mathrm{new}}(\lambda), 1
          \right\rbrace
          }
          = \frac{
          \mathcal{N}\left\lbrace
          \widehat{A}_{n}^{\mathrm{new}}(\lambda)
          \setminus \widehat{A}_{0,n}
          \right\rbrace
          }
          {\max\left[
          \mathcal{N}\left\lbrace \widehat{A}_{0,n}^{\mathrm{new}}(\lambda)
          \right\rbrace, 1
          \right]
          },
    \end{align}
        where $\widehat{I}_{n}^{\mathrm{new}}(\lambda) =
        \mathcal{N} \left\lbrace \widehat{A}_{n}^{\mathrm{new}}(\lambda) \bigcap
        \widehat{A}_{0,n}\right\rbrace$ and
        $\widehat{U}_{n}^{\mathrm{new}}(\lambda) =
         \mathcal{N}\left\lbrace
          \widehat{A}_{n}^{\mathrm{new}}(\lambda)
          \setminus \widehat{A}_{0,n}
          \right\rbrace$.
\end{itemize}

To stabilize our estimator, we repeat the above steps $B$ times to
obtain the estimators $\widehat{p}_{n}^{(1)}(\lambda),\ldots,
\widehat{p}_{n}^{(B)}(\lambda)$ and subsequently compute
$\overline{\widehat{p}}_{n}(\lambda) =
B^{-1}\sum_{b=1}^{B}\widehat{p}_{n}^{(b)}(\lambda)$. The following results
are proved in the Appendix; throughout we
implicitly assume that all requisite moments exist and are finite.
\begin{restatable}{lemma}{sameCov}
    \label{lem:sameCov}
    Suppose $\widehat{A}_{0,n} = A_{0}$ with probability one, then
        $n^{-1}(\mathbb{X}_{\widehat{A}_{0,n}}, \mathbb{X}_{\mathrm{pseudo}})^\T
        (\mathbb{X}_{\widehat{A}_{0,n}},
                \mathbb{X}_{\mathrm{pseudo}}) =
                \boldsymbol{\Sigma}.$
                Furthermore, $\left\lbrace \widehat{I}_{n}(\lambda),
                  \widehat{U}_{n}(\lambda)\right\rbrace$
                and $\left\lbrace
                  \widehat{I}_{n}^{\mathrm{new}}(\lambda),
                  \widehat{U}_{n}^{\mathrm{new}}(\lambda)
                  \right\rbrace$ are equal in distribution.
\end{restatable}

The preceding result shows that were the set of important variables, $A_0$,
known, substituting the pseudo-variables for the unimportant
variables, $A_0^c$, does not affect the false selection rate.
Of course, $A_0$ is not known in practice;
the following result shows that preceding result holds provided that the
initial screening procedure is selection consistent.

\begin{restatable}{theorem}{errEst}
\label{thm:errEst}
Assume that  $\widehat{A}_{0,n}\rightarrow A_0$ with probability
one.  Then, for any bounded function
$g:\mathbb{R}^{2}\rightarrow\mathbb{R}$, it follows that
$$\sup_{\lambda}  \left |E \left[ g \left\lbrace \widehat{I}_{n}(\lambda),
      \widehat{U}_{n}(\lambda) \right\rbrace
 \right]
-  E \left[ g \left\lbrace \widehat{I}_{n}^{\mathrm{new}}(\lambda),
    \widehat{U}_{n}^{\mathrm{new}}(\lambda) \right\rbrace
\right]
\right |= o(1). $$
\end{restatable}


\begin{restatable}{corollary}{disConv}
\label{cor:disConv}
Assume that  $\widehat{A}_{0,n}\rightarrow A_0$ with probability
one.  Then, for any bounded function
$g:\mathbb{R}^{2}\rightarrow\mathbb{R}$, it follows that
$$\sup_{\lambda, t}  \left |P \left[ g \left\lbrace \widehat{I}_{n}(\lambda),
      \widehat{U}_{n}(\lambda)  \right\rbrace \leq t
 \right]
-  P \left[ g \left\lbrace \widehat{I}_{n}^{\mathrm{new}}(\lambda),
    \widehat{U}_{n}^{\mathrm{new}}(\lambda)
 \right\rbrace  \leq t
\right]
\right |= o(1). $$
\end{restatable}

\begin{restatable}{corollary}{MultierrEst}
\label{cor:MultierrEst}
Assume $\widehat{A}_{0,n}\rightarrow  A_{0}$ with probability one.
Then, setting $g(v,w) = v/\max(v+w, 1)$ shows
$$\sup_{\lambda}  \left | p_n(\lambda) -
B^{-1} \sum_{b=1}^{B} E \left\{ \widehat{p}^b_n(\lambda)  \right\}\right |= o(1). $$
\end{restatable}



The preceding results require a selection-consistent screening
procedure; we provide such a selection procedure
based on pseudo-variables in the Supplemental Materials.
While the theoretical assumption of selection consistency might be
still strong, empirical results in the next section suggest that
screening based on Lasso tuned by 10-fold cross validation (which is
not selection consistent) leads
to satisfactory results.

In small samples, we have found that the empirical performance of our procedure
can be improved by augmenting $\mathbb{X}_{\mathrm{new}}$ with a permutation of
$\mathbb{X}_{\widehat{A}_{0,n}}$, say $\mathbb{X}_{\mathrm{perm}} =
(\mathbb{X}_{\widehat{A}_{0,n}}, \mathbb{X}_{\mathrm{pseudo}},
G\mathbb{X}_{\widehat{A}_{0,n}})$ where $G$ is a random permutation matrix. The
intuition for adding this permutation is to compensate for over-estimation of
$A_0$ in finite samples (which in turn leads to underestimation of the false
selection rate).
  See Proposition 1.1
in the Supplemental Materials for an analog of Corollary \ref{cor:MultierrEst}
for this modified procedure.

Control of the FSR at specified error rate $\alpha$ is achieved by
first estimating the FSRs for a sequence of tuning parameter
values,
$\lambda_{(1)}, \ldots, \lambda_{(m)}$ and then selecting the tuning
parameter
$\widehat{\lambda} = \min \{\lambda_{(i)}: \widehat{p}(\lambda_{(i)}) \leq
\alpha \}$.
The final model is obtained by fitting the Lasso using $(\mathbb{X},
\mathbb{Y})$ at tuning parameter $\widehat{\lambda}$.

\subsection{Computation of pseudo-variables}
\label{sec:comput}
Our procedure for generating pseudo-variables is based
on the following result which is proved in the Supplemental
Materials.
\begin{restatable}{lemma}{sol}
    \label{lem:sol}
        For any non-empty subset $S$ of
        $\left\lbrace 1,\ldots, p\right\rbrace$ such that
        $\mathbb{X}_{S}$ has rank $r(S)$, denote $s = \mathcal{N}(S)$. Then
        $\left\lbrace \mathbb{X}_{S},
          \mathbb{X}_{\mathrm{pseudo}}(S)\right\rbrace^{\T} \left\lbrace
          \mathbb{X}_{S}, \mathbb{X}_{\mathrm{pseudo}}(S)\right\rbrace
        = \left( \mathbb{X}_{S}, \mathbb{X}_{S^c}\right)^{\T}
        \left(\mathbb{X}_{S}, \mathbb{X}_{S^c}\right)$
        if and only if
        $\mathbb{X}_{\mathrm{pseudo}}(S) =
        P_{\mathbb{X}_{S}}\mathbb{X}_{S^c} +
        \sqrt{n}V(S)\Omega(\mathbb{X}, S)$
        for some $V(S)$ and $\Omega(\mathbb{X}, S)$, where $V(S)$ is an
        $n \times \left\lbrace r - r(S)\right\rbrace$ orthonormal
        matrix that is orthogonal to
        $\mathbb{X}_{S}$, and $\Omega(\mathbb{X}, S)$ is an
        $\left\lbrace r - r(S)\right\rbrace \times \left\lbrace p - s \right\rbrace$ matrix such that
        $\Omega(\mathbb{X}, S)^\T \Omega(\mathbb{X}, S)= Q_{22}(S) -
        Q_{21}(S)Q_{11}^{-}(S) Q_{12}(S)$.
\end{restatable}
\noindent
Thus, the preceding result characterizes a class of potential
pseudo-variables indexed by the matrices $V(S)$ and
$\Omega(\mathbb{X}, S)$.

To generate pseudo-variables, the first part
$P_{\mathbb{X}_S} \mathbb{X}_{S^c}$ is calculated directly using a
QR decomposition.
The second part,
$V(S)$, is constructed using the form
$V(S) = V_1 V_2,$ where $V_1$ is an
$n \times \left\lbrace n - r(S)\right\rbrace$ orthonormal matrix which is orthogonal to
 $\mathbb{X}_{S}$, and $V_2$ is a random orthonormal matrix.
To find $V_1$, we compute the QR decomposition $\mathbb{X}_{S} = Q_xR_x$ and
then choose $V_1$ to be the last $n - r(S)$ columns of $Q_x$, which are an
orthonormal basis for the null space of $\mathbb{X}_{S}^T$.  Subsequently,
$V_2$ is a random orthonormal matrix distributed with Haar measure
\cite[][]{mezzadri2006generate}.

To find $\Omega$ such that
$\Omega(\mathbb{X}, S)^{\T}\Omega(\mathbb{X}, S) =
Q_{22}(S) -
Q_{21}(S)Q_{11}^{-}(S)
Q_{12}(S)$,
it is not necessary to compute $Q_{22}, Q_{21}, Q_{11}$, which is
computationally expensive for $p \gg n$. To see this, define
$\mathcal{E}_{1|2} = (I - P_{\mathbb{X}_{S}})
\mathbb{X}_{S^c}$
so that
$\mathcal{E}_{1|2}^\T\mathcal{E}_{1|2} =
(Q_{22}-Q_{21}Q_{11}^{-}Q_{12})$,
then compute the QR decomposition
$\mathcal{E}_{1|2} = Q_\mathcal{E}R_\mathcal{E}$ and choose
$\Omega$ to be the first $r - r(S)$ rows of $R_\mathcal{E}$.

\section{Simulations}
\label{sec:sim}
We examine the finite-sample performance of the proposed method in
terms of FSR control and true selection rate (TSR) across
data sets with varying dimension, number of nonzero coefficients,
signal strength,   and  correlation structure.   Our examination is
based
on the comparison of the
following methods:
pseudo-1, the proposed variable addition method with the screening
procedure given in the Supplemental Materials; pseudo-2, the proposed
variable addition method with the screening done using the Lasso tuned
with 10-fold cross-validation;
Knockoff and Knockoff+ \citep[][]{barber2015controlling};
and pseudo-Wu, a variable-addition method proposed to control
FSR in forward selection \citep[][]{wu2007controlling}.


In implementing our proposed methods we included the permutation term as
discussed in Section 2; results without the permutation are presented in the
Supplemental Materials. In our implementation of the proposed pseudo-variable
methods, we repeated pseudo-variable generation $B=20$ times in each iteration.
The knockoff and knockoff+ methods, are as implemented in the R package
\emph{knockoff} with default parameter settings. The pseudo-Wu is as
implemented on the authors' website. Their implementation requires $n > 2p$ so
that $p$-values for all variables can be
calculated. 
As suggested by the authors, we use a bootstrap size $B = 200$ for the
pseudo-Wu method.

The data are generated from the linear model
$Y_i = X_i^T\boldsymbol{\beta}_0 + \epsilon_i,\ i = 1, \ldots n, $
where $\epsilon_i \sim N(0, 1)$ and $X_i \sim ~ N(0_{p \times 1}, C)$
with $C_{ij} = \rho^{|i-j|}$.  Define
$\boldsymbol{\beta}_0 = (A, A, 0, 0, A, 0, \ldots, 0, 0)^t$, where $A$
is the signal amplitude and positions of nonzero coefficient are sampled
without replacement from $\{1, \ldots, p\}$. Denote the number of
nonzero coefficients by $s$.

We study four different factors thought to influence FSR/TSR
estimation:
    \begin{itemize}
        \item[1.] predictor dimension: we fix $n = 200,\ \rho = 0.5,\ A = 1,\ s = 5$,
            and vary $p = 30$, $70$, $110$, $150,\ 190,\ 230,\ 330,\ 430,\ 530$;
        \item[2.] correlation magnitudes: we fix $n = 200,\ p = 50,\ A = 1,\ s = 5$,
            and vary $\rho = 0,\ 0.1$, $0.2$, $\ldots, 0.8, \ 0.9$;
        \item[3.] signal amplitude: we fix $n = 200,\ p = 50,\ \rho = 0.5,\ s = 5$,
            and vary $A = 0.1,\ 0.2,\ 0.3,\  \ldots,\ 1$;
        \item[4.] number of nonzero coefficients: we fix $n = 200,\ p
                  = 50,\ \rho = 0.5,\ A = 1$,
            and vary $s = 1,\ 2,\ 3,\  \ldots,\ 20$.
    \end{itemize}

For each of the above combinations of parameter values, we first generate
twenty different  values of $\boldsymbol{\beta}_0$
and then for each $
\boldsymbol{\beta}_0$, we generate $50$ datasets. Thus, for each
combination of parameter values, we generate 1,000 replicates. Results
are based on the average across these replicates.
In all settings, we set $\alpha=0.20$, to be the target false
selection rate;
 results for additional values of
$\alpha$ are given in the Supplemental Materials.


Simulation performance across the different parameter settings are displayed in
Figures \ref{fig:dim}-\ref{fig:amp}. Standard errors of the FSR and TSR
averages are less than 0.015 for all simulation settings. It can be seen that
the proposed methods are consistently less biased than alternatives in terms of
FSR. For TSR, knockoff and knockoff+ decrease rapidly as $p \rightarrow n$.
Furthermore, when number of true signal is sparse, i.e., when $s$ is small,
knockoff+ and Pseudo-Wu have low power.

Both of the proposed pseudo-variable methods performed favorably relative to
competing methods.  That pseudo-2 performed well is encouraging as it does not
satisfy the selection consistency criterion required in our theoretical
results, suggesting that the proposed pseudo-variable methods are robust to
mild violations of this assumption. Furthermore, in the Supplemental Materials,
we present application of the proposed methods to a Cox proportional hazards
model as well as a logistic regression model; these simulations are
qualitatively similar to those presented here suggesting the proposed method
can be applied more generally than the linear model case for which our theory
was developed.

	\begin{figure}
		\begin{subfigure}{.5\textwidth}
			\centering
			\includegraphics[scale = 0.35]{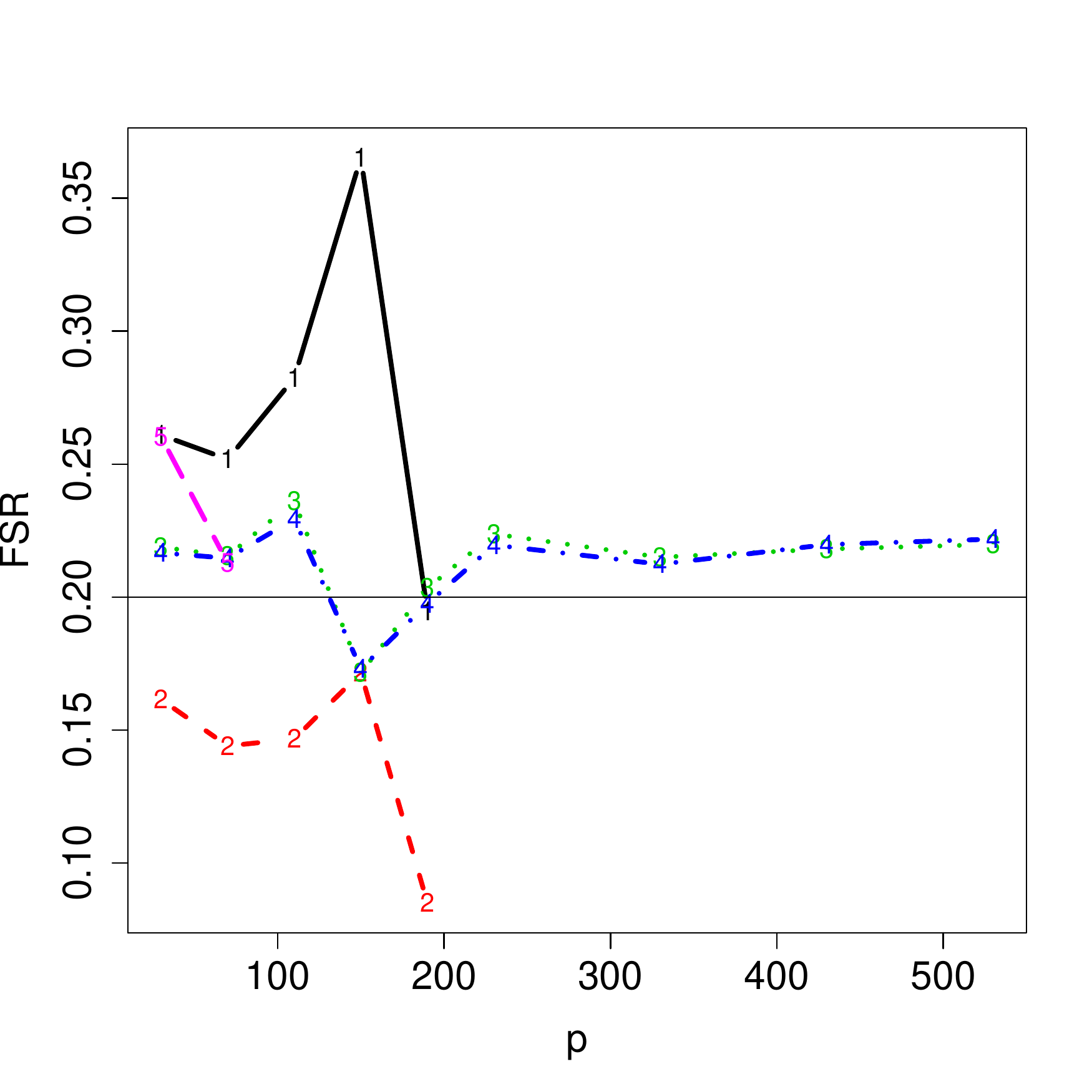}
			\caption{False selection rate vs $p$}
		\end{subfigure}%
		\begin{subfigure}{.5\textwidth}
			\centering
			\includegraphics[scale = 0.35]{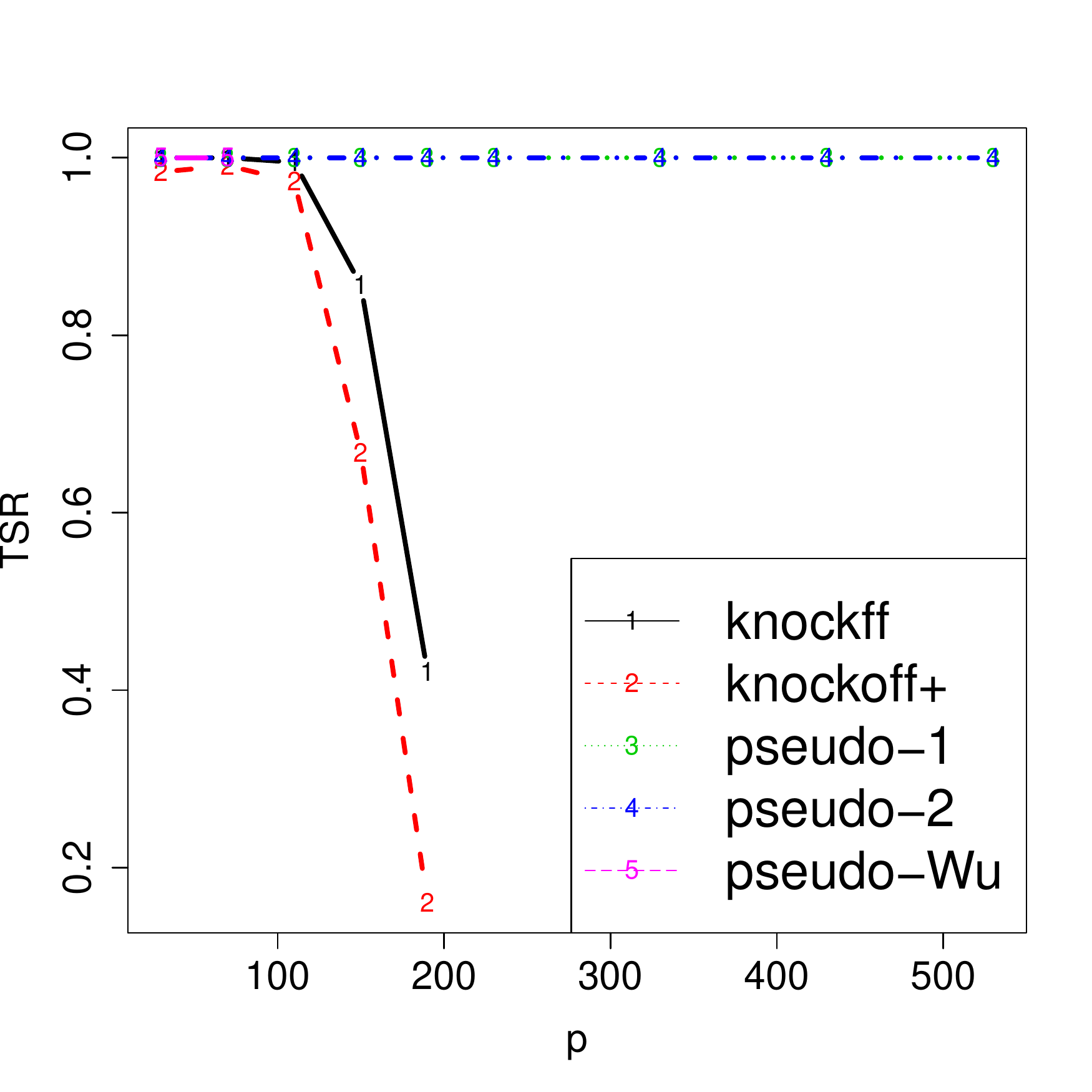}
			\caption{True selection rate vs $p$}
		\end{subfigure}
		\caption{Performances under different dimensions at
                  $\alpha = 0.2$. Left and right figure shows the average
                 FSR and TSR respectively. The Knockff and Knockoff methods require $n > p$ and the Wu's pseudo-variable method requires $n > 2p$. }
		\label{fig:dim}
	\end{figure}

	
	\begin{figure}
		\begin{subfigure}{.5\textwidth}
			\centering
			\includegraphics[scale = 0.35]{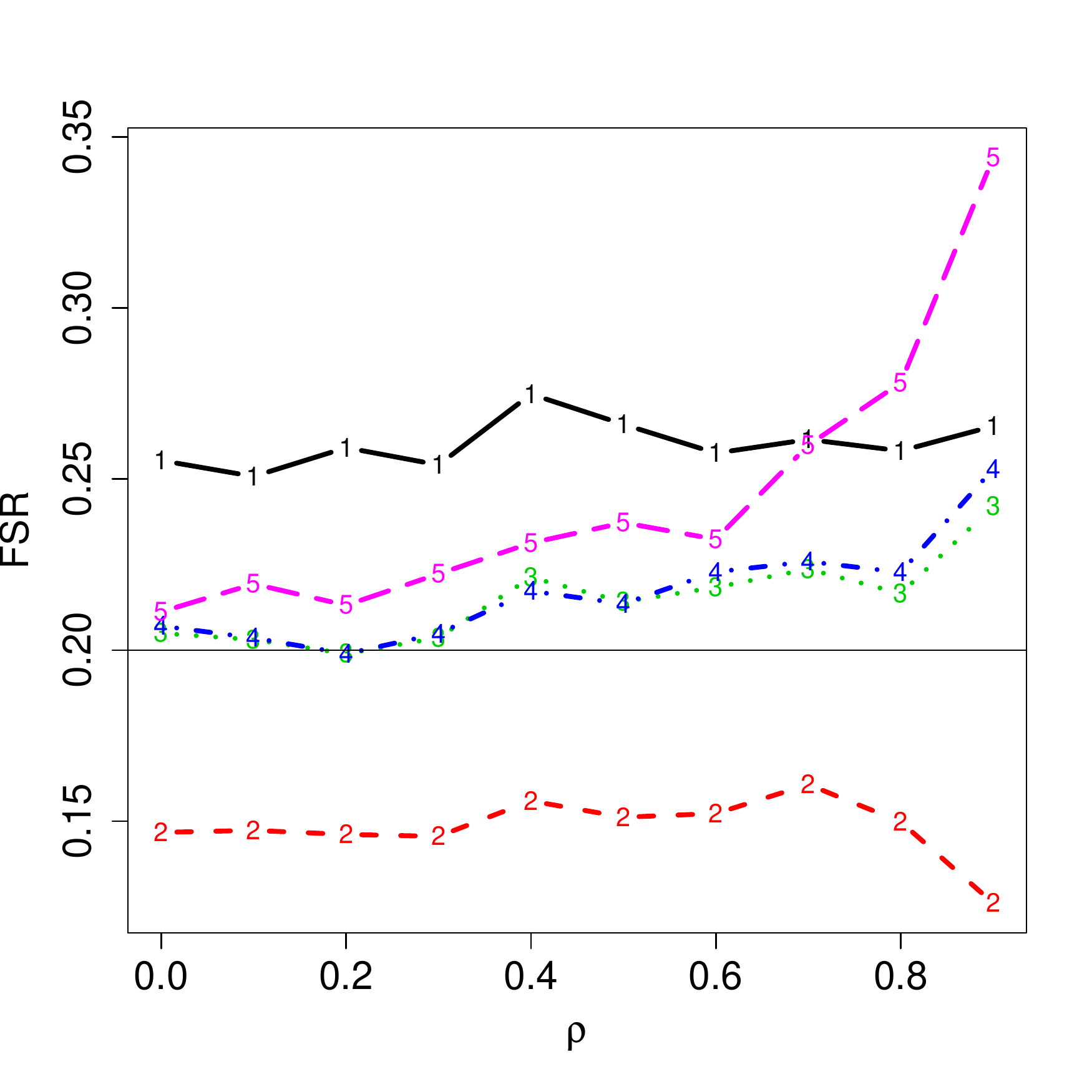}
			\caption{False selection rate vs $\rho$}
		\end{subfigure}%
		\begin{subfigure}{.5\textwidth}
			\centering
			\includegraphics[scale = 0.35]{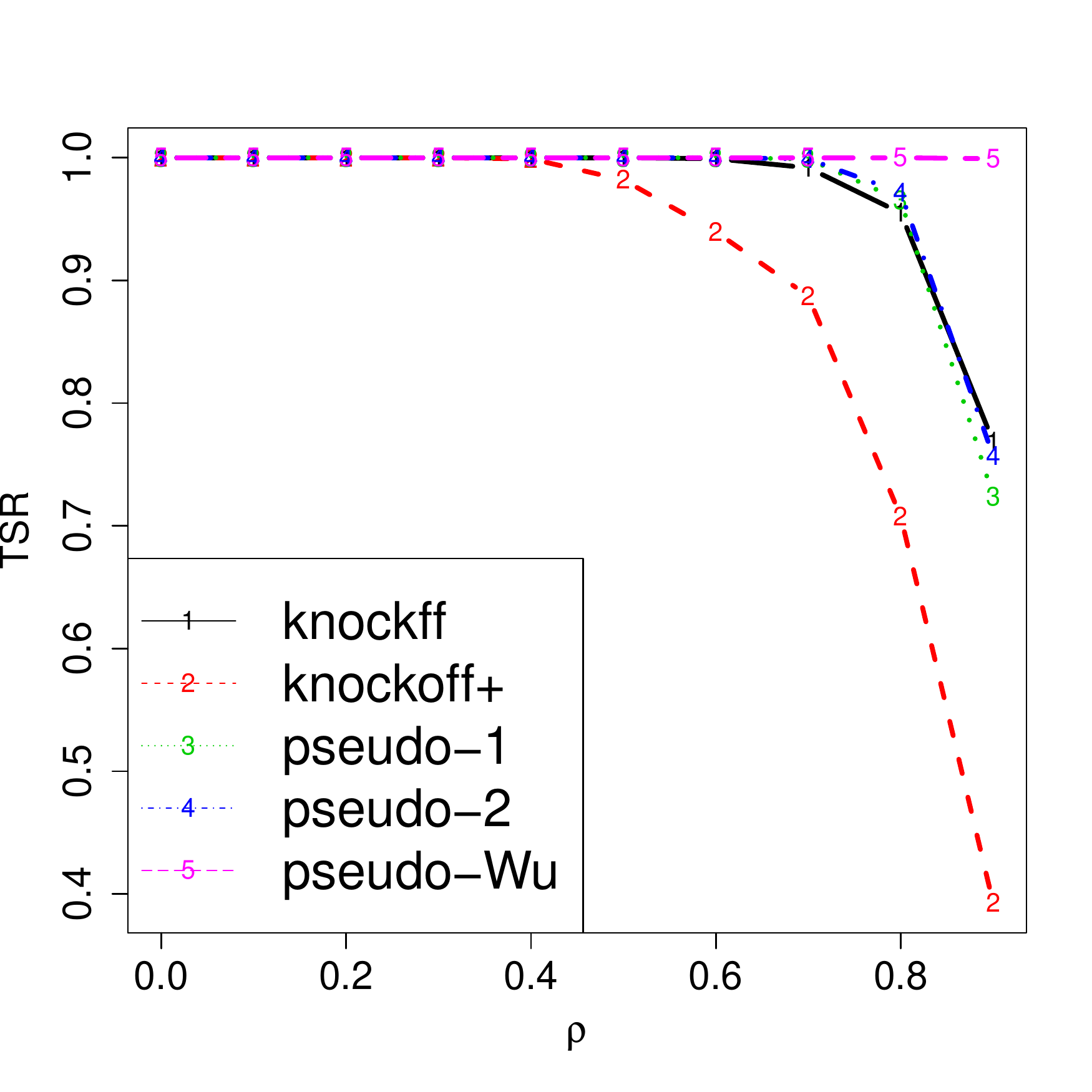}
			\caption{True selection rate vs $\rho$}
		\end{subfigure}
		\caption{Performances under different correlations at $\alpha = 0.2$. Left and right figure shows the average FSR and TSR respectively.}
		\label{fig:cor}
	\end{figure}

	\begin{figure}
		\begin{subfigure}{.5\textwidth}
			\centering
			\includegraphics[scale = 0.35]{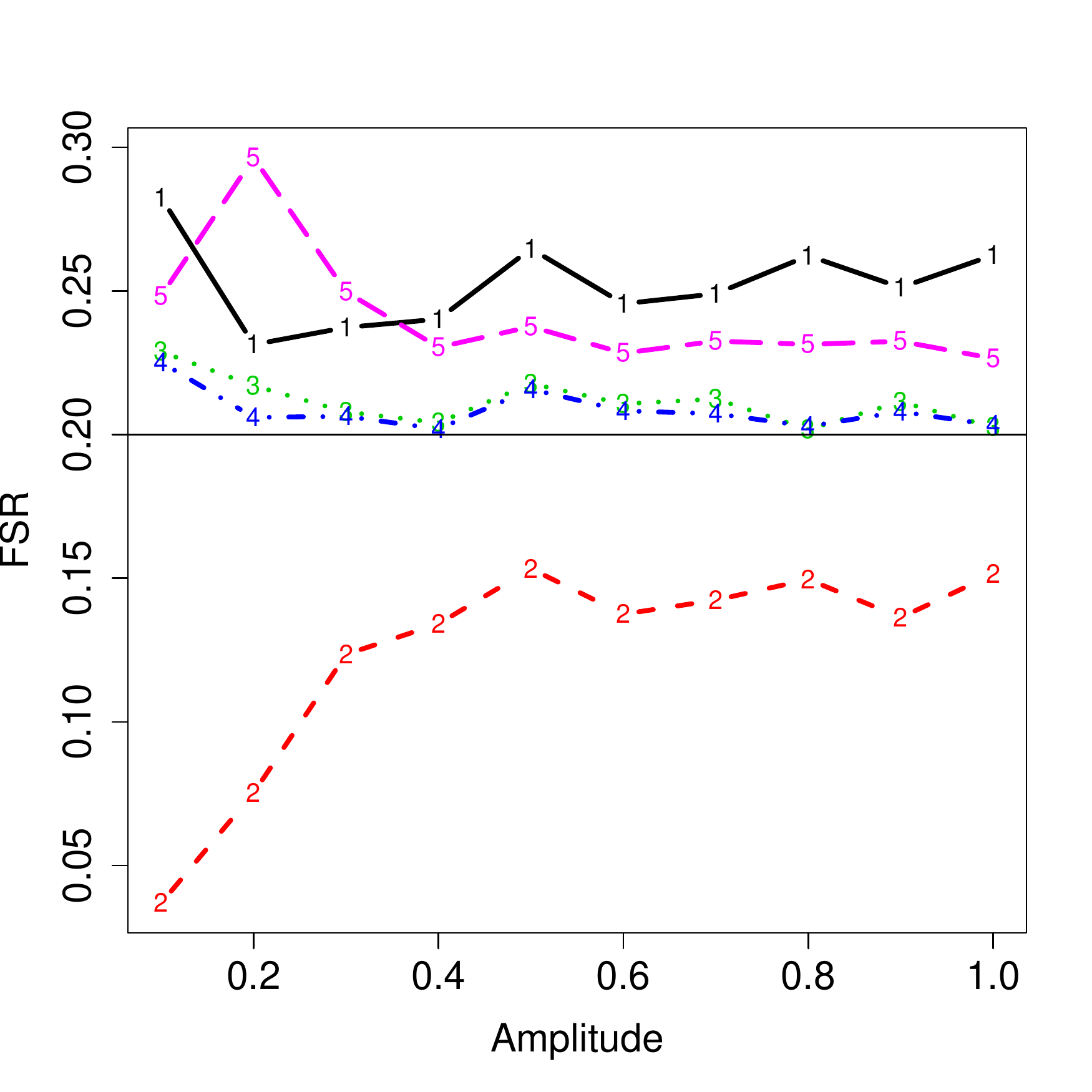}
			\caption{False selection rate vs $A$}
		\end{subfigure}%
		\begin{subfigure}{.5\textwidth}
			\centering
			\includegraphics[scale = 0.35]{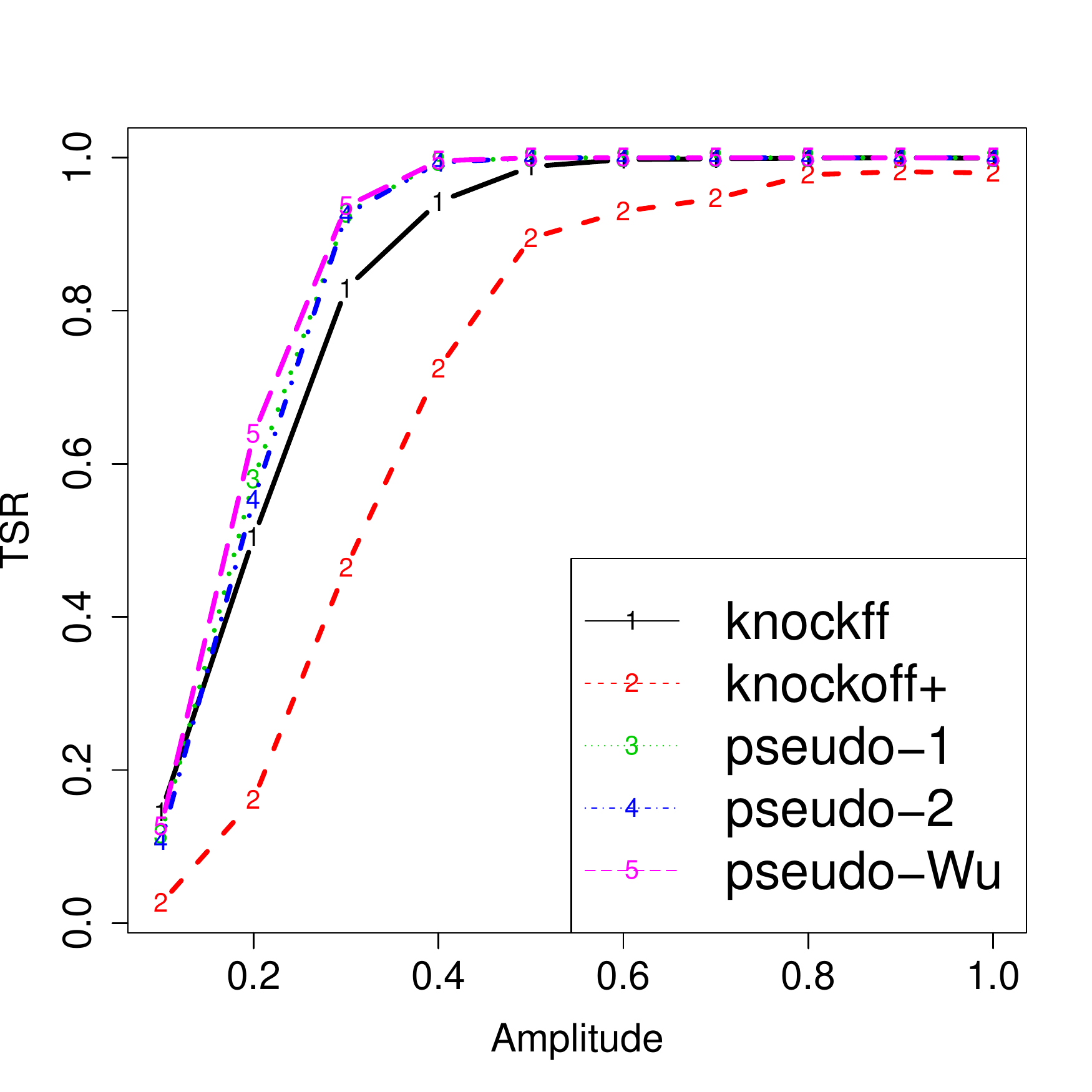}
			\caption{True selection rate vs $A$}
		\end{subfigure}
		\caption{Performances under different coefficient amplitude at $\alpha = 0.2$. Left and right figure shows the average FSR and TSR respectively.}
		\label{fig:amp}
	\end{figure}

	\begin{figure}
		\begin{subfigure}{.5\textwidth}
			\centering
			\includegraphics[scale = 0.35]{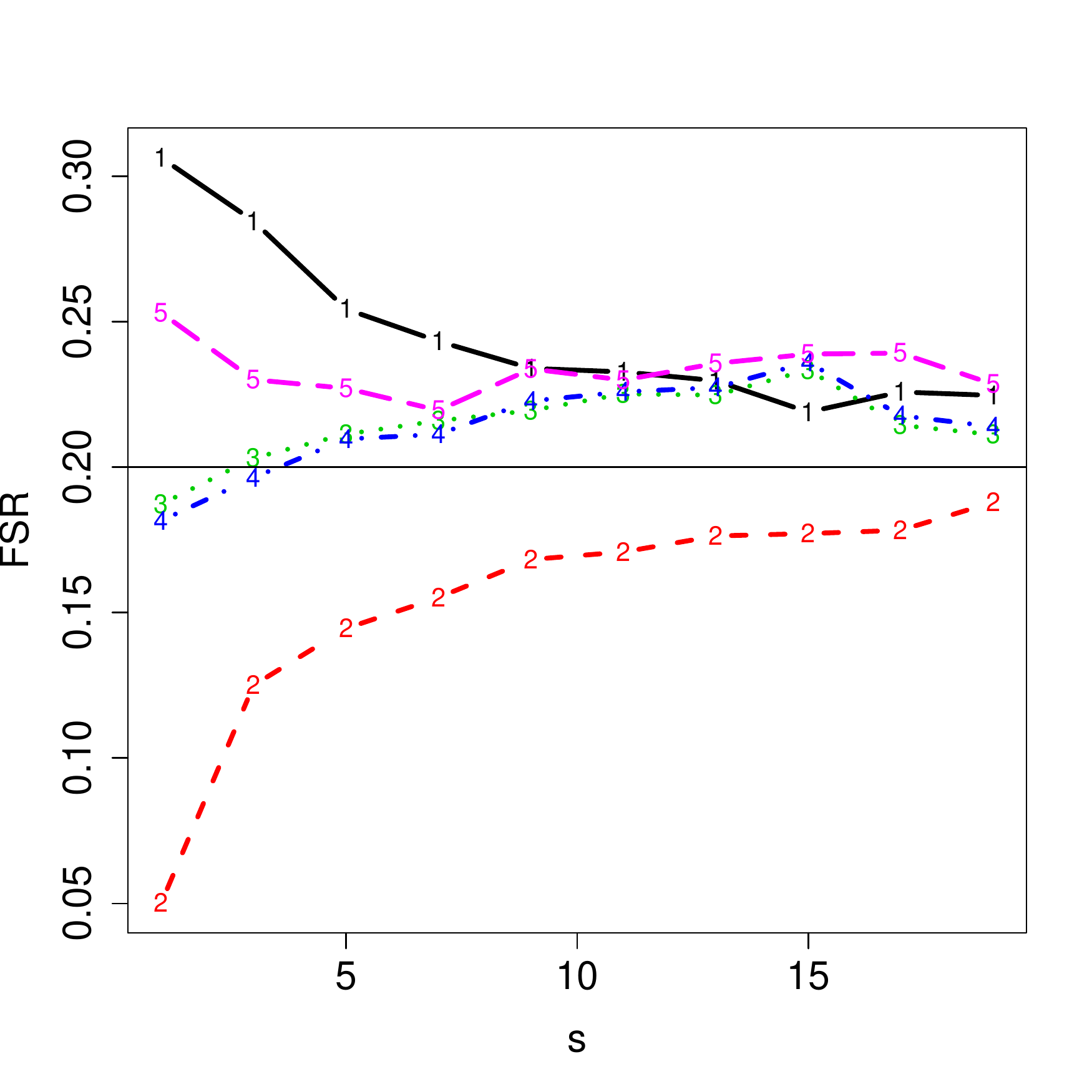}
			\caption{False selection rate vs $s$}
		\end{subfigure}%
		\begin{subfigure}{.5\textwidth}
			\centering
			\includegraphics[scale = 0.35]{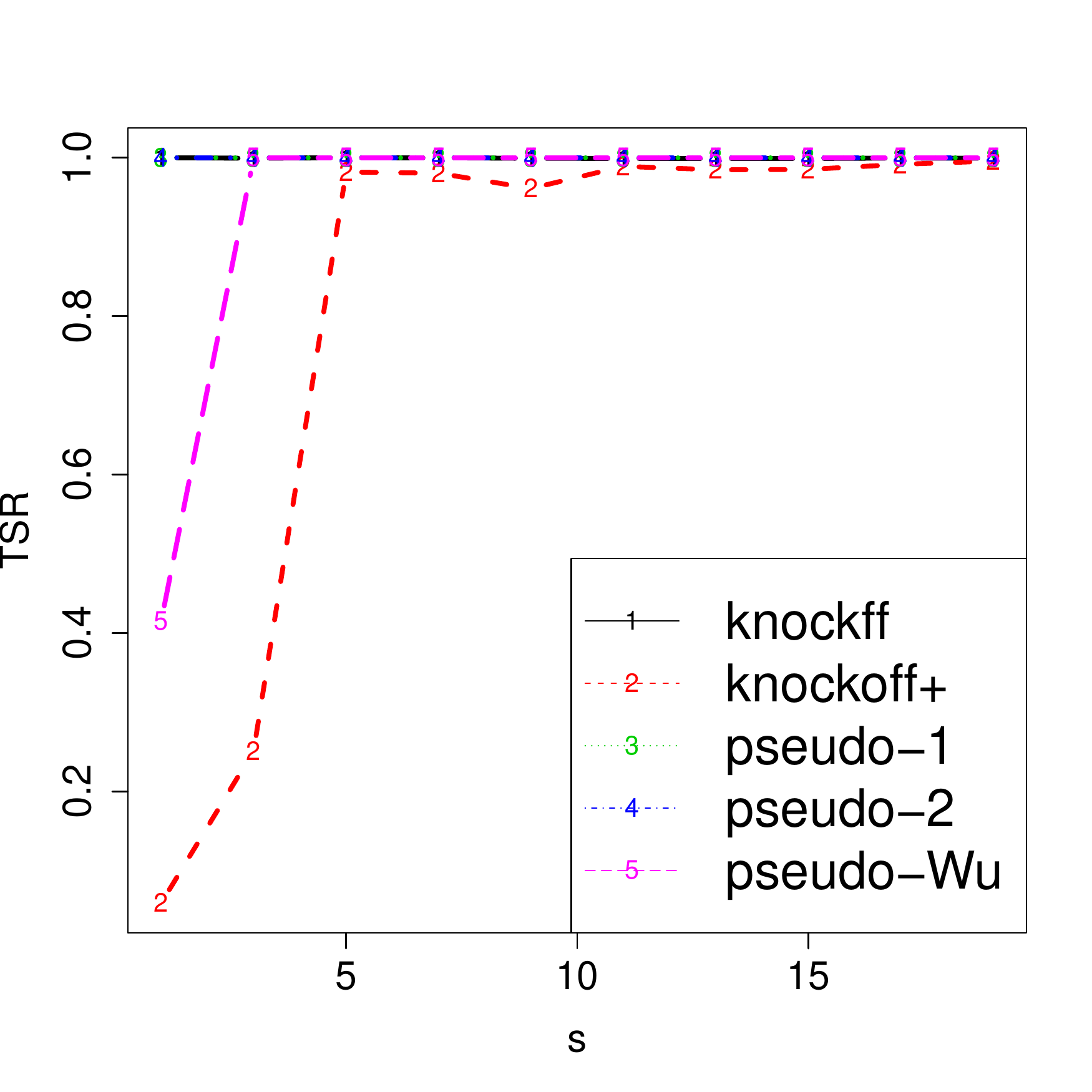}
			\caption{True selection rate vs $s$}
		\end{subfigure}
		\caption{Performances under different number of nonzero coefficients at $\alpha = 0.2$. Left and right figure shows the average FSR and TSR respectively.}
		\label{fig:spa}
	\end{figure}

\section{Illustrative examples}
\label{sec:real}
\subsection{Prostate cancer data}
Our first illustrative example uses data from a study of prostate specific
antigen (PSA) in $n=97$ prostate cancer patients \citep{stamey1989prostate}.
One of the goals of the study was to understand the relationship between PSA
and  eight biomarkers: log cancer volume (lcavol); log prostate weight
(lweight); age in years (age); log of the amount of benign prostatic
hyperplasia  (lbph); seminal vesicle invasion (svi); log of capsular
penetration (lcp); Gleason score (gleason); and percent of Gleason scores 4 or
5 (pgg45).

As in previous analyses, we fit a linear model for the regression of log PSA
(lpsa) on the preceding eight biomarkers.
We fit the model using the proposed pseudo-variable method with
screening done using the Lasso tuned using 10-fold cross validation,
and $B=100$ resamples. The Lasso solution
path is presented in Figure \ref{fig:prostate}. The vertical line on the left
corresponds to an estimated FSR of $\widehat{\alpha} = 0.1$; the associated model
has
three predictors:
lcavol, lweight, svi. The middle vertical line corresponds to an estimated
FSR of  $\widehat{\alpha} = 0.2$; the  associated model has four predictors:
lcavol, lweight, svi, pgg45.
The right vertical
line corresponds to an estimated FSR of $\widehat{\alpha}=0.3$; the
associated model has five predictors: lcavol, lweight, svi, pgg45, lbph.

\begin{figure}
    \centering
    \includegraphics[scale =0.5]{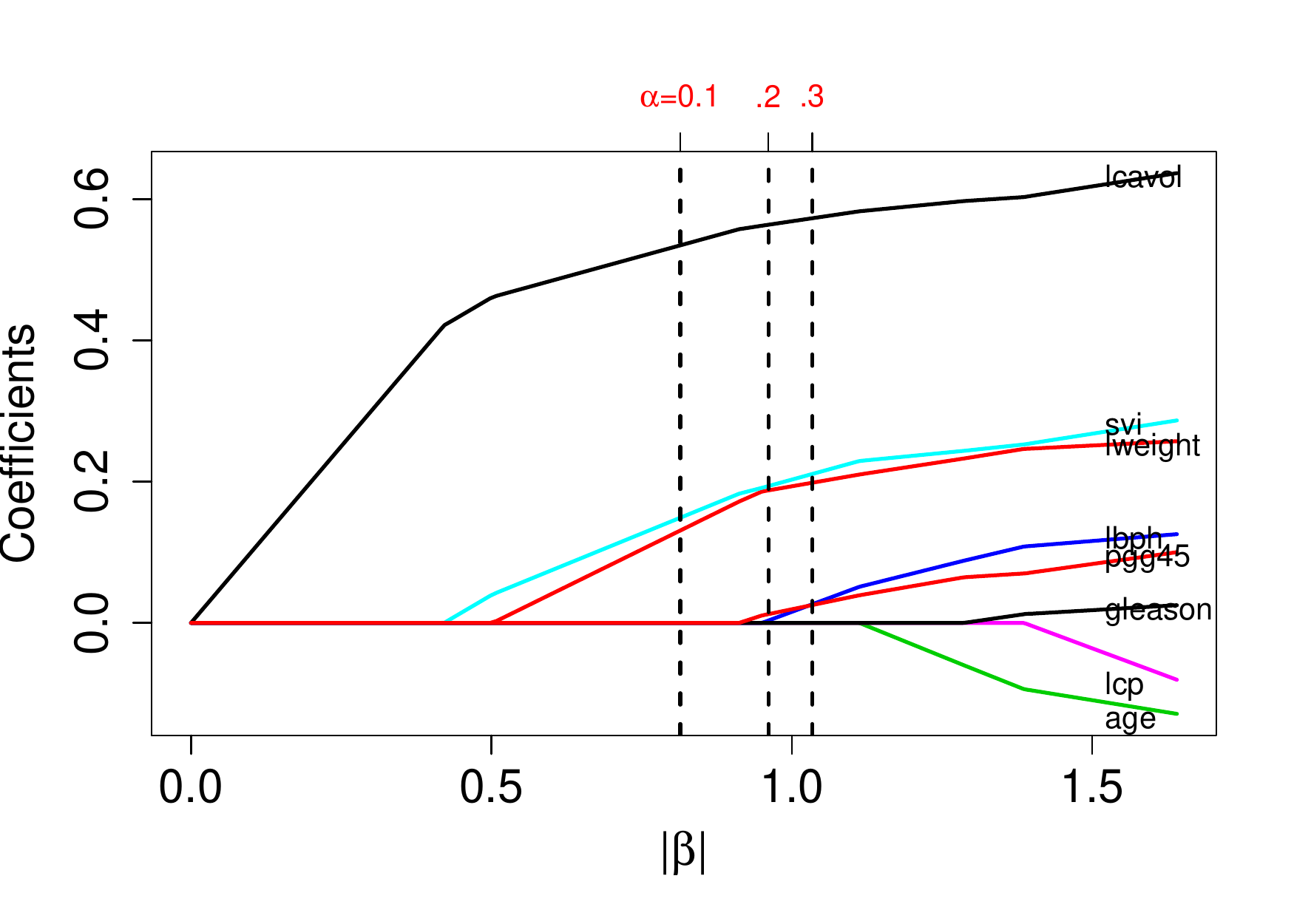}
    \caption{Lasso solution path for prostate cancer data.
          Vertical lines from left to right correspond to estimated
          FSRs of $\alpha=0.1$, $\alpha=0.2$, and $\alpha=0.3$.}
    \label{fig:prostate}
\end{figure}

\subsection{Leukemia cancer gene expression data}
Our second illustrative example used data from leukemia study
\citep[see][]{efron2012large}.  The primary outcome is binary
cancer type: acute myeloid leukemia (AML) or acute lymphoblastic
leukemia (ALL).  The goal is understand how gene expression data
relates to cancer type.  Expression levels are measured for
$p=7128$ genes on $n=72$ subjects.  Thus, this second example
demonstrates the use of the proposed method in the $p\gg n$ setting.


We fit a penalized logistic regression model for cancer type on gene
expression
levels. To estimate FSR, screening is done using the Lasso tuned using 10-fold cross validation, and we set $B = 100$. The Lasso solution
path with estimated FSRs is displayed in Figure
\ref{fig:prostate}. The vertical lines
on the figure, read from left to right, correspond to estimated FSRs
of $\widehat{\alpha} = 0.1, 0.2$ and $0.3$; it can be seen that these
correspond to four, six and eleven selected genes. The choice of
an appropriate model along this path should be dictated by the
costs associated with a false positive and other domain-specific
considerations.

\begin{figure}
    \centering
    \includegraphics[scale =0.5]{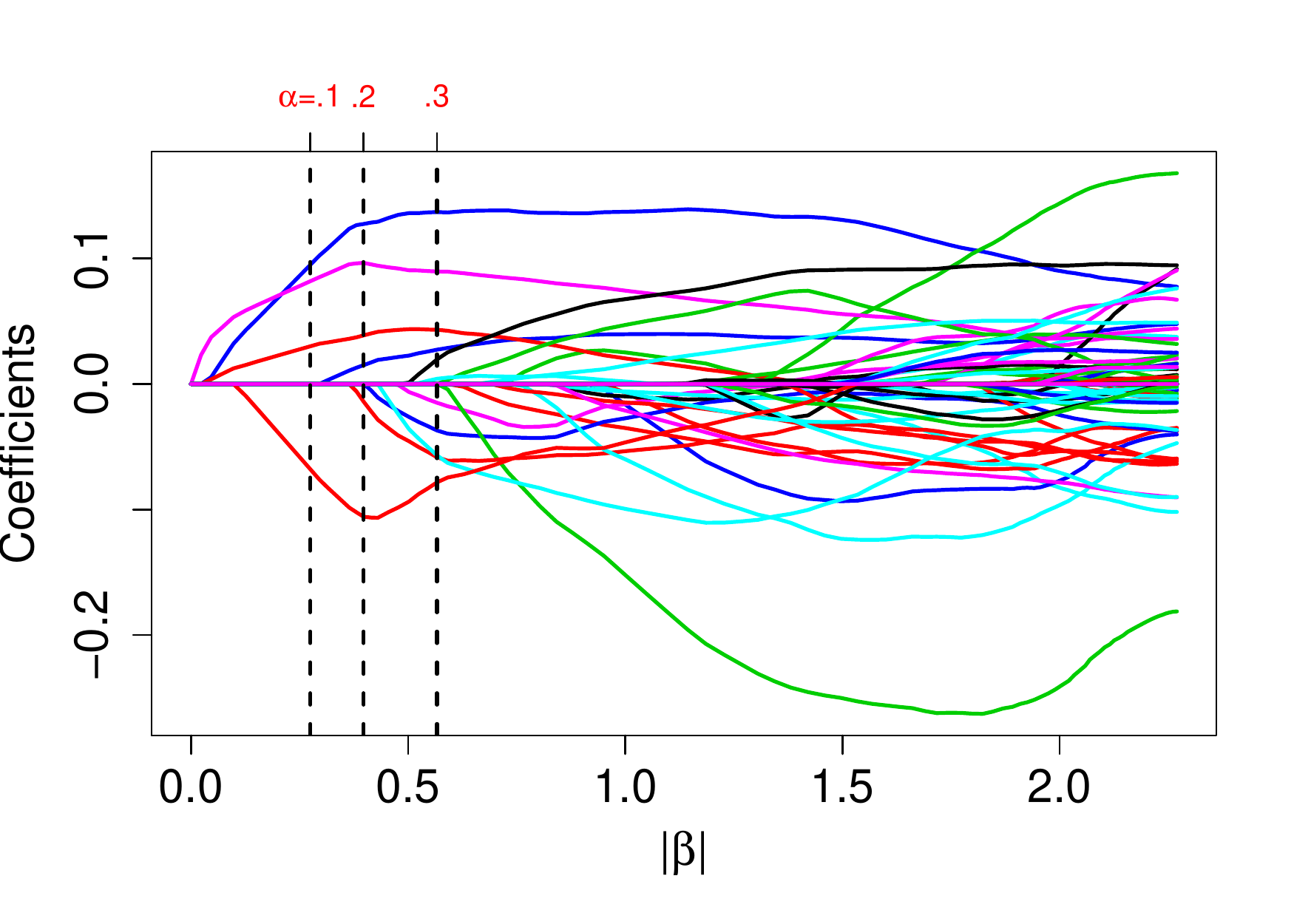}
    \caption{Lasso solution path for leukemia cancer gene expression data.
          Vertical lines from left to right correspond to estimated
          FSRs of $\alpha=0.1$, $\alpha=0.2$, and $\alpha=0.3$.}
    \label{fig:leukemia}
\end{figure}

\section{Conclusion}
\label{sec:conc}
We proposed a novel variable-addition method to estimate the FSR in penalized
regression.  The proposed method provides (asymptotically) unbiased
estimates of the FSR uniformly over the solution path even when $p \gg
n$.  The primary motivation for the proposed methodology was to
label the solution path with estimated operating characteristics that
are meaningful in a domain context.   While our focus was on
linear models with a fixed design, simulation results suggest broader
applicability.   Indeed, one of the appealing features of
variable-addition methods is that they can be applied (at least in
principle) to black-box models. Evaluation of the theoretical
properties of the proposed method to such models is a topic
for future research.

\newpage
\bigskip
\begin{center}
    {\large\bf SUPPLEMENTARY MATERIAL}
\end{center}

\begin{description}
    \item[Simulation results:] Additional simulation results are presented in the online supplement to this article.

    \item[Phony-variables algorithm for screening:] Details of pseudo-variables algorithm for screening with proof of selection consistency.

    \item[Proofs and technical details:] Detailed proofs are provided in the online supplement to this article.

    \item[R package:] A R package is provided in the online supplement to this article.

\end{description}

%
%
%
%

%

\bibliographystyle{apa}

\bibliography{pseudo}


\newpage
\section{Proof and Technical Details}
\sameCov*

\begin{proof}
        First for $n^{-1}(\mathbb{X}_{\widehat{A}_0})^\T\mathbb{X}_{\mathrm{pseudo}}$, we have  \begin{align*}
        & n^{-1}(\mathbb{X}_{\widehat{A}_0})^\T  \mathbb{X}_{\mathrm{pseudo}} \\
        & = n^{-1}(\mathbb{X}_{\widehat{A}_0})^\T\left\{  P_{\mathbb{X}_{\widehat{A}_0}} \mathbb{X}_{\widehat{A}_0^c} + \sqrt{n} V\Omega\right\} \\
        & = Q_{12}(A_0),
        \end{align*}
        where the last equality holds since $V$ is orthogonal to the column space of $\mathbb{X}_{\widehat{A}_0}$.

        Then,
        \begin{align*}
        & n^{-1}\mathbb{X}_{\mathrm{pseudo}} ^\T\mathbb{X}_{\mathrm{pseudo}}  \\
        & = n^{-1}\left\{  P_{\mathbb{X}_{\widehat{A}_0}} \mathbb{X}_{\widehat{A}_0^c} + \sqrt{n} V\Omega\right\}^\T \left\{  P_{\mathbb{X}_{\widehat{A}_0}} \mathbb{X}_{\widehat{A}_0^c} + \sqrt{n} V\Omega\right\} \\
        &= \left \{n^{-1}\mathbb{X}_{\widehat{A}_0^c}^\T P_{\mathbb{X}_{\widehat{A}_0}}P_{\mathbb{X}_{\widehat{A}_0}} \mathbb{X}_{\widehat{A}_0^c} \right \} + \left\{ \Omega^\T\Omega\right\} \\
        & = \left \{n^{-1}\mathbb{X}_{\widehat{A}_0^c}^\T \left [\mathbb{X}_{\widehat{A}_0} \{\mathbb{X}_{\widehat{A}_0}^\T\mathbb{X}_{\widehat{A}_0}\}^-\mathbb{X}_{\widehat{A}_0}^\T \right ] \mathbb{X}_{\widehat{A}_0^c}\right\} + \left\{ \Omega^\T\Omega\right\}\\
        &= Q_{21}Q_{11}^{-}Q_{12} + (Q_{22}-Q_{21}Q_{11}^{-}Q_{12}) \\
        &= Q_{22}(\widehat{A}_0).
        \end{align*}

    This completes the proof that
    \begin{align*}
    n^{-1}(\mathbb{X}_{\widehat{A}_0}, \mathbb{X}_{\mathrm{pseudo}})^\T (\mathbb{X}_{\widehat{A}_0}, \mathbb{X}_{\mathrm{pseudo}}) = \boldsymbol{\Sigma}.
    \end{align*}

    Then we know,
     $n^{-1/2}(\mathbb{X}_{\widehat{A}_0}, \mathbb{X}_{\mathrm{pseudo}})^T\boldsymbol{\epsilon}$ follows a multivariate normal distribution with mean $\mathbf{0}$ and covariance $\boldsymbol{\Sigma}$, which is the same as distribution of $n^{-1/2}(\mathbb{X}_{A_0}, \mathbb{X}_{{\widehat{A}_0}^c})\boldsymbol{\epsilon}$. Since the Lasso solution only depends on $\boldsymbol{\Sigma}$ and $(\mathbb{X}_{\widehat{A}_0}, \mathbb{X}_{\mathrm{pseudo}}^T\boldsymbol{\epsilon})$,  we have that $\left (\widehat{I}_n\left(\lambda\right), \widehat{U}_n(\lambda) \right)$  and $\left (\widehat{I}_n^{\mathrm{new}}\left(\lambda\right), \widehat{U}_n^{\mathrm{new}}(\lambda) \right)$ are identically distributed.
\end{proof}

\errEst*
\begin{proof}
    First, we have
    \begin{align}
    \label{prth1eq1}
    \sup_{\lambda}  \left |E \left\{ g \left(\widehat{I}_n^{\mathrm{new}}(\lambda), \widehat{U}_n^{\mathrm{new}}(\lambda) \right) \mid \widehat{A}_{0, n} = A_0\right\} -  E \left\{ g \left(\widehat{I}_n^{\mathrm{new}}(\lambda), \widehat{U}_n^{\mathrm{new}}(\lambda)\right) \right\} \right |= o(1).
    \end{align}
    This follows from
    \begin{align*}
    &  E \left\{ g \left(\widehat{I}_n^{\mathrm{new}}(\lambda), \widehat{U}_n^{\mathrm{new}}(\lambda)\right) \right\} \\
    &= E \left\{ g \left(\widehat{I}_n^{\mathrm{new}}(\lambda), \widehat{U}_n^{\mathrm{new}}(\lambda) \right) \mid \widehat{A}_{0, n} = A_0\right\} P(\widehat{A}_{0, n} = A_0) \\ & +E \left\{ g \left(\widehat{I}_n^{\mathrm{new}}(\lambda), \widehat{U}_n^{\mathrm{new}}(\lambda) \right) \mid \widehat{A}_{0, n} \neq  A_0\right\}P(\widehat{A}_{0, n} \neq A_0),
    \end{align*}
    and $P(\widehat{A}_{0, n} = A_0)  = 1 -o(1)$.

    Then, suppose modifying screening step to be deterministic, one constructs pseudo-variables always using the true active set, i.e., $\mathbb{X}_{\mathrm{pseudo}} =  P_{\mathbb{X}_{A_0}}\mathbb{X}_{A_0^c} + \sqrt{n} V(A_0) \Omega(A_0)$. Denote the corresponding number of important variables at $\lambda$ as $\widehat{I}_n^{1}(\lambda)$, the number of unimportant variables at $\lambda$ as $\widehat{U}_n^{1}(\lambda)$. Then we have
    \begin{align}
    \label{prth1eq2}
    \sup_{\lambda}  \left |E \left\{ g \left(\widehat{I}_n^{\mathrm{new}}(\lambda), \widehat{U}_n^{\mathrm{new}}(\lambda) \right) \mid \widehat{A}_{0, n} = A_0\right\}  -  E \left\{ g \left(\widehat{I}_n^{1}(\lambda), \widehat{U}_n^{1}(\lambda) \right) \right\}\right |= o(1).
    \end{align}
    And by Lemma \ref{lem:sameCov}
    \begin{align}
    \label{prth1eq3}
    \sup_{\lambda}  \left |E \left\{ g \left(\widehat{I}_n(\lambda), \widehat{U}_n(\lambda) \right) \right\} -  E \left\{ g \left(\widehat{I}_n^{1}(\lambda), \widehat{U}_n^{1}(\lambda) \right) \right\}\right |= o(1).
    \end{align}
    Therefore, the results follow by combining eq. (\ref{prth1eq1}), (\ref{prth1eq2}) and (\ref{prth1eq3}).
\end{proof}

%
\disConv*
\MultierrEst*
\begin{proof}
Corollary \ref{cor:disConv} and Corollary \ref{cor:MultierrEst} are immediate results from Theorem \ref{thm:errEst}.
\end{proof}

\sol*
\begin{proof}
Sufficient condition is proven by similar argument of Lemma \ref{lem:sameCov}. It is remaining to prove necessary condition.
First, we should have
$\mathbb{X}_{S}^\T \mathbb{X}_{\mathrm{pseudo}} = \mathbb{X}_{S}^\T \mathbb{X}_{S^c}$. Therefore, $\mathbb{X}_{\mathrm{pseudo}}$ should have the form $\mathbb{X}_{\mathrm{pseudo}} = \mathbb{X}_{S^c} + V_\epsilon$, where $V$ is a matrix orthogonal to column space of $\mathbb{X}_{S}$. Then
\begin{align*}
\mathbb{X}_{\mathrm{pseudo}} &= \mathbb{X}_{S^c} + V_\epsilon \\
&=  \mathbb{X}_{S^c} + (I - P_{\mathbb{X}_{S}})\mathbb{X}_{S^c} + V_\epsilon - (I - P_{\mathbb{X}_{S}})\mathbb{X}_{S^c} \\
&= P_{\mathbb{X}_{S}}\mathbb{X}_{S^c} + V^*_\epsilon,
\end{align*}
where $V^*_\epsilon$ is a matrix orthogonal to column space of $\mathbb{X}_{S}$. Express $V^*_\epsilon$ to be $V^*_\epsilon = \sqrt{n}V A$, where $V$ is an orthonormal matrix that orthogonal to column space of $\mathbb{X}_{S}$ and $A$ to be any matrix with right dimension.

Then we will prove $A^TA = (Q_{22} - Q_{21}Q_{11}^-Q_{12})$. To satisfy condition that
$\mathbb{X}_{\mathrm{pseudo}}^\T \mathbb{X}_{\mathrm{pseudo}} = \mathbb{X}_{S^c}^\T \mathbb{X}_{S^c}$. Namely, $\left (P_{\mathbb{X}_{S}}\mathbb{X}_{S^c} + \sqrt{n}V A \right )^\T \left (P_{\mathbb{X}_{S}}\mathbb{X}_{S^c} + \sqrt{n} V A \right ) =  \mathbb{X}_{S^c}^\T \mathbb{X}_{S^c}$. By simple calculation, we have
$n A^\T A = \mathbb{X}_{S^c}^\T (I - P_{\mathbb{X}_{S}}) \mathbb{X}_{S^c}$. Therefore, $A^TA = (Q_{22} - Q_{21}Q_{11}^-Q_{12})$.

\end{proof}

\subsection{Proof of error rate estimation with permutation added}
Suppose that fitting Lasso with $(\mathbb{X}, G \mathbb{X}_{\widehat{A}_{0, n}}, Y)$, denote the corresponding number of unimportant variables and important variables at $\lambda$ as $\widehat{U}^*_n(\lambda)$ and $\widehat{I}^*_n(\lambda)$ respectively. Assume $ E [ U^*_n(\lambda) /
\max\{U^*_n(\lambda) + I^*_n(\lambda), 1\} ] \geq  E [ \widehat{U}_n /
\max\{\widehat{U}_n + \widehat{I}_n, 1\} ]$, then we have the following result
\begin{proposition}
    \label{pro:perm}
If $\lim_{n\rightarrow \infty} P(\widehat{A}_{0,n} = A_0) = 1$, then
$\sup_{\lambda}  \left\{E(\widehat{p}_n(\lambda)) - p_n(\lambda) \right\} \geq o(1)$.
\end{proposition}
\begin{proof}
By similar argument with the proof of Theorem \ref{thm:errEst}, we have $$\sup_{\lambda}  \left[ E\left (\widehat{p}_n(\lambda) \right ) - E\left \{\frac{\widehat{U}^*_n(\lambda)}{ \max(\widehat{U}^*_n(\lambda) + \widehat{I}^*_n(\lambda), 1)} \right \} \right] =  o(1).$$ Then combining with the assumption above, we have $\sup_{\lambda}  \left\{E(\widehat{p}_n(\lambda)) - p_n(\lambda) \right\} \geq o(1)$.
\end{proof}
\begin{remark}
The assumption  $ E [ U^*_n(\lambda) /
\max\{U^*_n(\lambda) + I^*_n(\lambda), 1\} ] \geq  E [ \widehat{U}_n /
\max\{\widehat{U}_n + \widehat{I}_n, 1\} ]$ implies that FSR is higher if more unimportant variables are used to fit Lasso. One can easily verify the assumption if $G\mathbb{X}_{A_0}$ is orthogonal to $\mathbb{X}$.
\end{remark}

To remove the assumption    $ E [ U^*_n(\lambda) /
\max\{U^*_n(\lambda) + I^*_n(\lambda), 1\} ] \geq  E [ \widehat{U}_n /
\max\{\widehat{U}_n + \widehat{I}_n, 1\} ]$, one can modify the estimator as $\widehat{p}_n(\lambda) = \max \{\widehat{p}_{n, 1}(\lambda), \{\widehat{p}_{n, 2}(\lambda)\}$, where $\widehat{p}_{n, 1}(\lambda)$ and $\widehat{p}_{n, 2}(\lambda)$ is the FSR estimated at $\lambda$ using the method with and without permutation added respectively. Then $\sup_{\lambda}  \left\{E(\widehat{p}_n(\lambda)) - p_n(\lambda) \right\} \geq o(1)$ follows from $$\sup_{\lambda}  \left[ E(\widehat{p}_{n,2}(\lambda)) - E \left \{ \widehat{U}_n /
\max\{\widehat{U}_n + \widehat{I}_n, 1\} \right \} \right] \geq  o(1).$$ However, it may double the computational complexity, while no significant benefits is observed in the simulation studies since $\widehat{p}_{n, 1}(\lambda)$ is usually bigger than  $\widehat{p}_{n, 2}(\lambda)$.
\newpage
\section{Simulation results for $\alpha = 0.1, 0.3$}
\begin{figure}[!h]
        \begin{subfigure}{.5\textwidth}
            \centering
            \includegraphics[scale = 0.34]{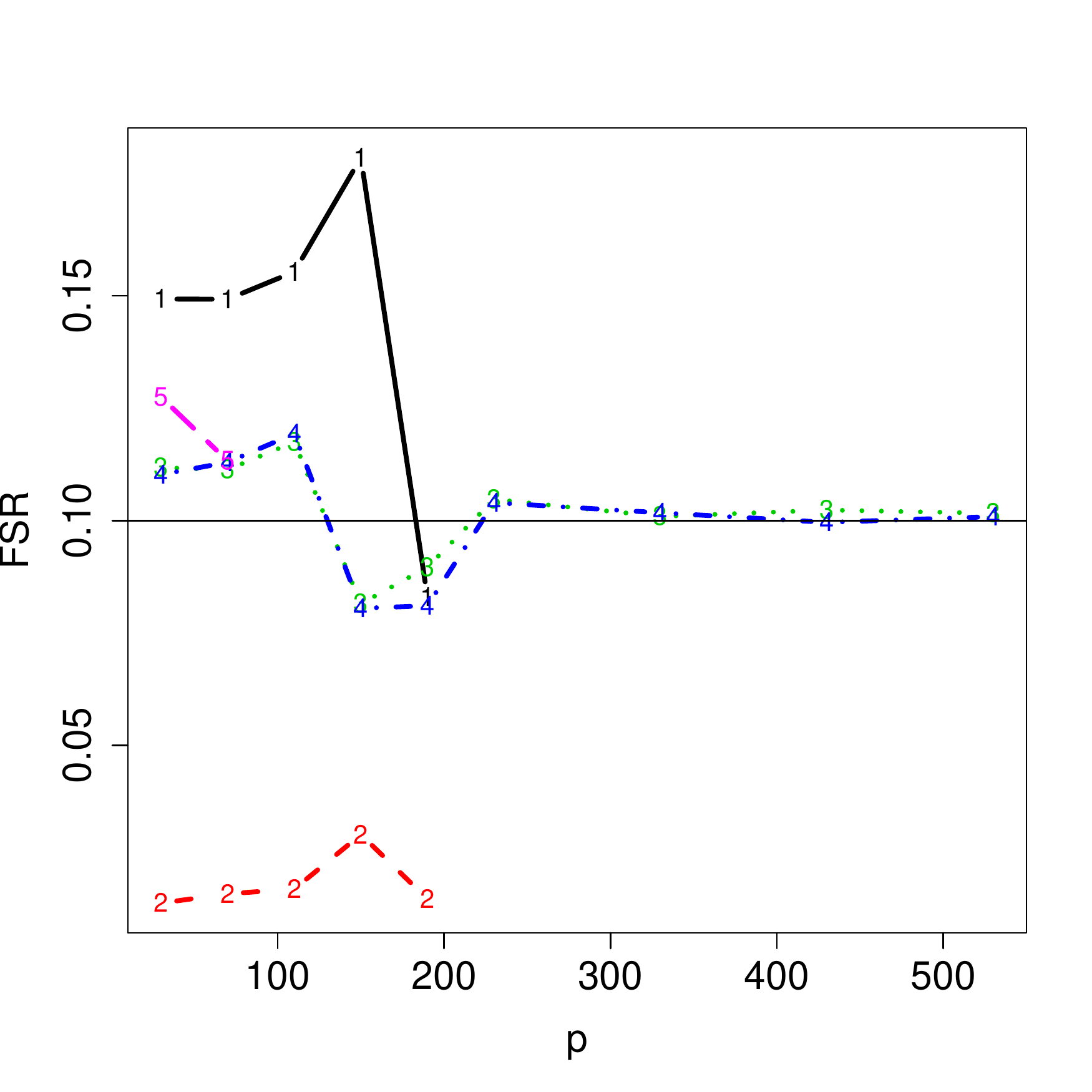}
            \caption{False selection rate vs $p$}
        \end{subfigure}%
        \begin{subfigure}{.5\textwidth}
            \centering
            \includegraphics[scale = 0.34]{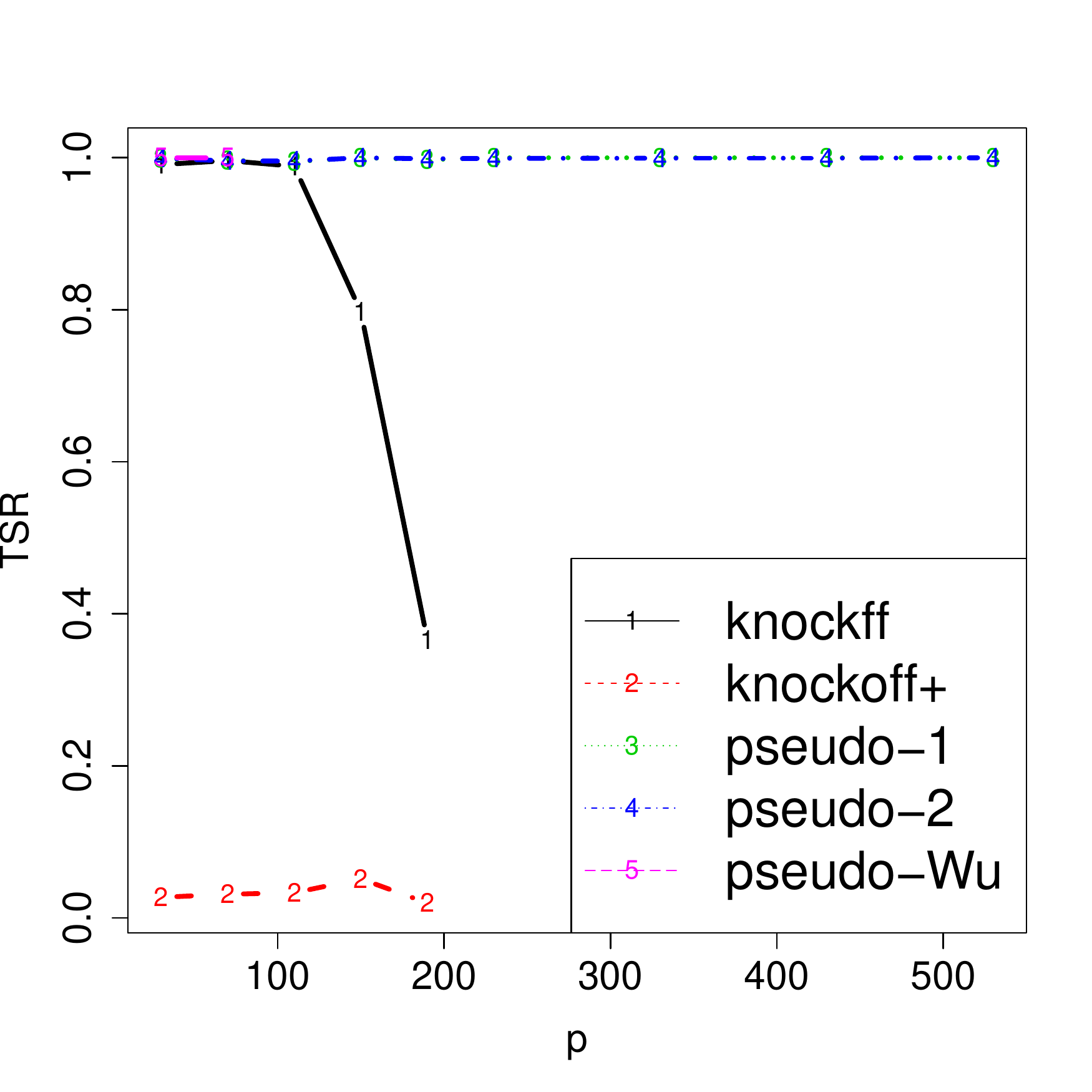}
            \caption{True selection rate vs $p$}
        \end{subfigure}
        \caption{Penalized regression; Performances under different dimensions at $\alpha = 0.1$.}
        \label{fig:dim2}
    \end{figure}

    \begin{figure}[!h]
        \begin{subfigure}{.5\textwidth}
            \centering
            \includegraphics[scale = 0.34]{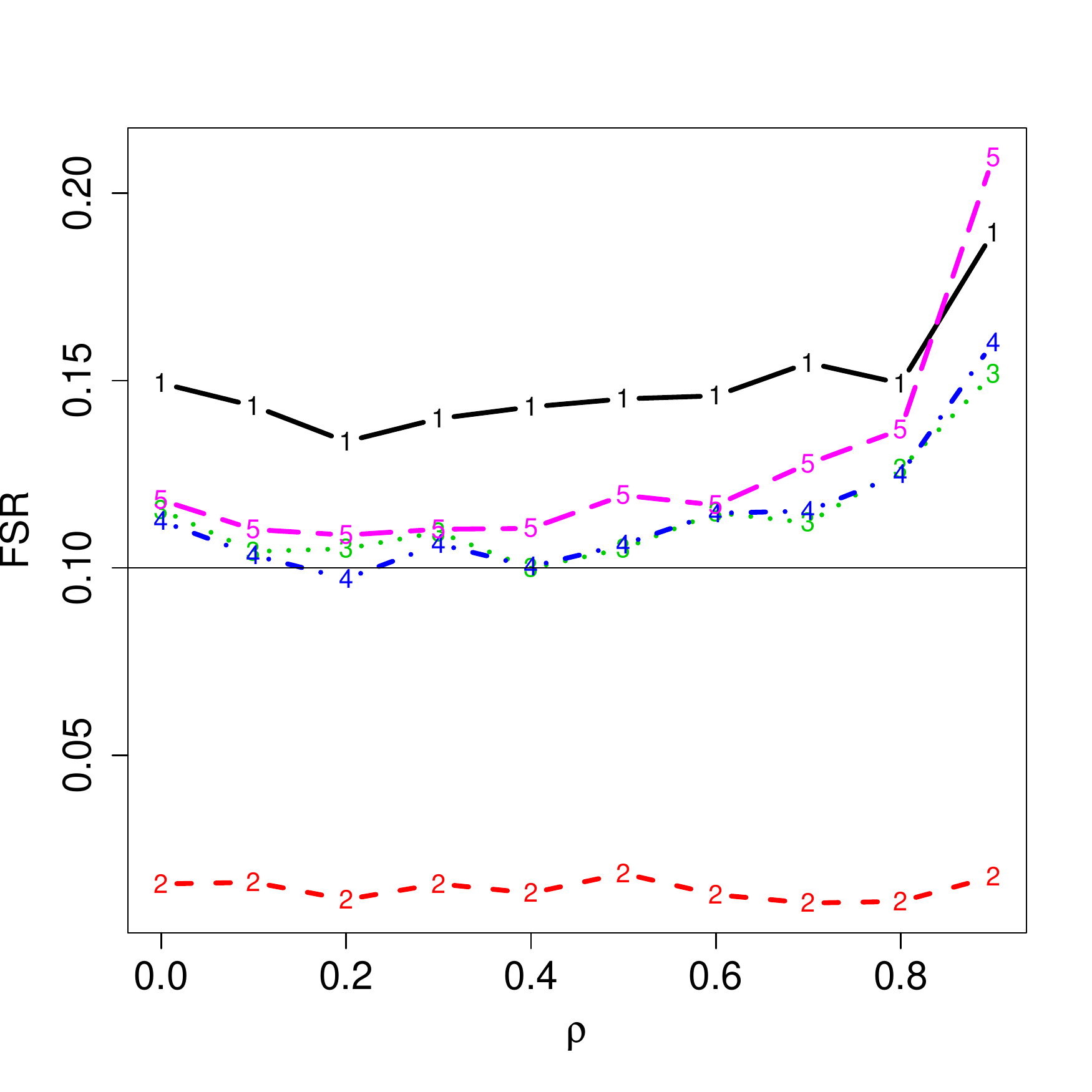}
            \caption{False selection rate vs $\rho$}
        \end{subfigure}%
        \begin{subfigure}{.5\textwidth}
            \centering
            \includegraphics[scale = 0.34]{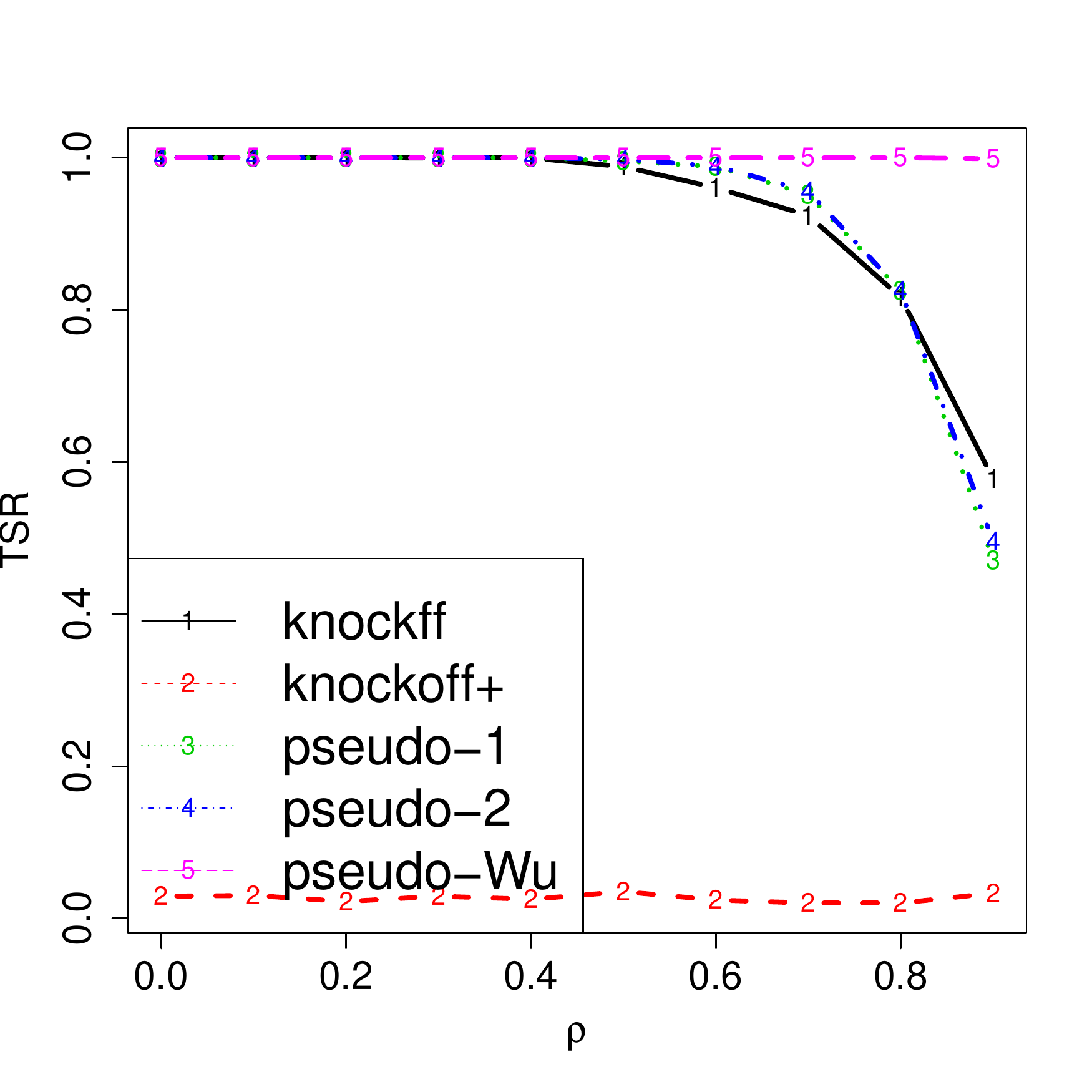}
            \caption{True selection rate vs $\rho$}
        \end{subfigure}
        \caption{Penalized regression; Performances under different correlations at $\alpha = 0.1$.}
        \label{fig:cor2}
    \end{figure}

    \begin{figure}[!h]
        \begin{subfigure}{.5\textwidth}
            \centering
            \includegraphics[scale = 0.34]{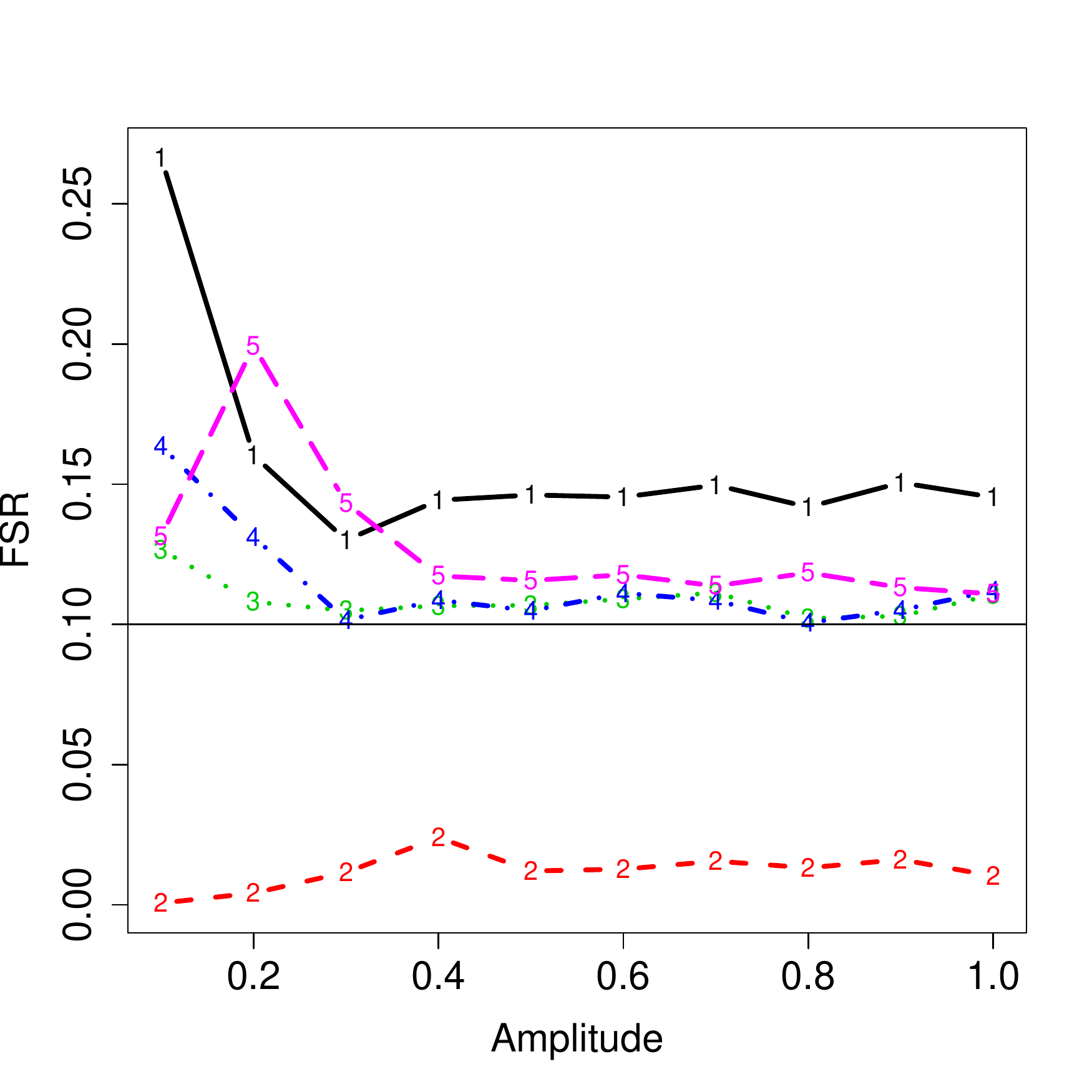}
            \caption{False selection rate vs $A$}
        \end{subfigure}%
        \begin{subfigure}{.5\textwidth}
            \centering
            \includegraphics[scale = 0.34]{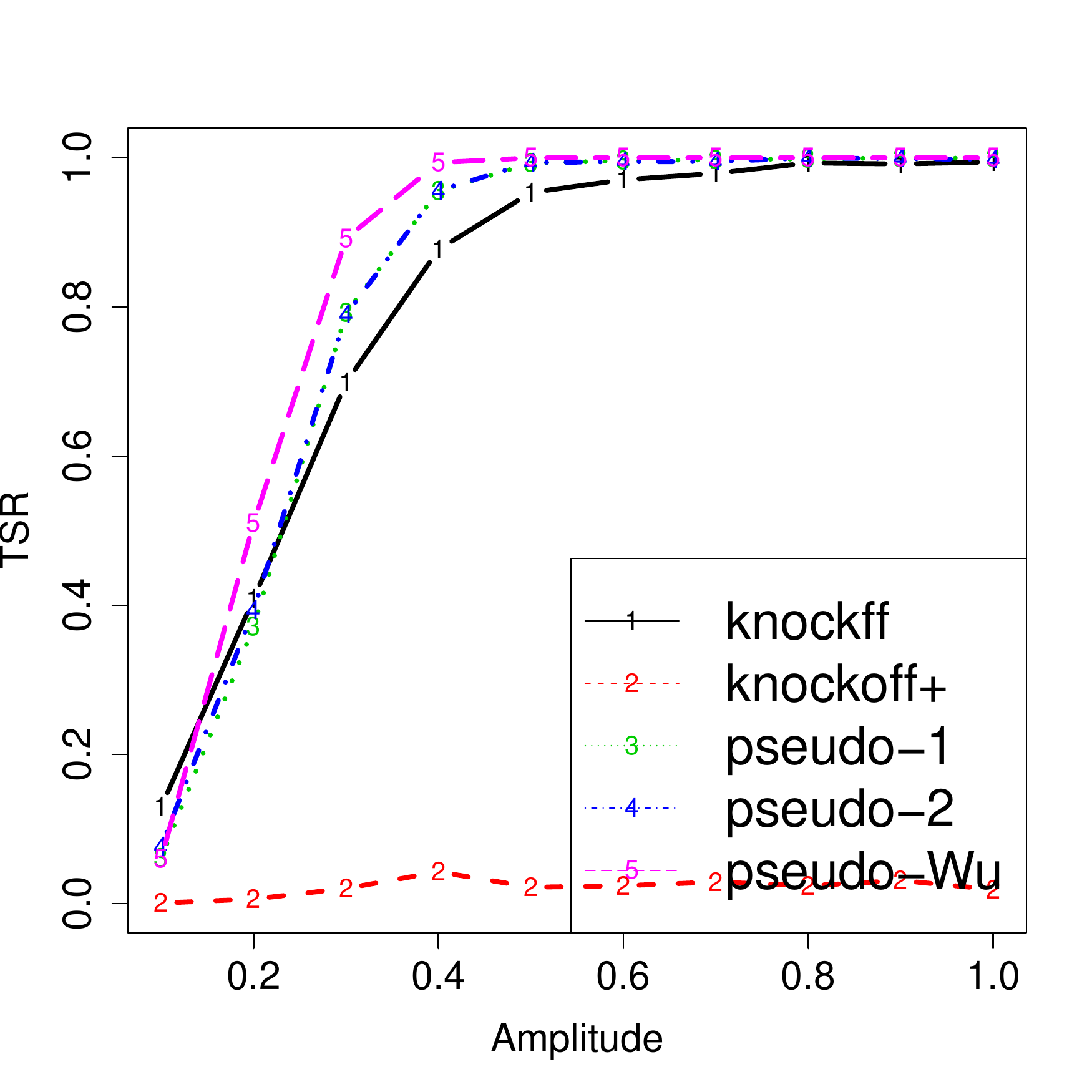}
            \caption{True selection rate vs $A$}
        \end{subfigure}
        \caption{Penalized regression; Performances under different coefficient amplitude at $\alpha = 0.1$.}
        \label{fig:amp2}
    \end{figure}

    \begin{figure}[!h]
        \begin{subfigure}{.5\textwidth}
            \centering
            \includegraphics[scale = 0.32]{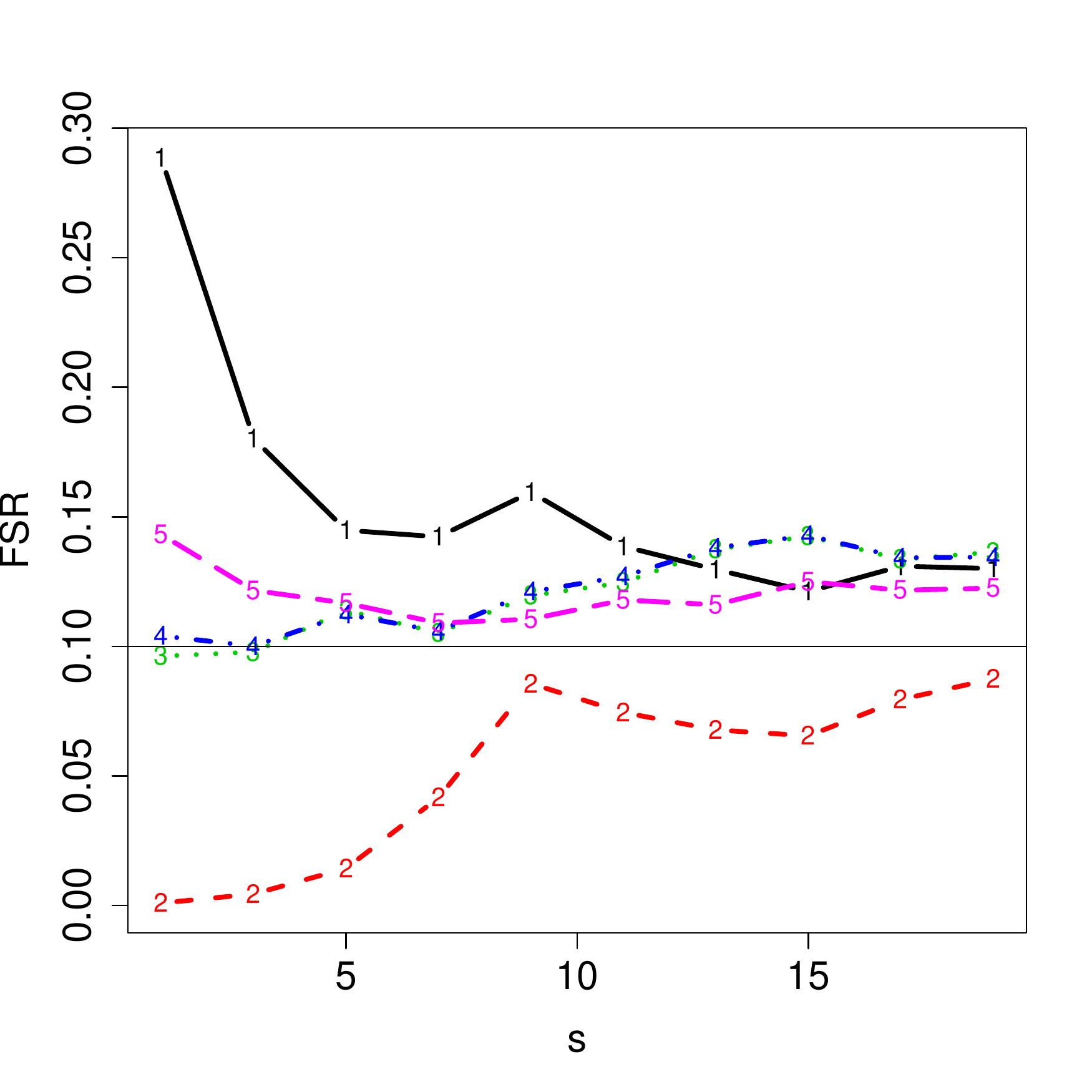}
            \caption{False selection rate vs $s$}
        \end{subfigure}%
        \begin{subfigure}{.5\textwidth}
            \centering
            \includegraphics[scale = 0.32]{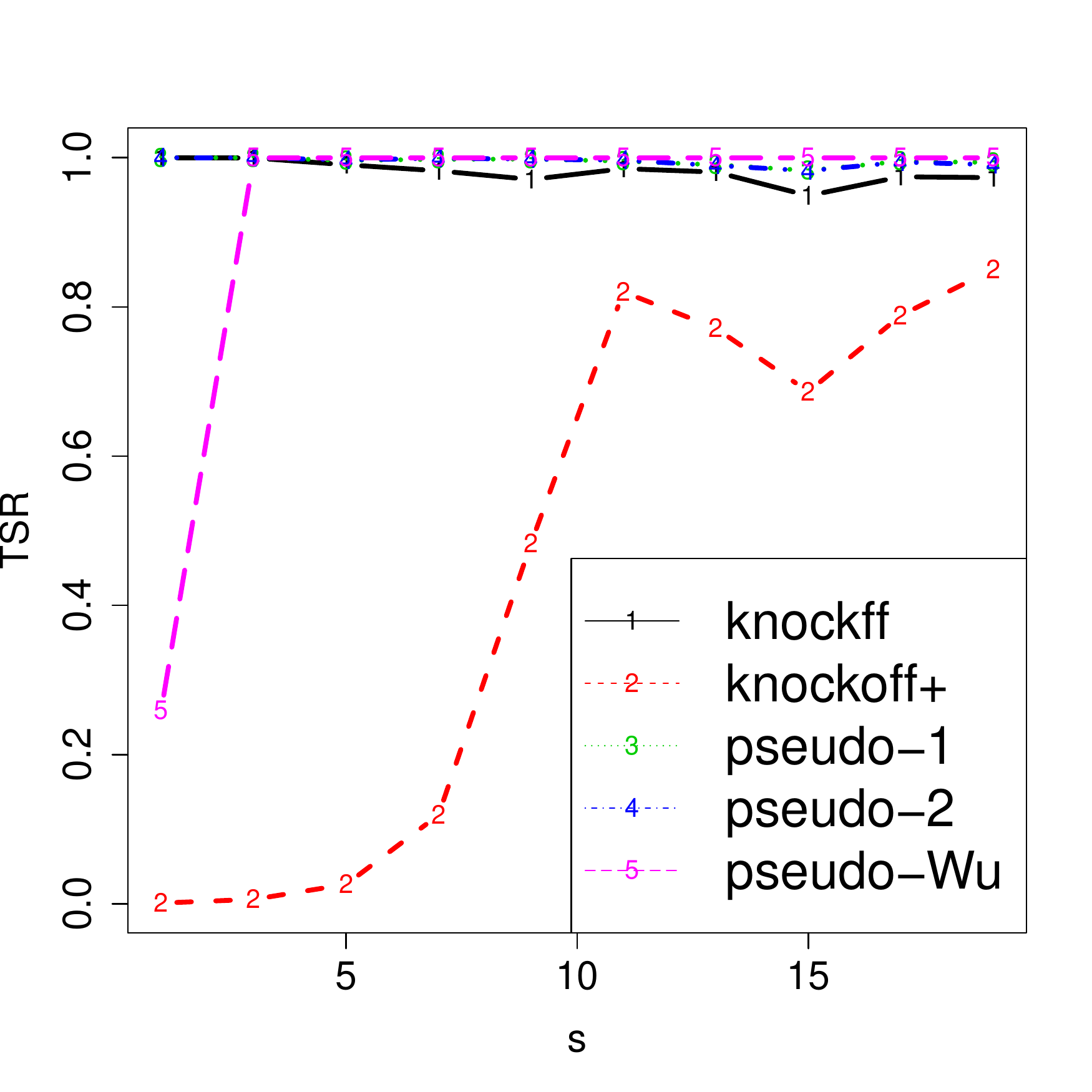}
            \caption{True selection rate vs $s$}
        \end{subfigure}
        \caption{Penalized regression; Performances under different number of nonzero coefficients at $\alpha = 0.1$.}
        \label{fig:spa2}
    \end{figure}

\begin{figure}[!h]
    \begin{subfigure}{.5\textwidth}
        \centering
        \includegraphics[scale=0.35]{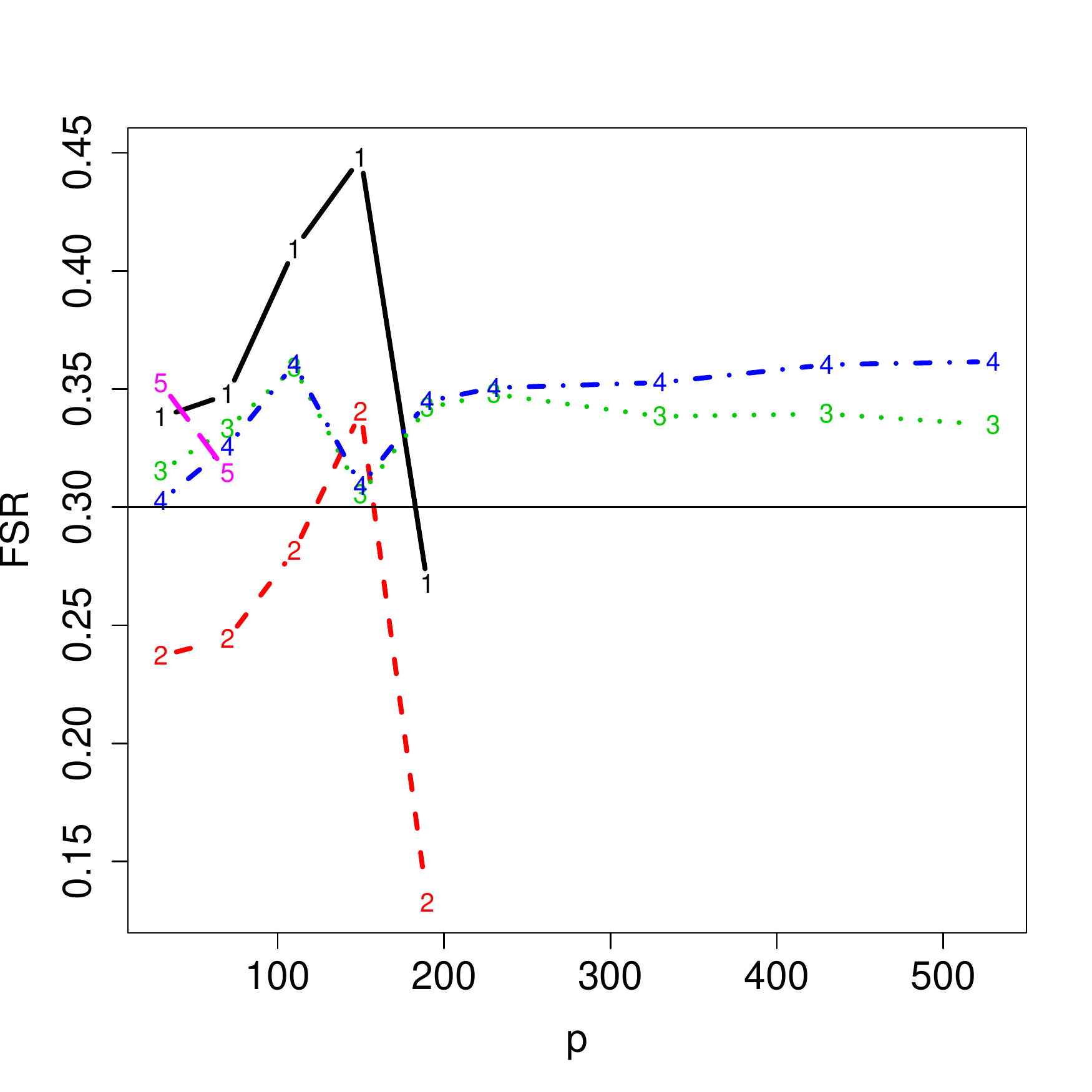}
        \caption{False selection rate vs $p$}
    \end{subfigure}%
    \begin{subfigure}{.5\textwidth}
        \centering
        \includegraphics[scale=0.35]{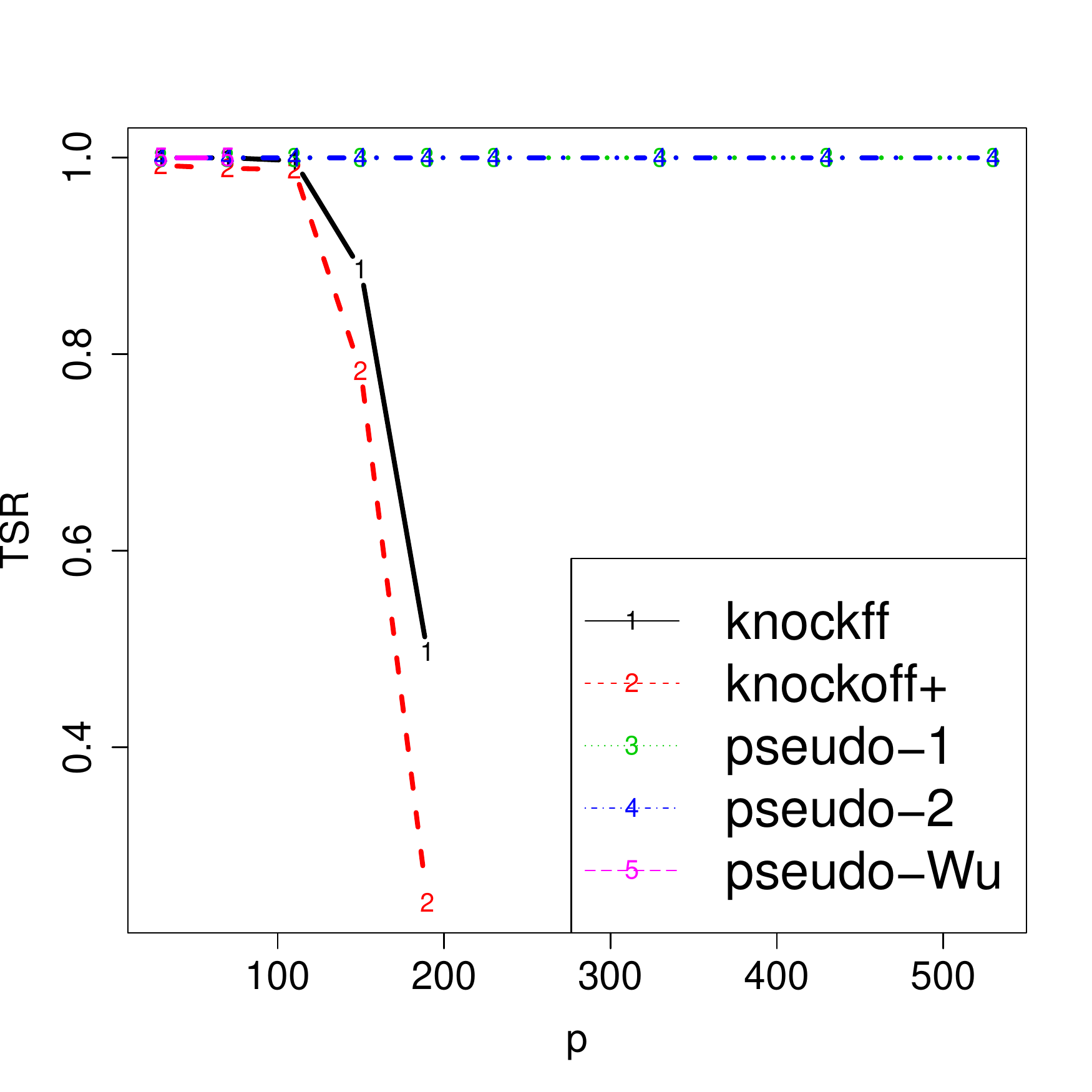}
        \caption{True selection rate vs $p$}
    \end{subfigure}
    \caption{Penalized regression; Performances under different dimensions at $\alpha = 0.3$.}
    \label{fig:dim3}
\end{figure}

\begin{figure}[!h]
    \begin{subfigure}{.5\textwidth}
        \centering
        \includegraphics[scale=0.35]{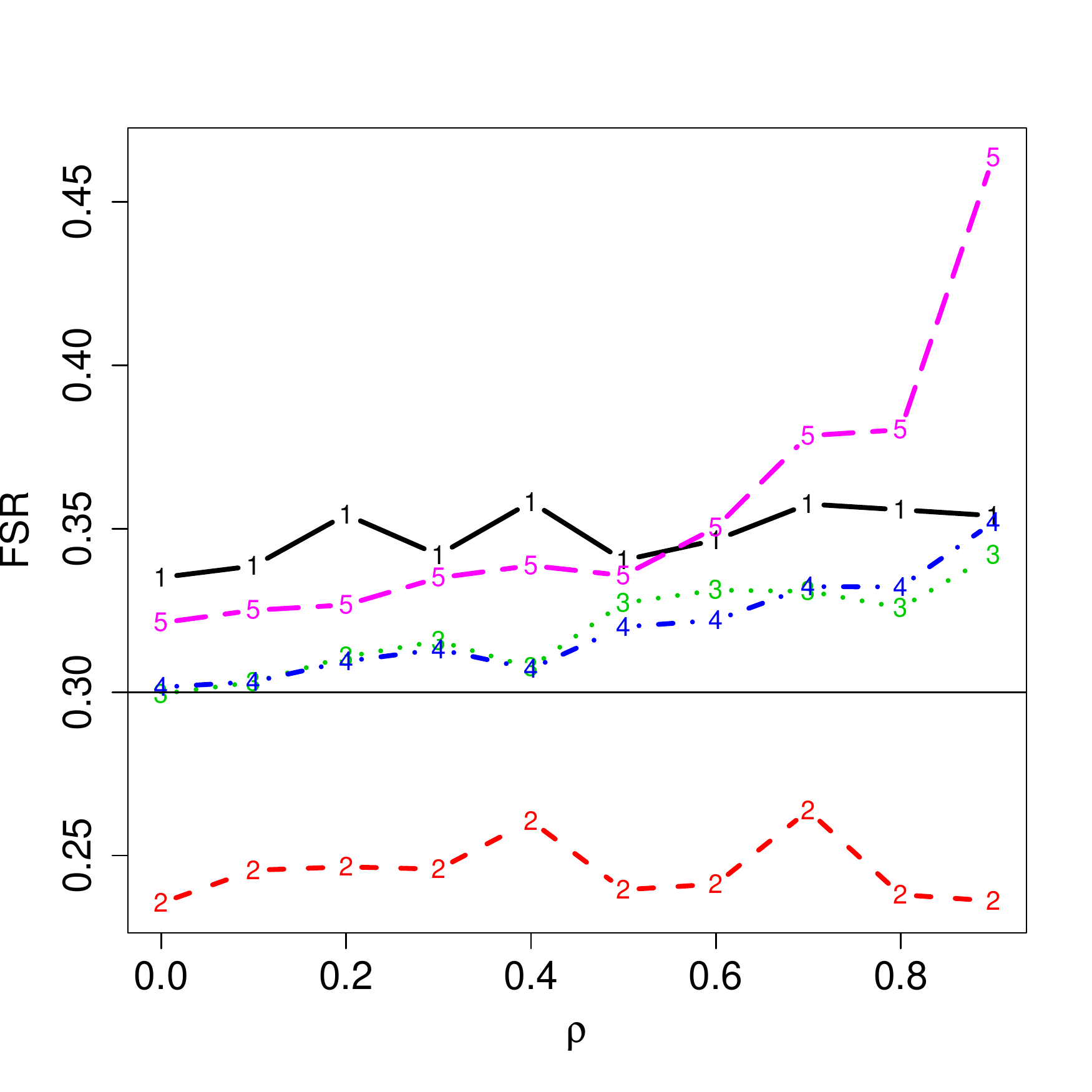}
        \caption{False selection rate vs $\rho$}
    \end{subfigure}%
    \begin{subfigure}{.5\textwidth}
        \centering
        \includegraphics[scale=0.35]{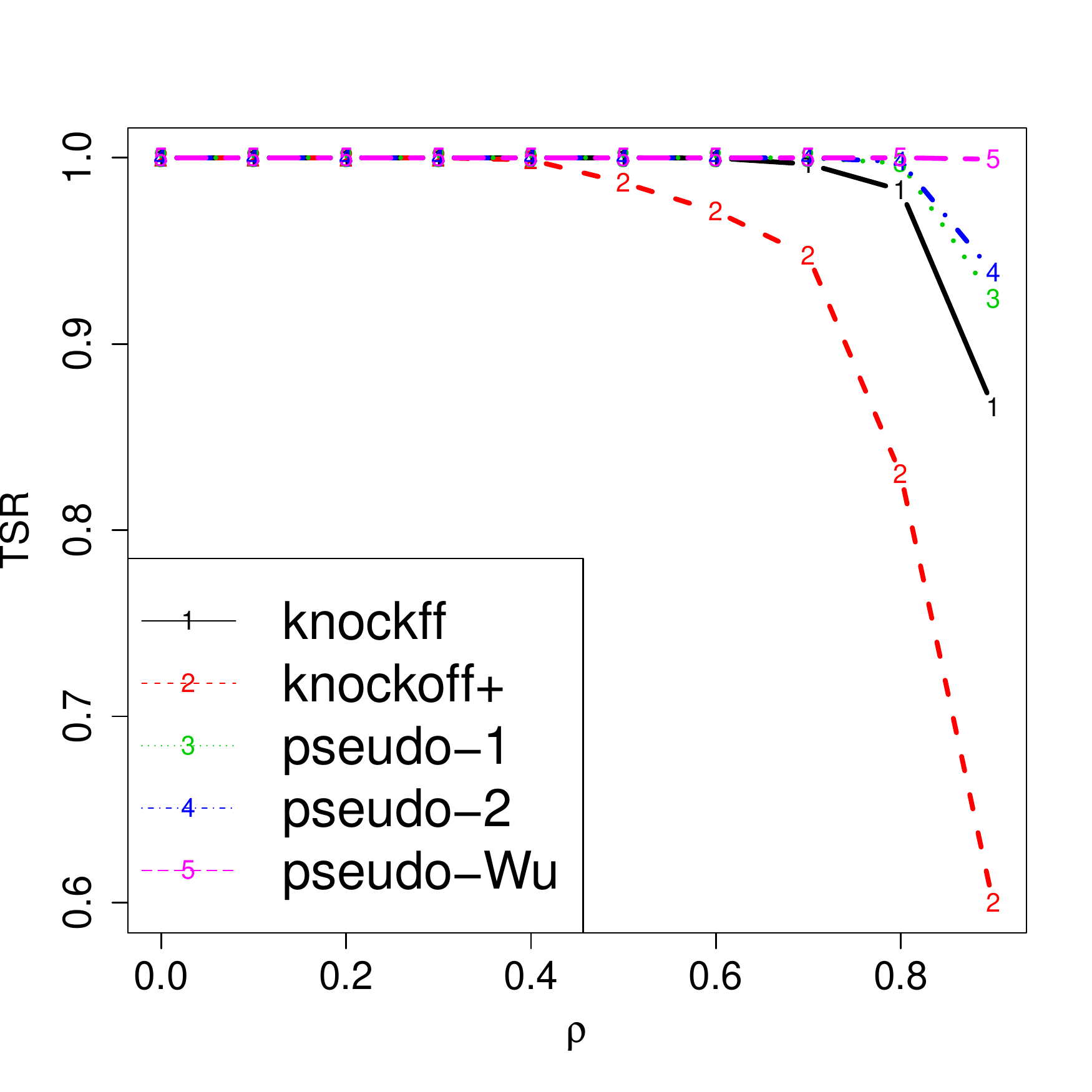}
        \caption{True selection rate vs $\rho$}
    \end{subfigure}
    \caption{Penalized regression; Performances under different correlations at $\alpha = 0.3$.}
    \label{fig:cor3}
\end{figure}

\begin{figure}[!h]
    \begin{subfigure}{.5\textwidth}
        \centering
        \includegraphics[scale=0.35]{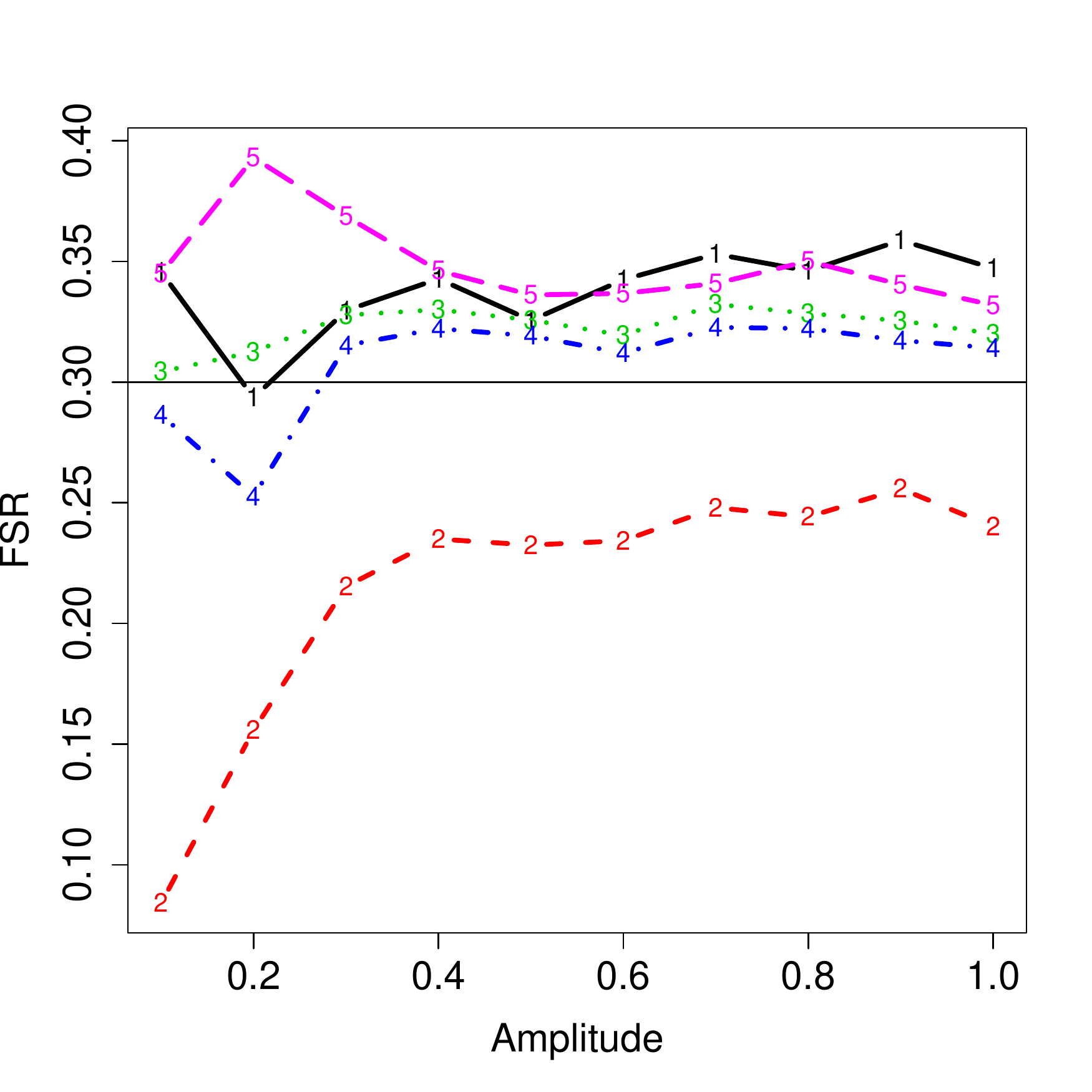}
        \caption{False selection rate vs $A$}
    \end{subfigure}%
    \begin{subfigure}{.5\textwidth}
        \centering
        \includegraphics[scale=0.35]{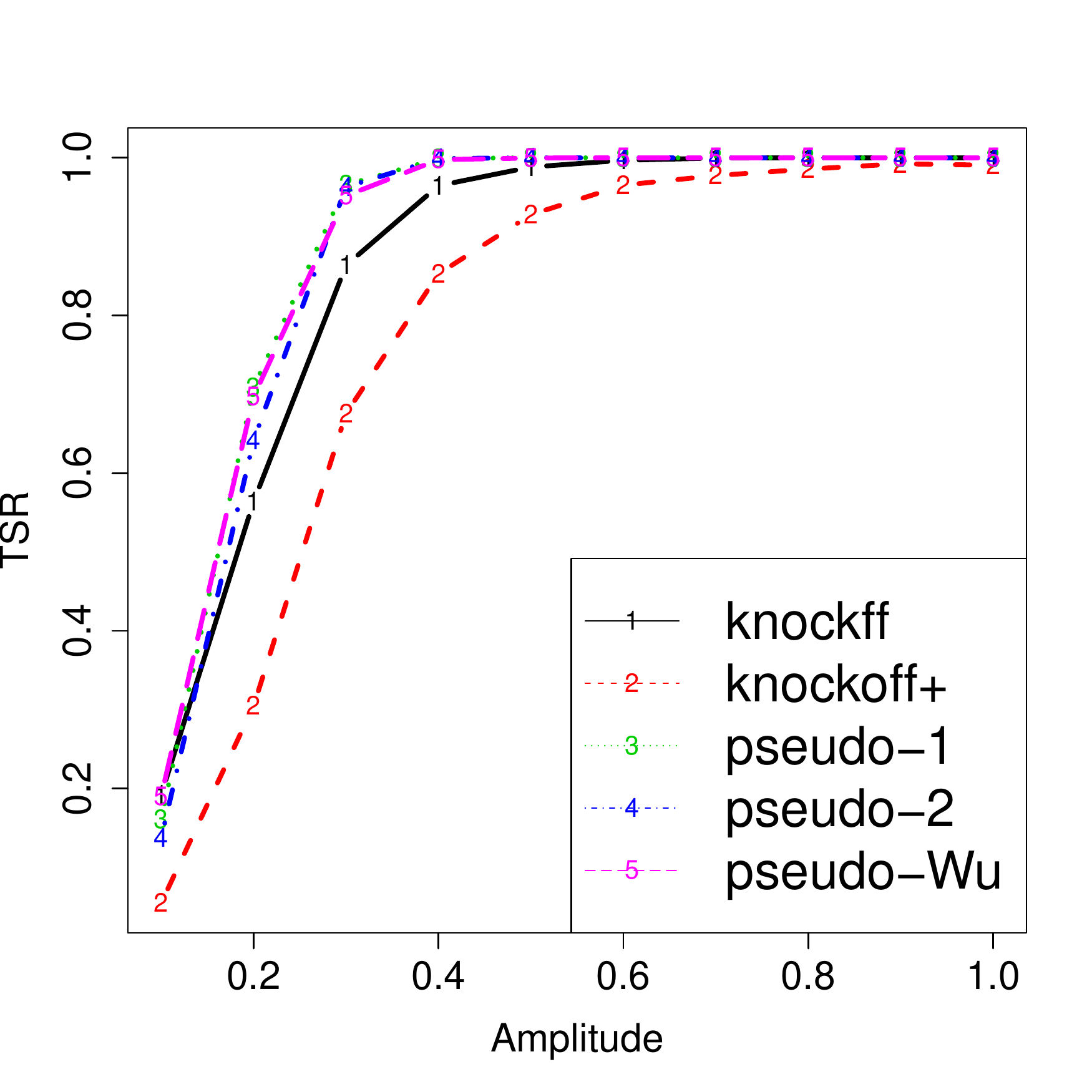}
        \caption{True selection rate vs $A$}
    \end{subfigure}
        \caption{Penalized regression; Performances under different coefficient amplitude at $\alpha = 0.3$.}
    \label{fig:amp3}
\end{figure}

\begin{figure}[!h]
    \begin{subfigure}{.5\textwidth}
        \centering
        \includegraphics[scale=0.32]{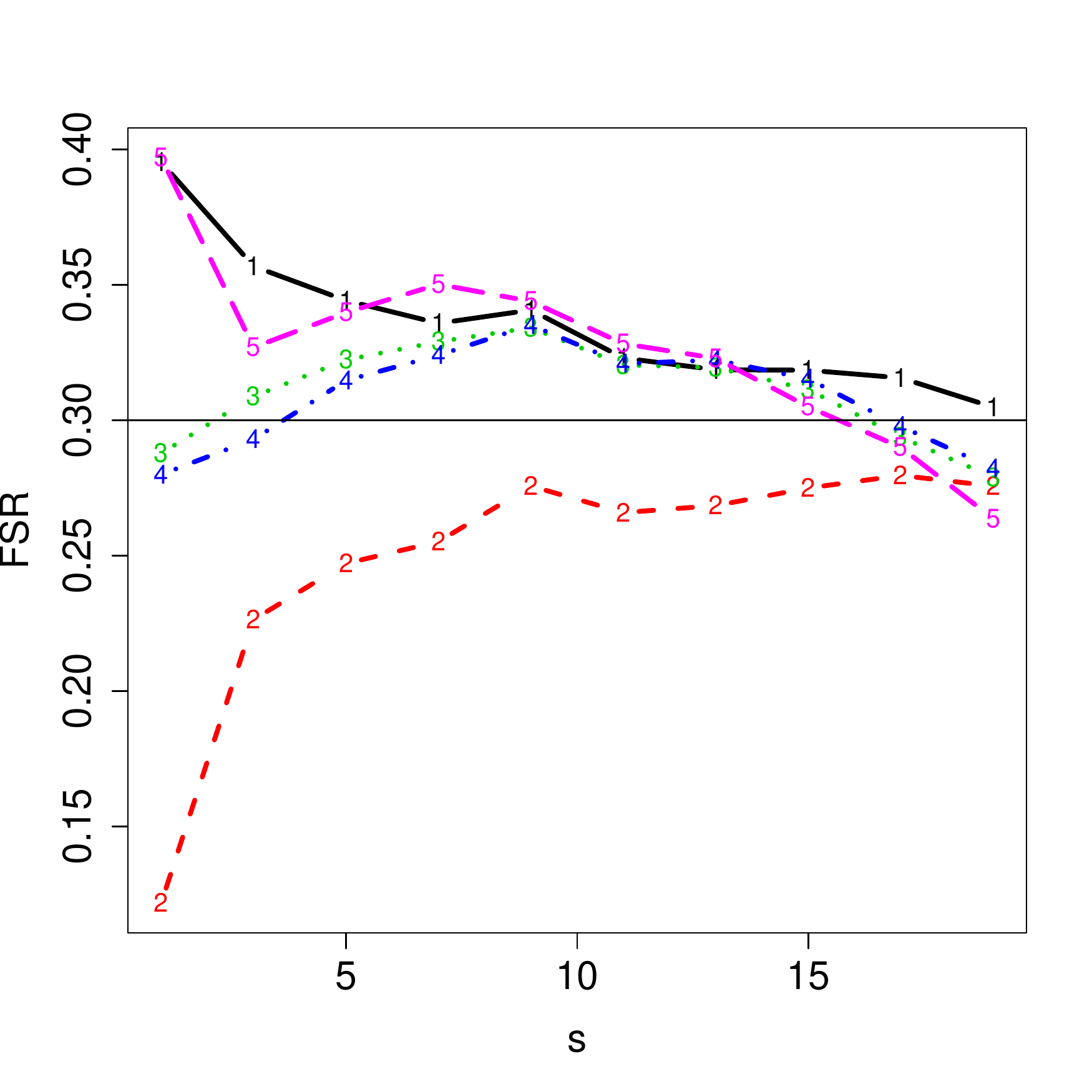}
        \caption{False selection rate vs $s$}
    \end{subfigure}%
    \begin{subfigure}{.5\textwidth}
        \centering
        \includegraphics[scale=0.32]{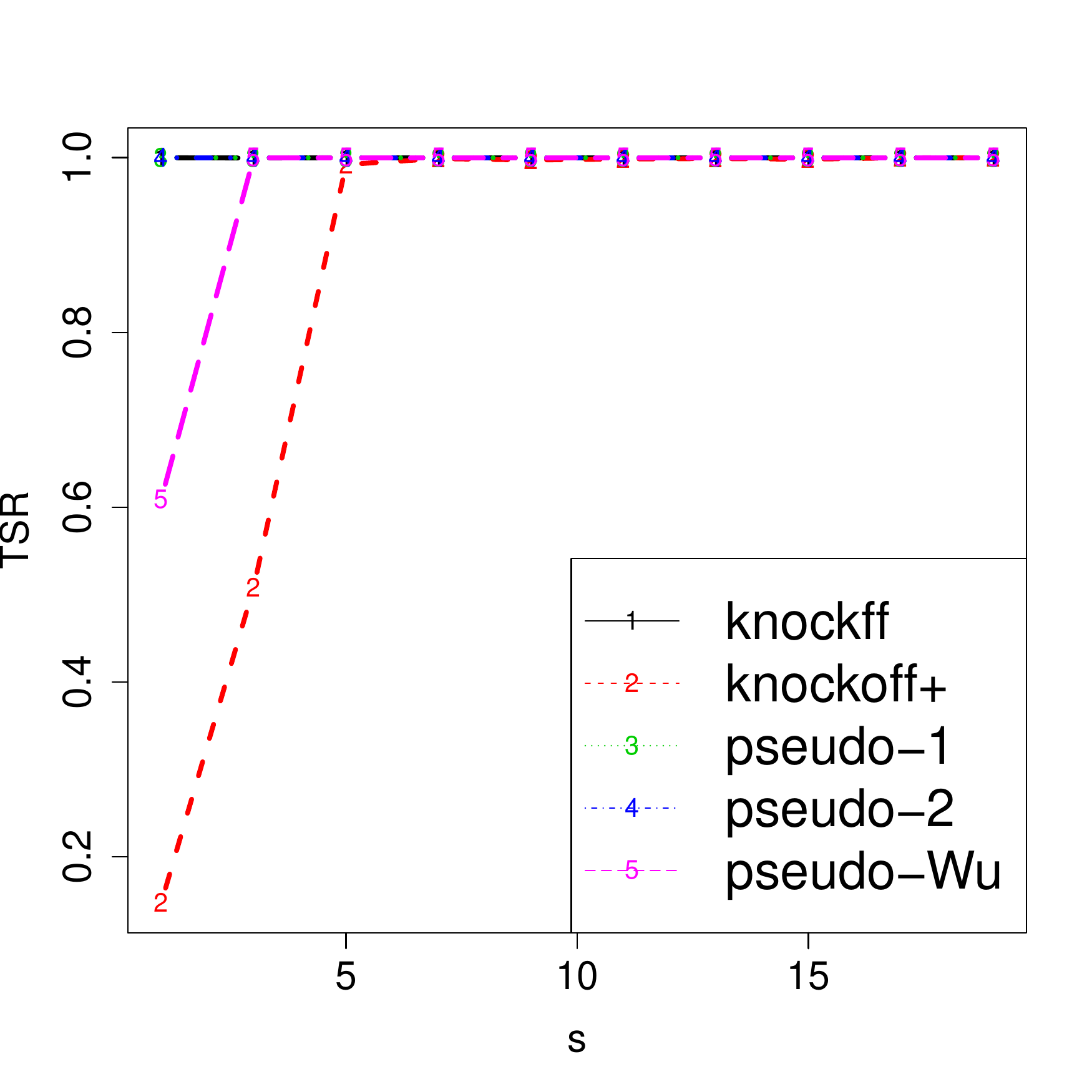}
        \caption{True selection rate vs $s$}
    \end{subfigure}
        \caption{Penalized regression; Performances under different number of nonzero coefficients at $\alpha = 0.3$.}
    \label{fig:spa3}
\end{figure}

\newpage
\section{Simulation results without adding permutation}
\begin{figure}[!h]
    \begin{subfigure}{.5\textwidth}
        \centering
        \includegraphics[scale=0.3]{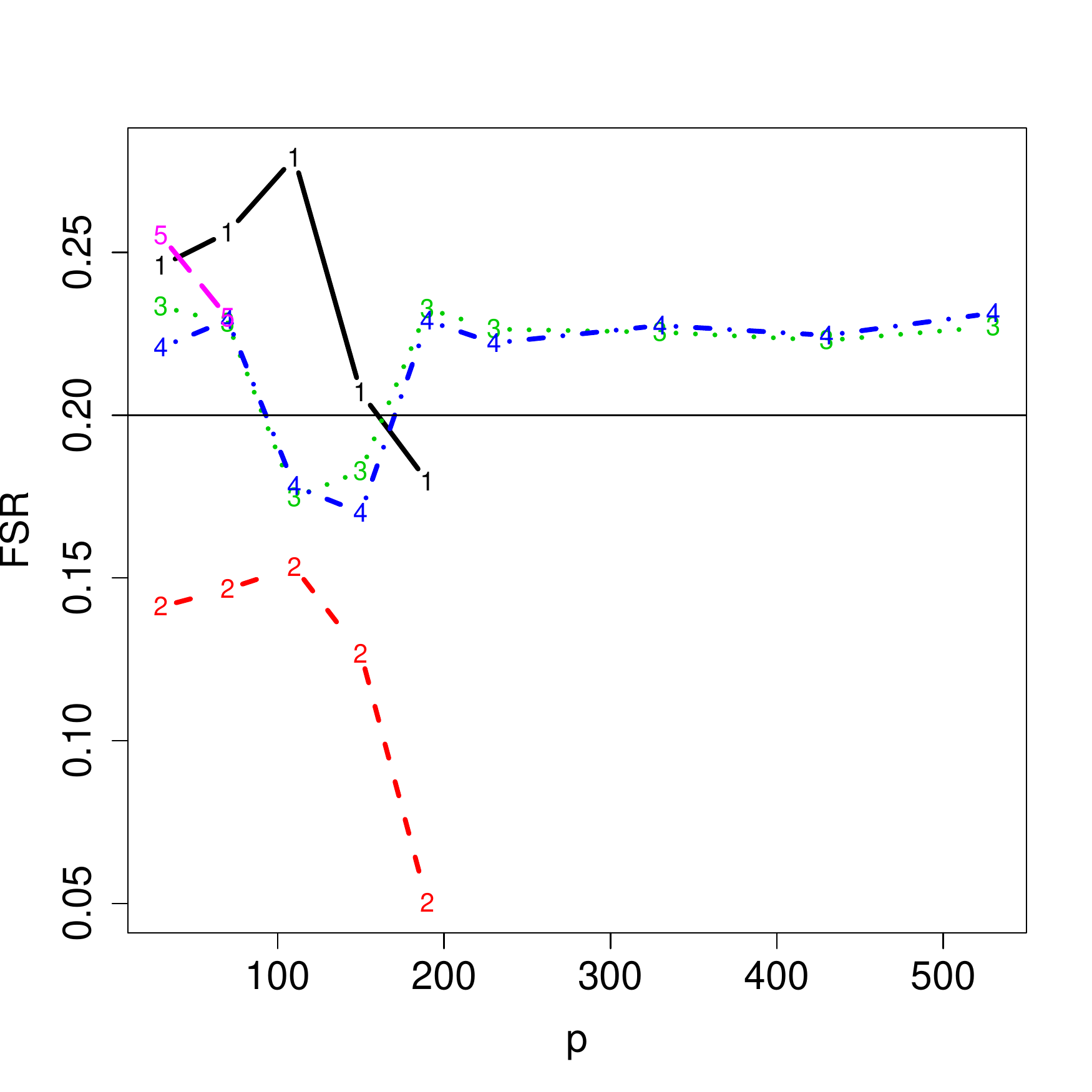}
        \caption{False selection rate vs $p$}
    \end{subfigure}%
    \begin{subfigure}{.5\textwidth}
        \centering
        \includegraphics[scale=0.3]{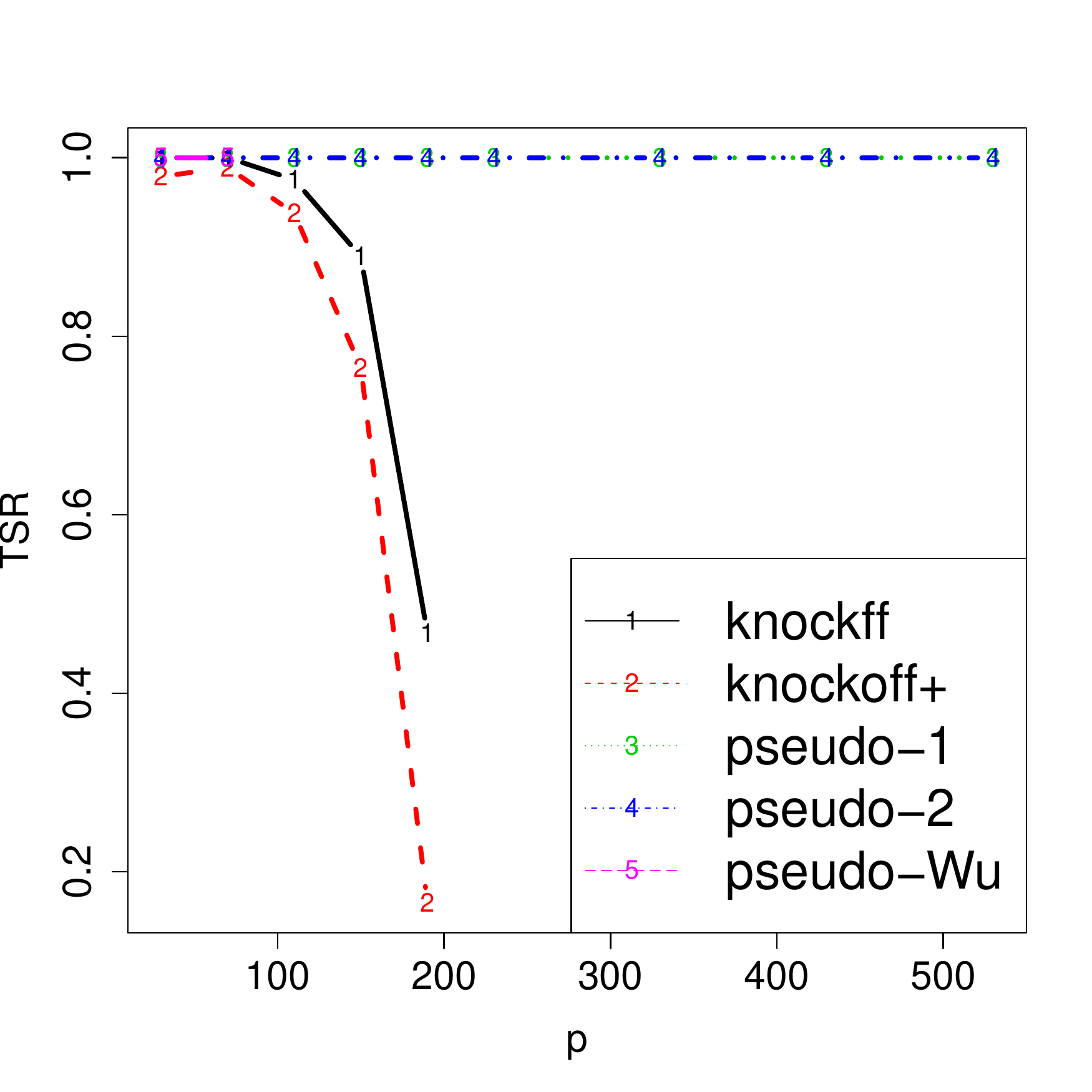}
        \caption{True selection rate vs $p$}
    \end{subfigure}
    \caption{Results without permutation added; Performances under different dimensions at $\alpha = 0.2$}
\end{figure}

\begin{figure}[!h]
    \begin{subfigure}{.5\textwidth}
        \centering
        \includegraphics[scale=0.29]{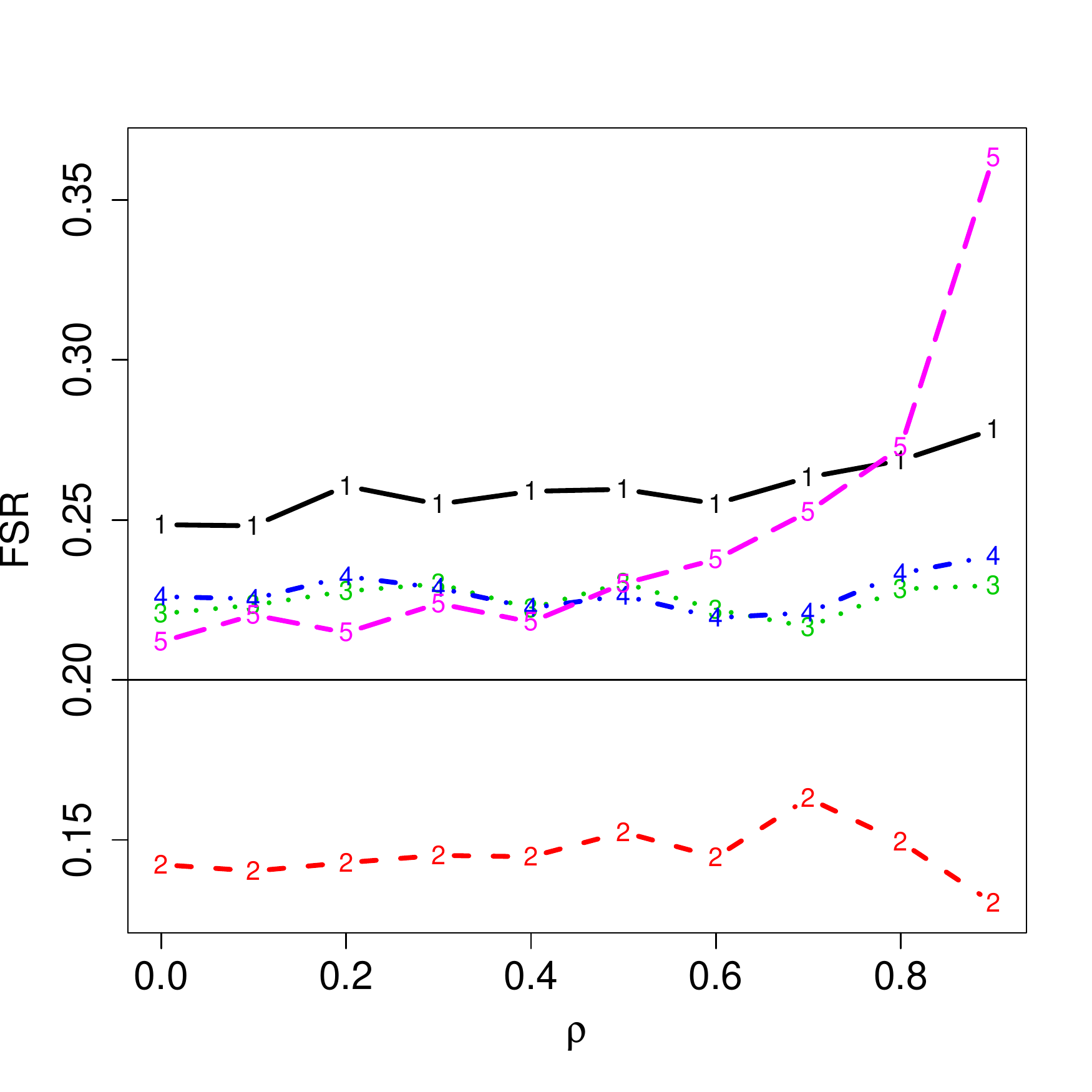}
        \caption{False selection rate vs $\rho$}
    \end{subfigure}%
    \begin{subfigure}{.5\textwidth}
        \centering
        \includegraphics[scale=0.29]{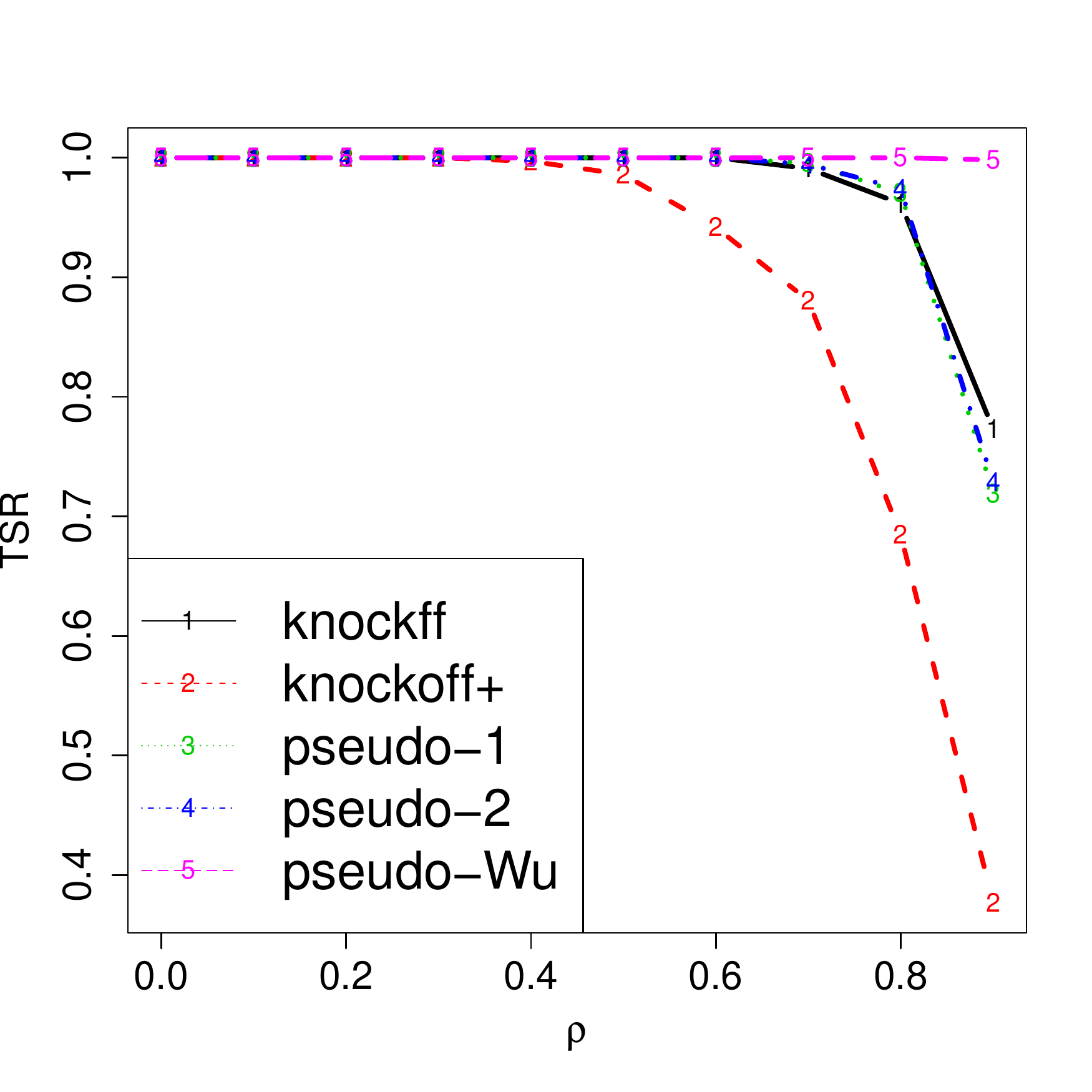}
        \caption{True selection rate vs $\rho$}
    \end{subfigure}
    \caption{Results without permutation added; Performances under different correlations at $\alpha = 0.2$}
\end{figure}

\begin{figure}[!h]
    \begin{subfigure}{.5\textwidth}
        \centering
        \includegraphics[scale=0.35]{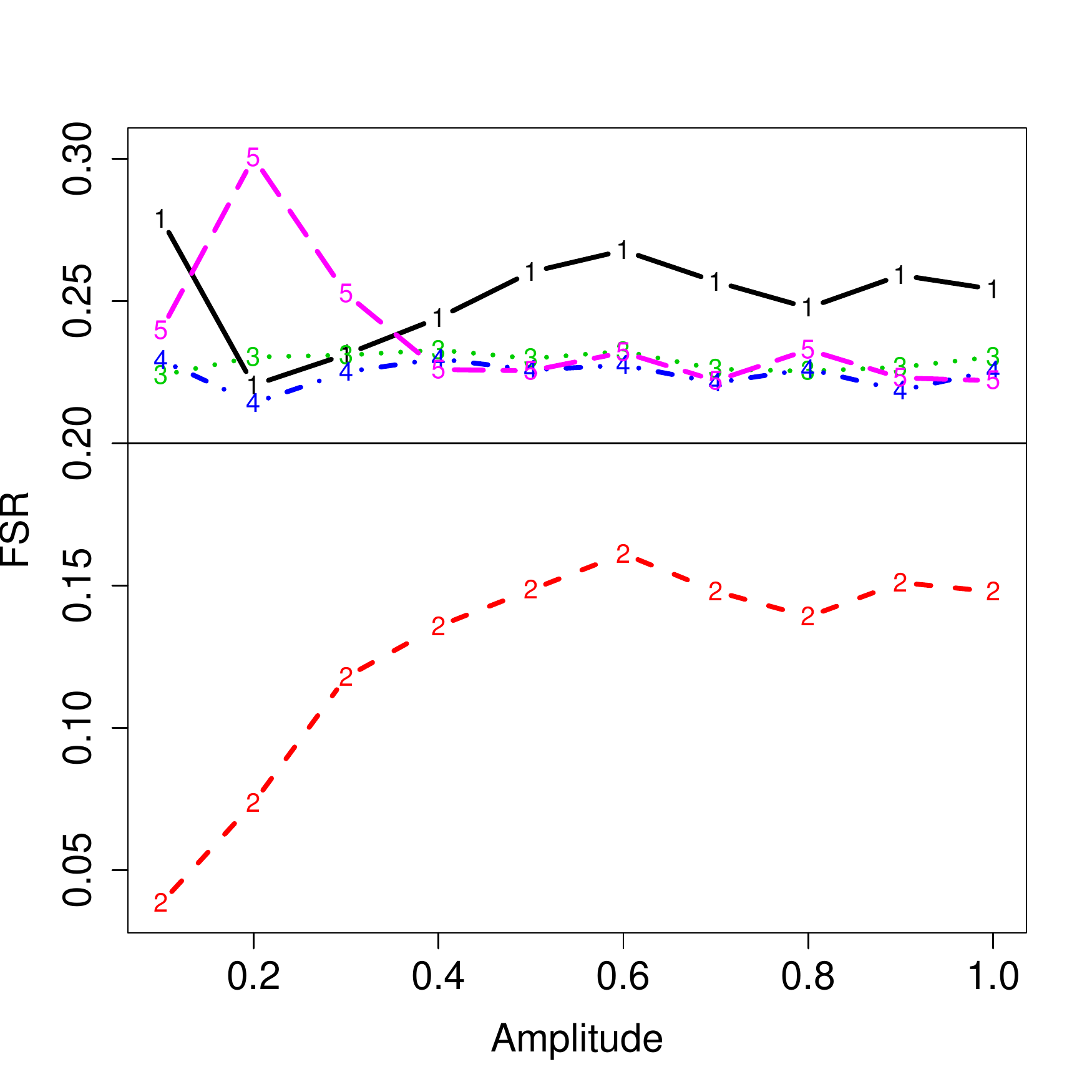}
        \caption{False selection rate vs $A$}
    \end{subfigure}%
    \begin{subfigure}{.5\textwidth}
        \centering
        \includegraphics[scale=0.35]{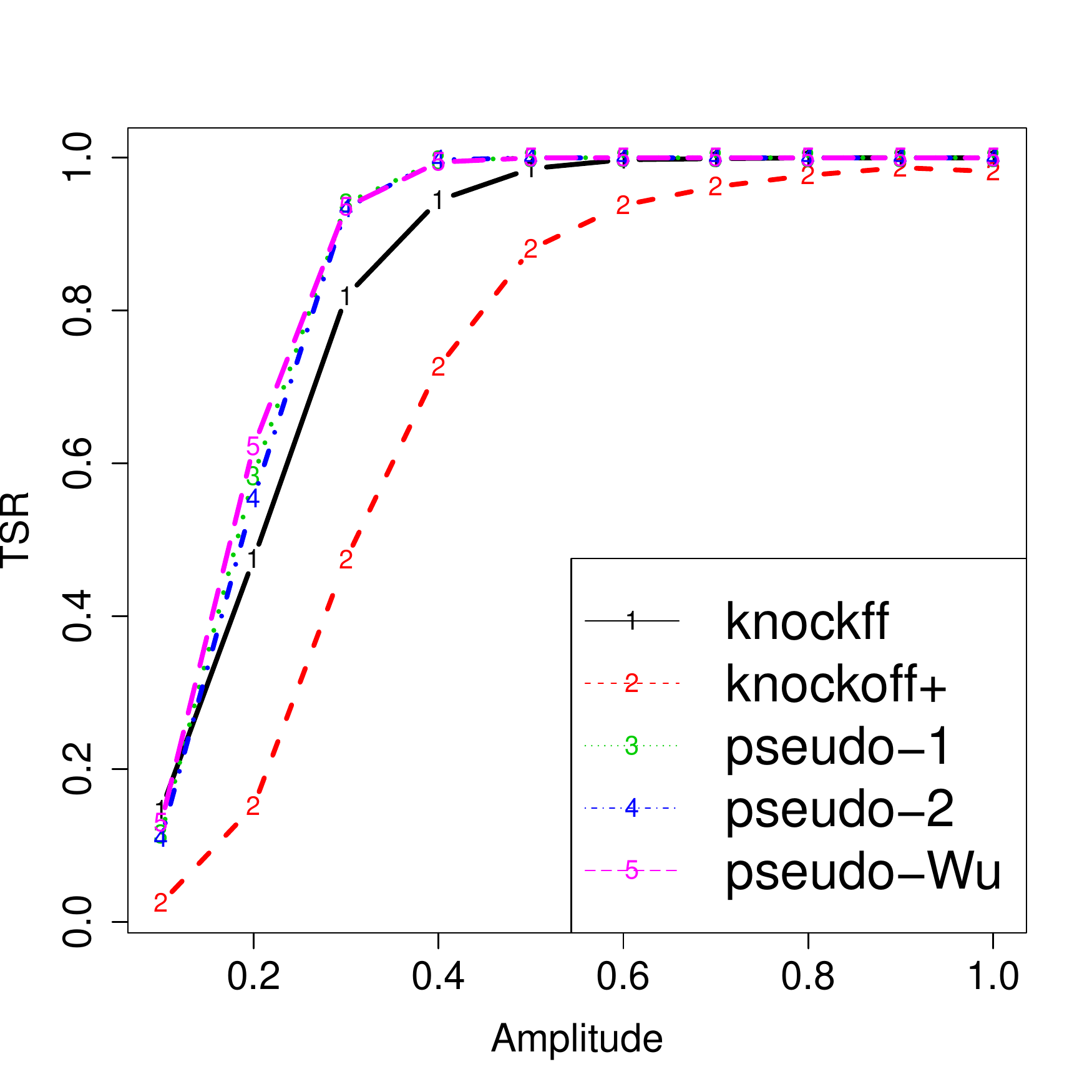}
        \caption{True selection rate vs $A$}
    \end{subfigure}
    \caption{Results without permutation added; Performances under different coefficient amplitude at $\alpha = 0.2$}
\end{figure}

\begin{figure}[!h]
    \begin{subfigure}{.5\textwidth}
        \centering
        \includegraphics[scale=0.32]{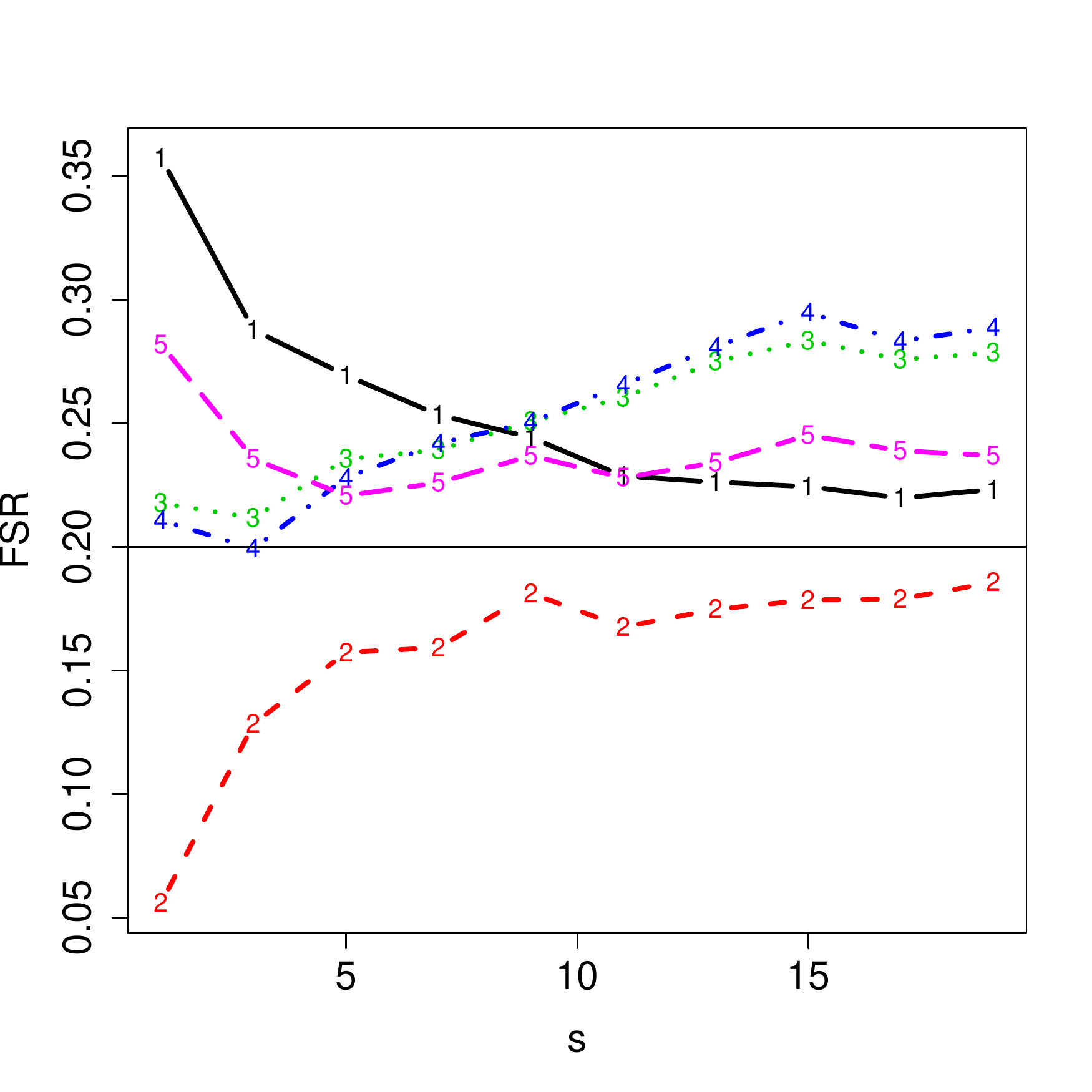}
        \caption{False selection rate vs $s$}
    \end{subfigure}%
    \begin{subfigure}{.5\textwidth}
        \centering
        \includegraphics[scale=0.32]{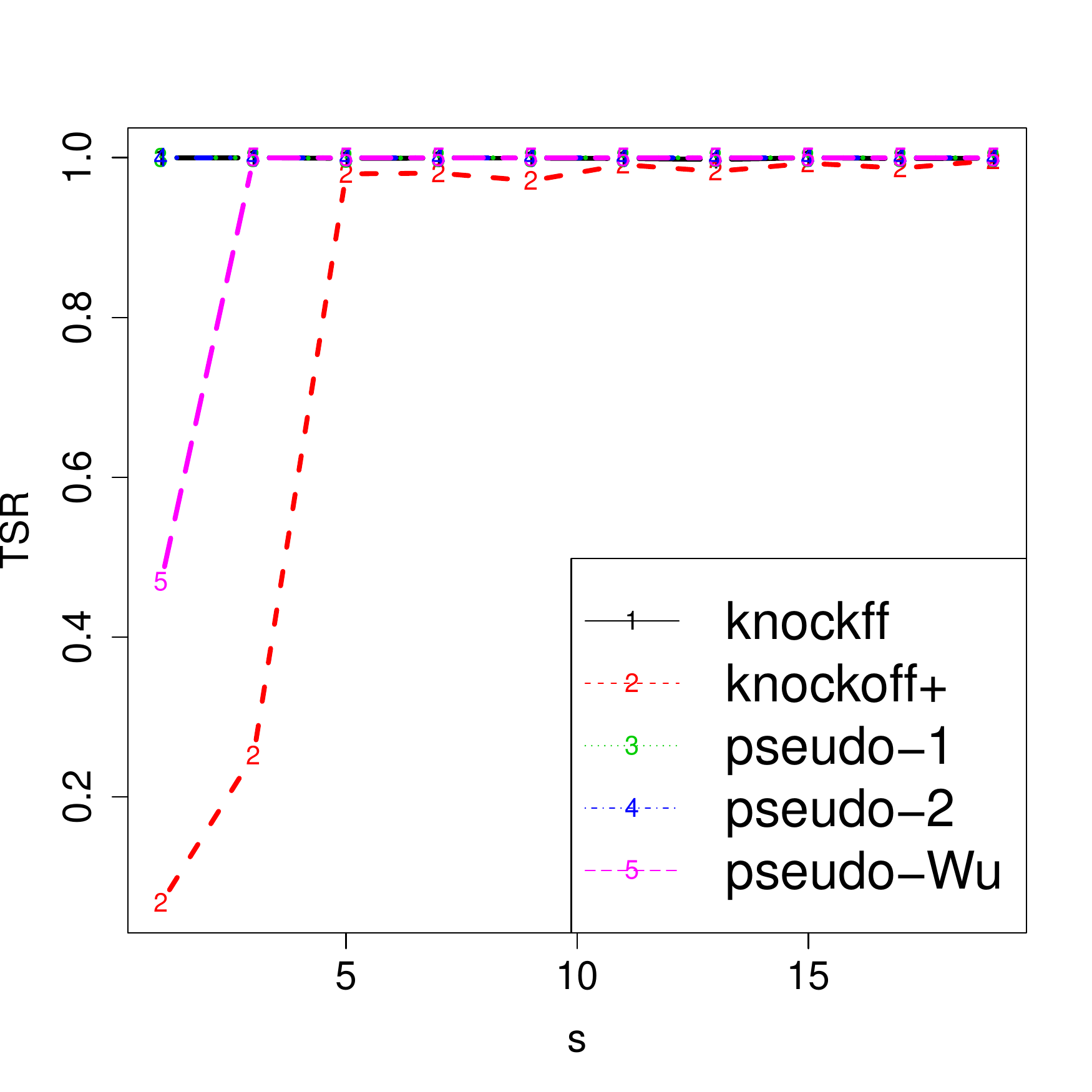}
        \caption{True selection rate vs $s$}
    \end{subfigure}
    \caption{Results without permutation added; Performances under different number of nonzero coefficients at $\alpha = 0.2$}
\end{figure}

\newpage
\section{Simulation results for logistic model}
\label{sec:logistic}
In this section, we apply the pseudo-variable method to the penalized logistic model. The Lasso estimator for logistic model at $\lambda$ is defined as
\begin{equation*}
\widehat{\boldsymbol{\beta}}(\lambda; \mathbb{Y}, \mathbb{X}) =
\arg\min_{\boldsymbol{\beta}\in\mathbb{R}^p}
\left \{ -\frac{1}{n}\sum_{i=1}^n \left \{Y_iX_i^\T\boldsymbol{\beta} - \log(1 + e^{X_i^T\boldsymbol{\beta}}) \right \}
+\lambda \sum_{j = 1}^p  |\beta_j| \right \}.
\end{equation*}

In the simulation, $Y_i, i = 1, \ldots, n$, are generated from Bernoulli distribution with parameter
$1 / \left\{1+ \exp(-X_i^T\boldsymbol{\beta}_0)\right\}$,
where $X_i \sim  N(0_{p \times 1}, C)$ with $C_{ij} = \rho^{|i-j|}$. Similar to simulation study of penalized regression, we consider four settings: different predictor dimensions, correlation magnitudes, signal amplitudes, and number of nonzero coefficients.

Simulation results are presented in Figure  \ref{fig:dim_logistic}, \ref{fig:cor_logistic}, \ref{fig:amp_logistic}, and \ref{fig:spa_logistic}. A similar pattern as seen for penalized regression can be observed from those figures. All methods except knockoff have a good control of FSR. For TSR, the pseudo-variable methods have better performance than the knockoff and knockoff+ method in most cases.

For all settings in above simulation study, $E(Y)$ is about $0.5$. This is considered as a easier problem comparing with $E(Y)$ is close to $0$ or $1$. To cover different cases for $E(Y)$, we introduce an intercept $c$ to the probability $1 / \{1 + \exp(c -X_i^\T\beta_0)\}$. By varying different $c$, datasets with different $E(Y)$ will be generated. We tried $c = 2.5$ and $5$ in our simulation studies. The corresponding $E(Y)$ is about $0.2$ and $0.05$ respectively. Similar patterns are observed as $c = 0$.

For the \emph{knockoff} package, only penalized regression is supported. To implement knockoff and knockoff+ method for penalized logistic model, we first call $create\_equicorrelated$ function in \emph{knockoff} package to construct knockoff variables, and then call \emph{glmnet} function with \emph{family = binomial} to get when the original variables and  knockoffs enter the solution path. Then we use the
default signed maximum statistics, i.e., $W_i = \max(Z_i, \tilde{Z}_i) \times \mathrm{sgn}(Z_i- \tilde{Z}_i)$, where $Z_i$ and $\tilde{Z}_i$ are the $\lambda$ when the $i$-th variable and corresponding knockoff enter the solution path.

\begin{figure}[!h]
    \begin{subfigure}{.5\textwidth}
        \centering
        \includegraphics[scale=0.32]{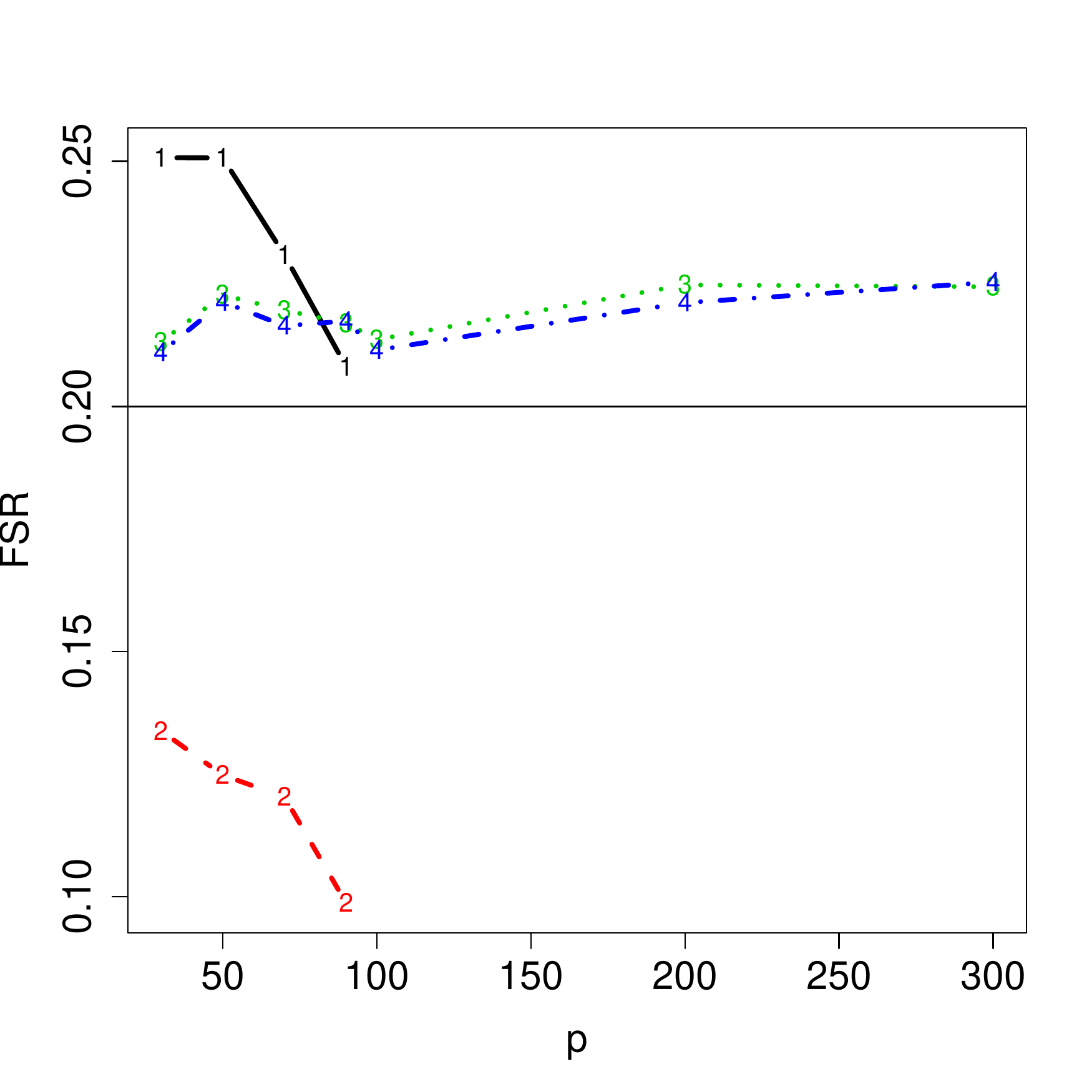}
        \caption{False selection rate vs $p$}
    \end{subfigure}%
    \begin{subfigure}{.5\textwidth}
        \centering
        \includegraphics[scale=0.32]{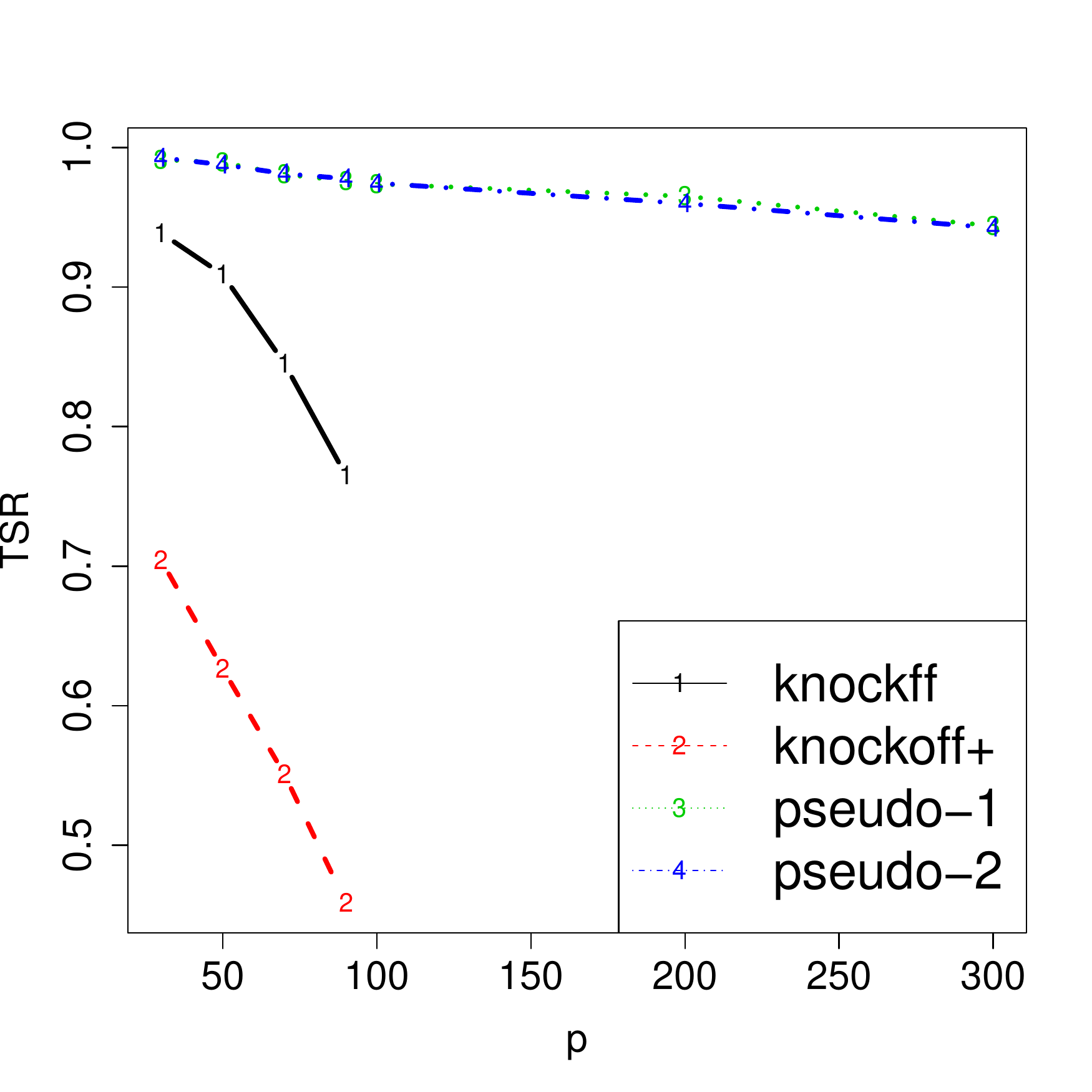}
        \caption{True selection rate vs $p$}
    \end{subfigure}
    \caption{Logistic model; Performances under different dimensions at $\alpha = 0.2$. Knockff and Knockoff+ method only work if $n > 2p$.}
    \label{fig:dim_logistic}
\end{figure}

\begin{figure}[!h]
    \begin{subfigure}{.5\textwidth}
        \centering
        \includegraphics[scale=0.3]{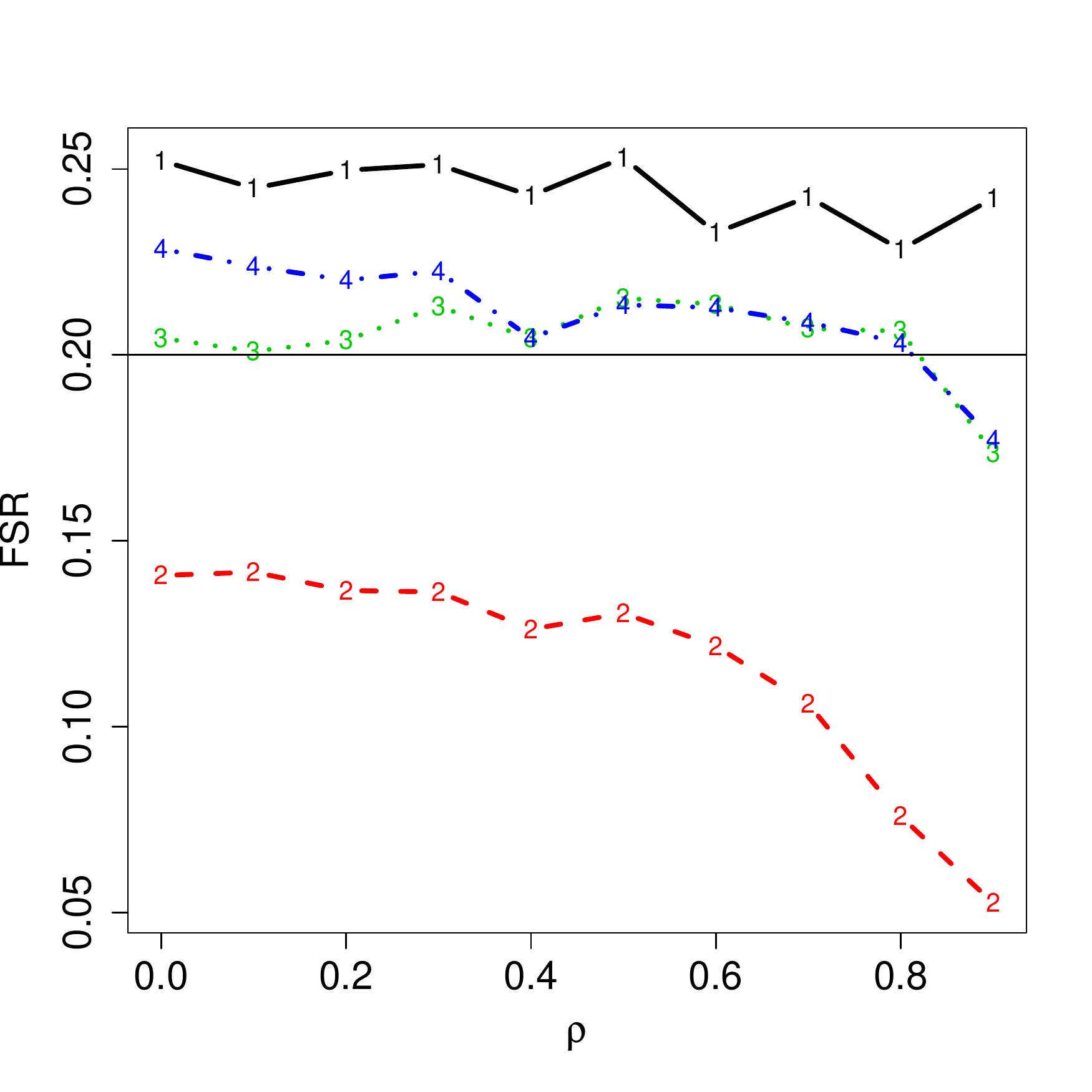}
        \caption{False selection rate vs $\rho$}
    \end{subfigure}%
    \begin{subfigure}{.5\textwidth}
        \centering
        \includegraphics[scale=0.3]{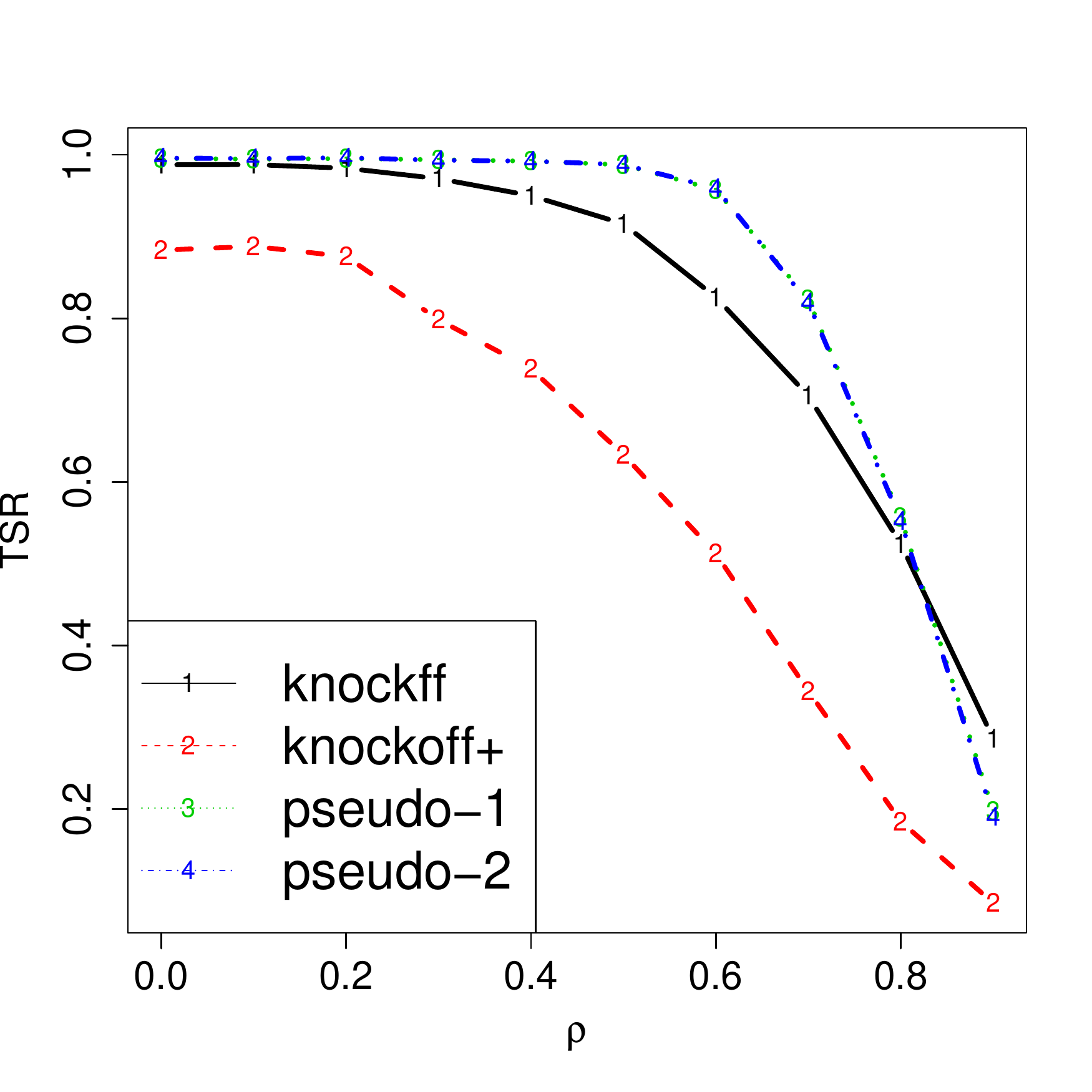}
        \caption{True selection rate vs $\rho$}
    \end{subfigure}
    \caption{Logistic model; Performances under different correlations at $\alpha = 0.2$.}
    \label{fig:cor_logistic}
\end{figure}

\begin{figure}[!h]
    \begin{subfigure}{.5\textwidth}
        \centering
        \includegraphics[scale=0.35]{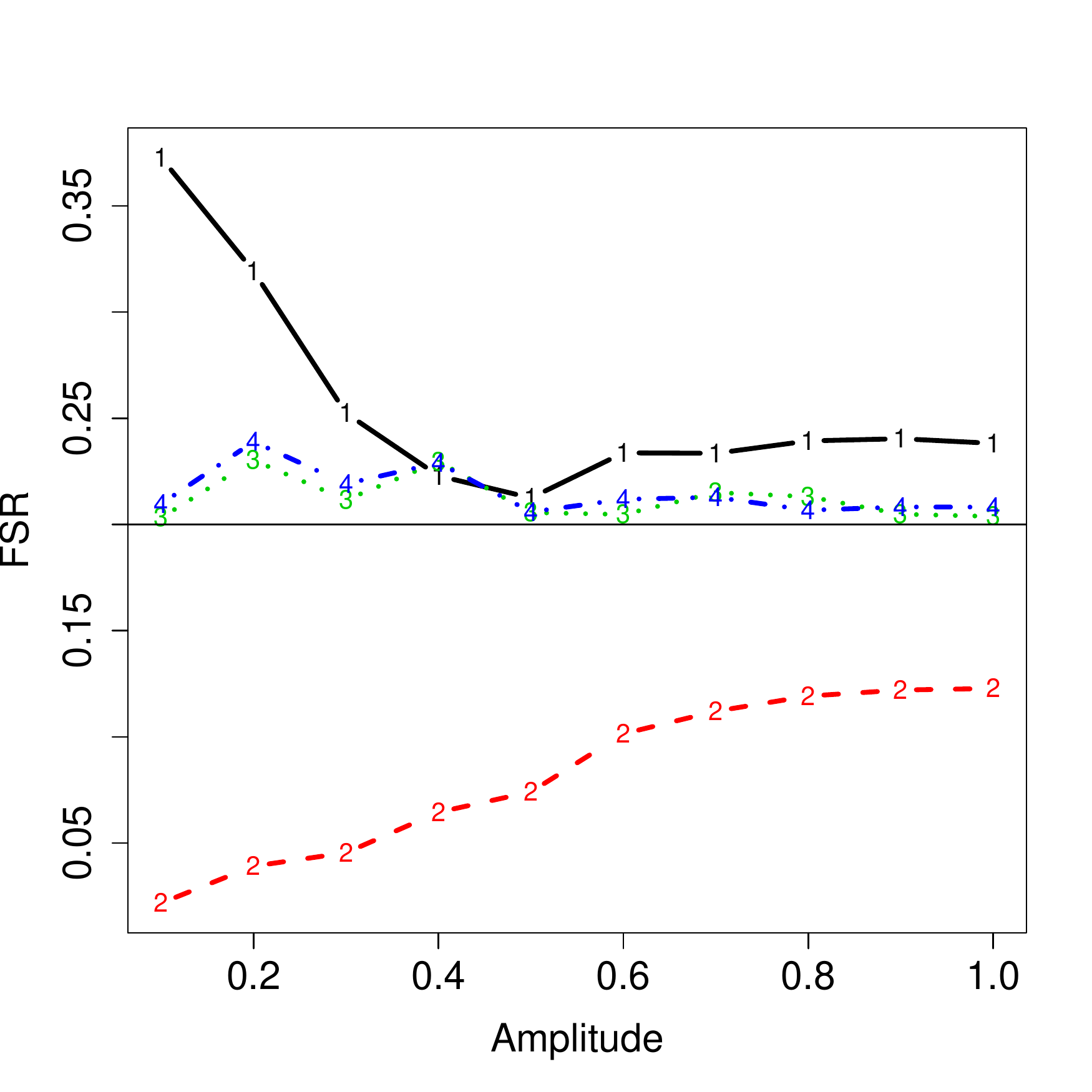}
        \caption{False selection rate vs $A$}
    \end{subfigure}%
    \begin{subfigure}{.5\textwidth}
        \centering
        \includegraphics[scale=0.35]{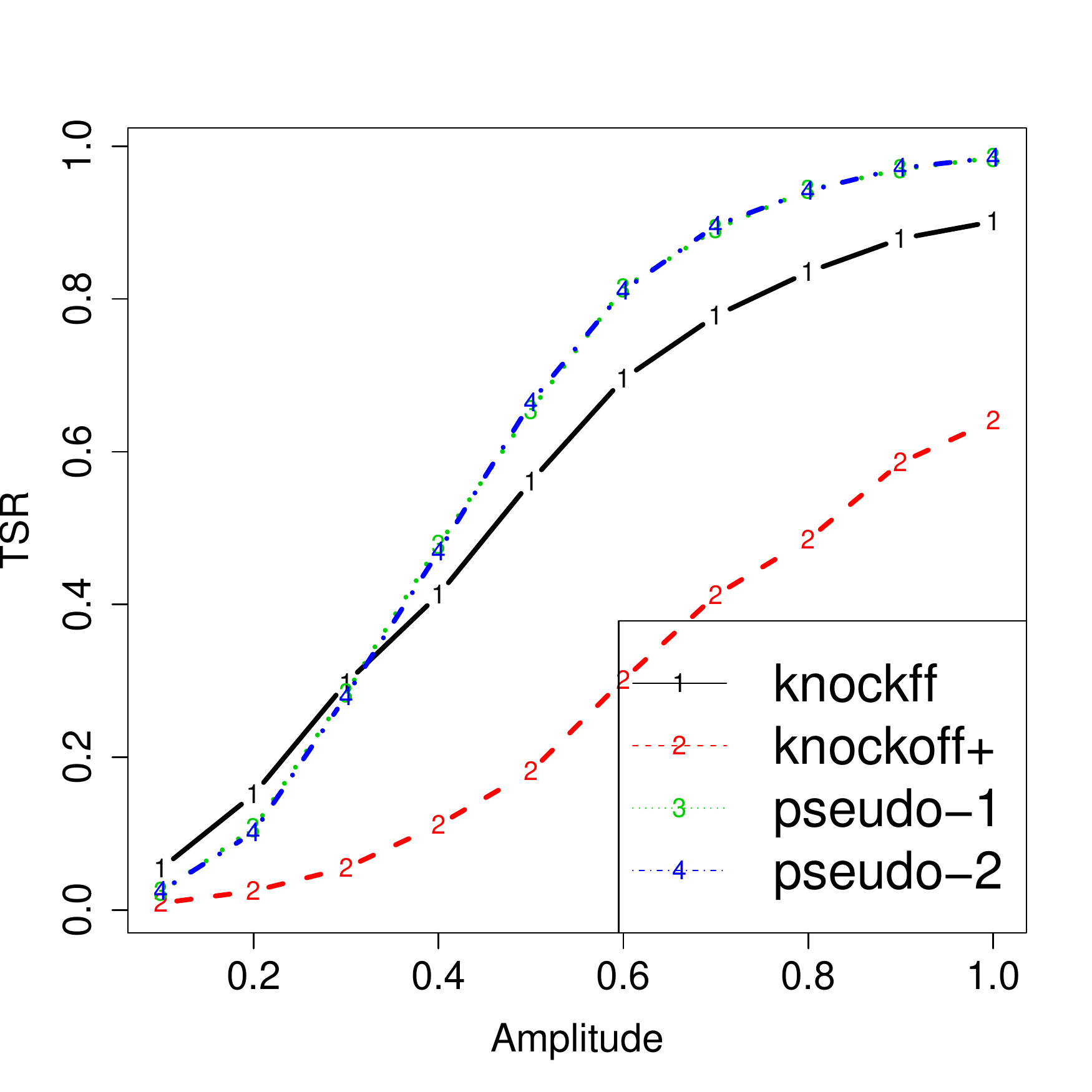}
        \caption{True selection rate vs $A$}
    \end{subfigure}
    \caption{Logistic model; Performances under different coefficient amplitude at $\alpha = 0.2$.}
    \label{fig:amp_logistic}
\end{figure}

\begin{figure}[!h]
    \begin{subfigure}{.5\textwidth}
        \centering
        \includegraphics[scale=0.35]{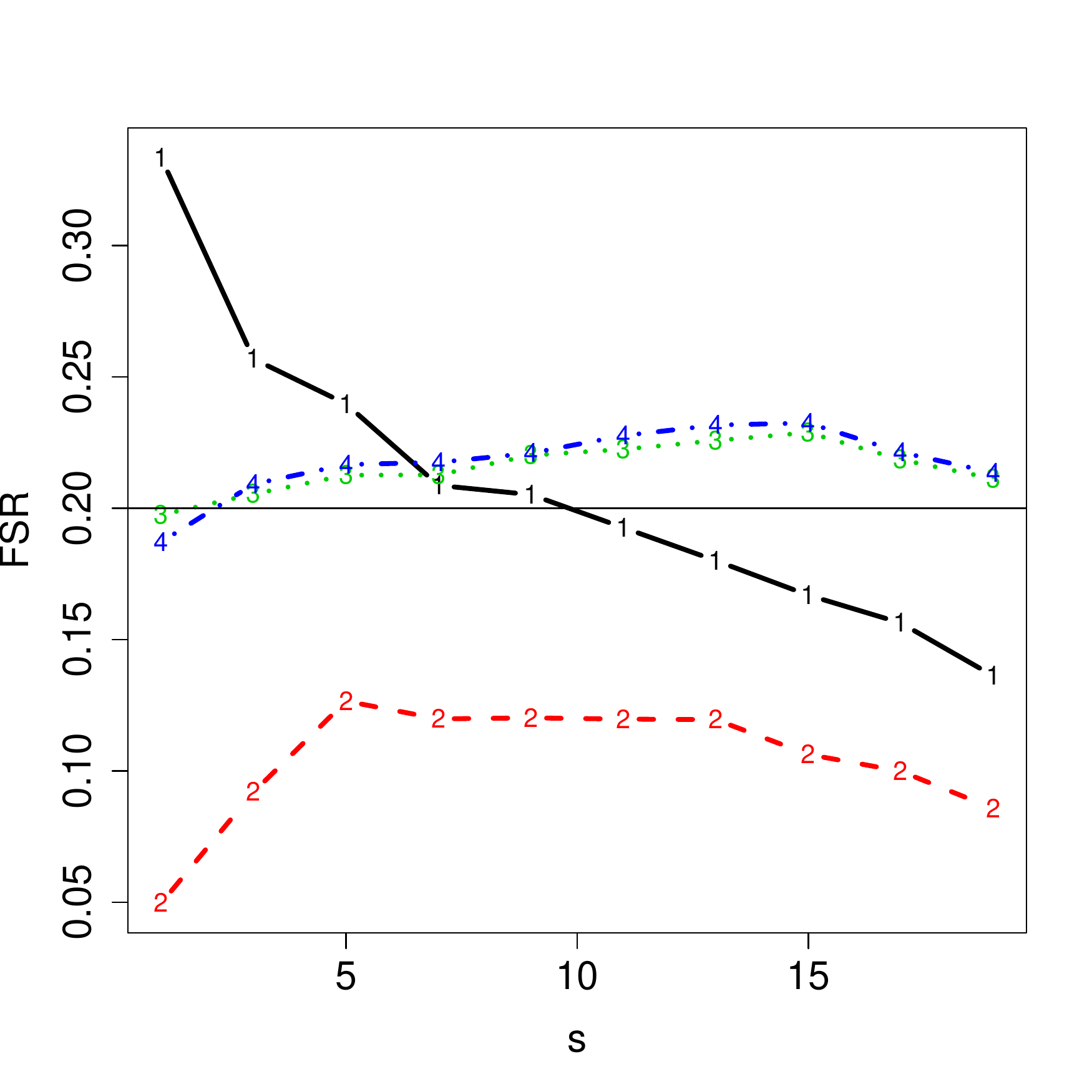}
        \caption{False selection rate vs $s$}
    \end{subfigure}%
    \begin{subfigure}{.5\textwidth}
        \centering
        \includegraphics[scale=0.35]{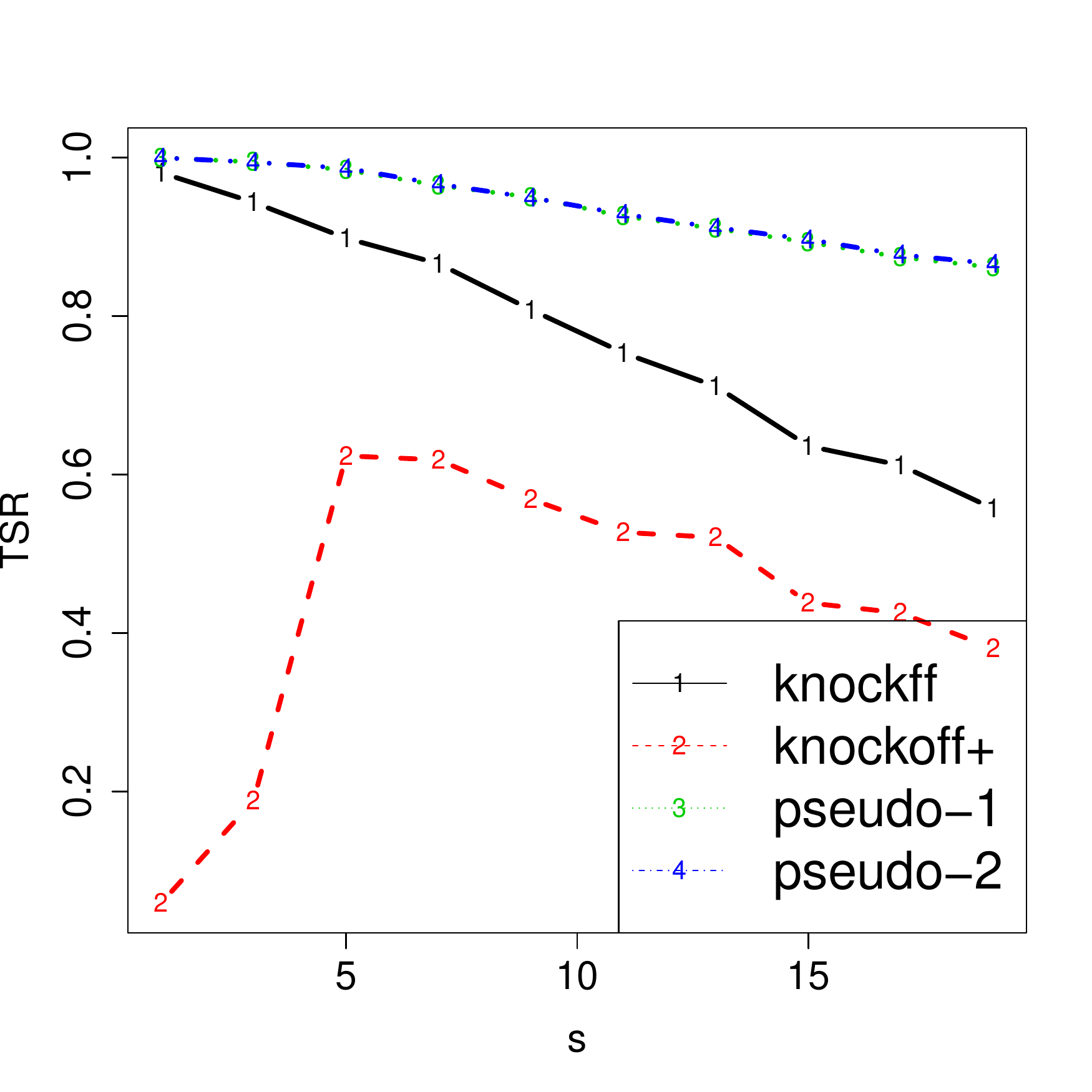}
        \caption{True selection rate vs $s$}
    \end{subfigure}
    \caption{Logistic model; Performances under different number of nonzero coefficients at $\alpha = 0.2$.}
    \label{fig:spa_logistic}
\end{figure}

\newpage
\section{Simulation results for Cox model}
\label{sec:Cox}
In this section, the penalized Cox model is considered. Suppose the data are of the form $(Y_i, X_i, \delta_i), i =1, \ldots, n$, where $Y_i$ is the observed time, $X_i$ is covariate and $\delta_i$ is censoring indicator with $0$ means right-censoring and $1$ indicates failure. The Lasso estimator for Cox model at $\lambda$ is defined as
\begin{equation*}
\widehat{\boldsymbol{\beta}}(\lambda; \mathbb{Y}, \mathbb{X}) =
\arg\min_{\boldsymbol{\beta}\in\mathbb{R}^p}
\left \{ -\frac{1}{n} \sum_{i: \delta_i=1}^n \left [  X_i^\T \boldsymbol{\beta} - \log \sum_{j:Y_j\geq Y_i} \exp(X_j^\T \boldsymbol{\beta}) \right]
+\lambda \sum_{j = 1}^p  |\beta_j| \right \}.
\end{equation*}

In the simulation, survival time $T_i$ follows a Weibull distribution with shape parameter 1 and scale $(0.01)\exp(X_i^T\boldsymbol{\beta}_0)$. Censoring time $C_i$ is exponential distributed with mean $1000$.
Similar to simulation study of penalized regression, we consider four settings: different predictor dimensions, correlation magnitudes, signal amplitudes, and number of nonzero coefficients. The censoring percentage are around $10\%$ to $35\%$ across different settings.

Simulation results are presented in Figure  \ref{fig:dim_cox}, \ref{fig:cor_cox}, \ref{fig:amp_cox}, and \ref{fig:spa_cox}. A similar pattern as seen for penalized regression can be observed from those figures. In terms of TSR, the pseudo-variable methods and the knockoff method have better performances than the knockoff+, which is too conservative in some cases. Then for FSR, only the knockoff method exceeds the target error rate significantly.

Similar to implementation of knockoff and knockoff+ method for penalized logistic model,  $create\_equicorrelated$ function is called first to construct knockoff variables, and \emph{glmnet} function is called with \emph{family = cox} to get when the original variables and knockoff enter the solution path. Signed maximum statistics are used for knockoff and knockoff+ method.

\begin{figure}[!h]
    \begin{subfigure}{.5\textwidth}
        \centering
        \includegraphics[scale=0.32]{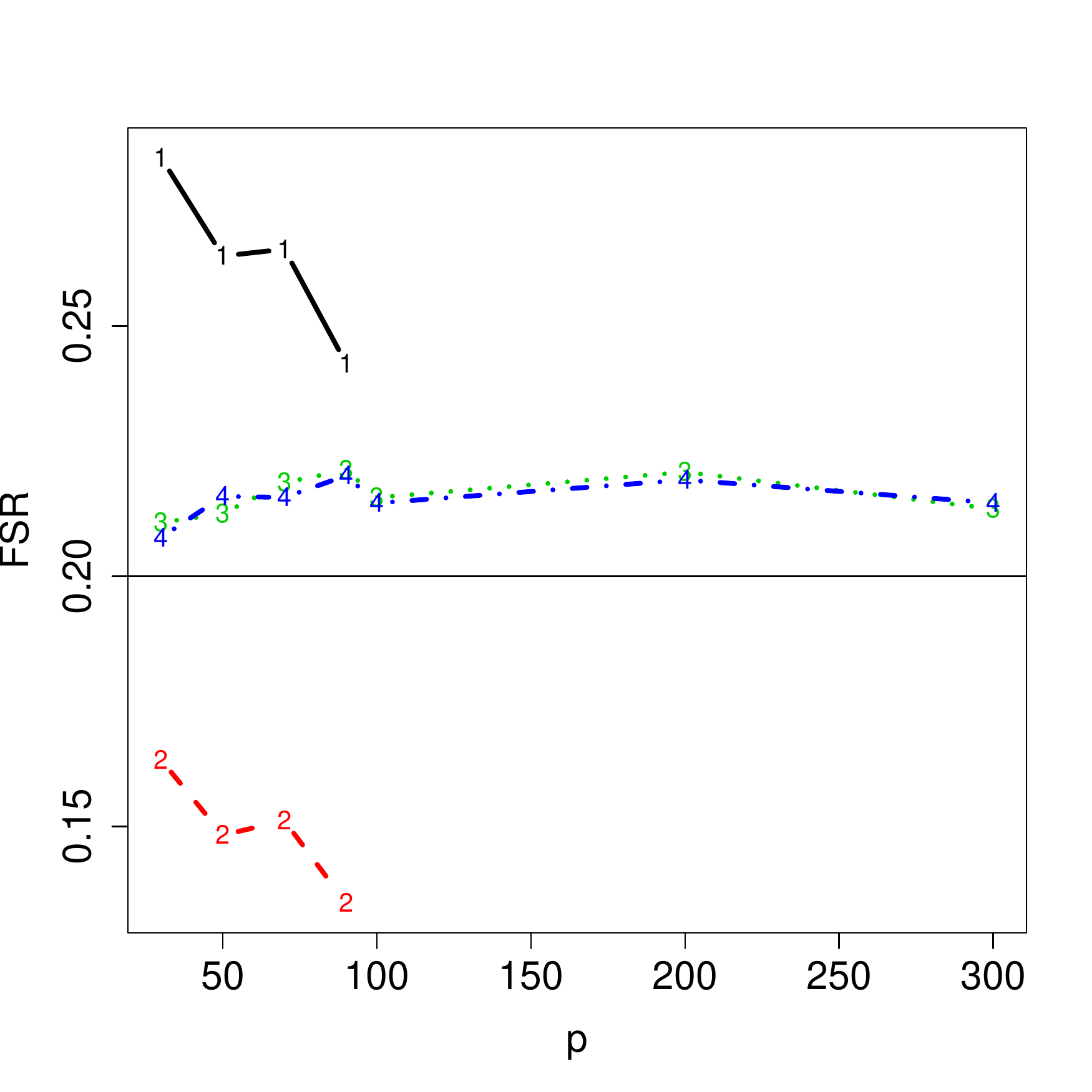}
        \caption{False selection rate vs $p$}
    \end{subfigure}%
    \begin{subfigure}{.5\textwidth}
        \centering
        \includegraphics[scale=0.32]{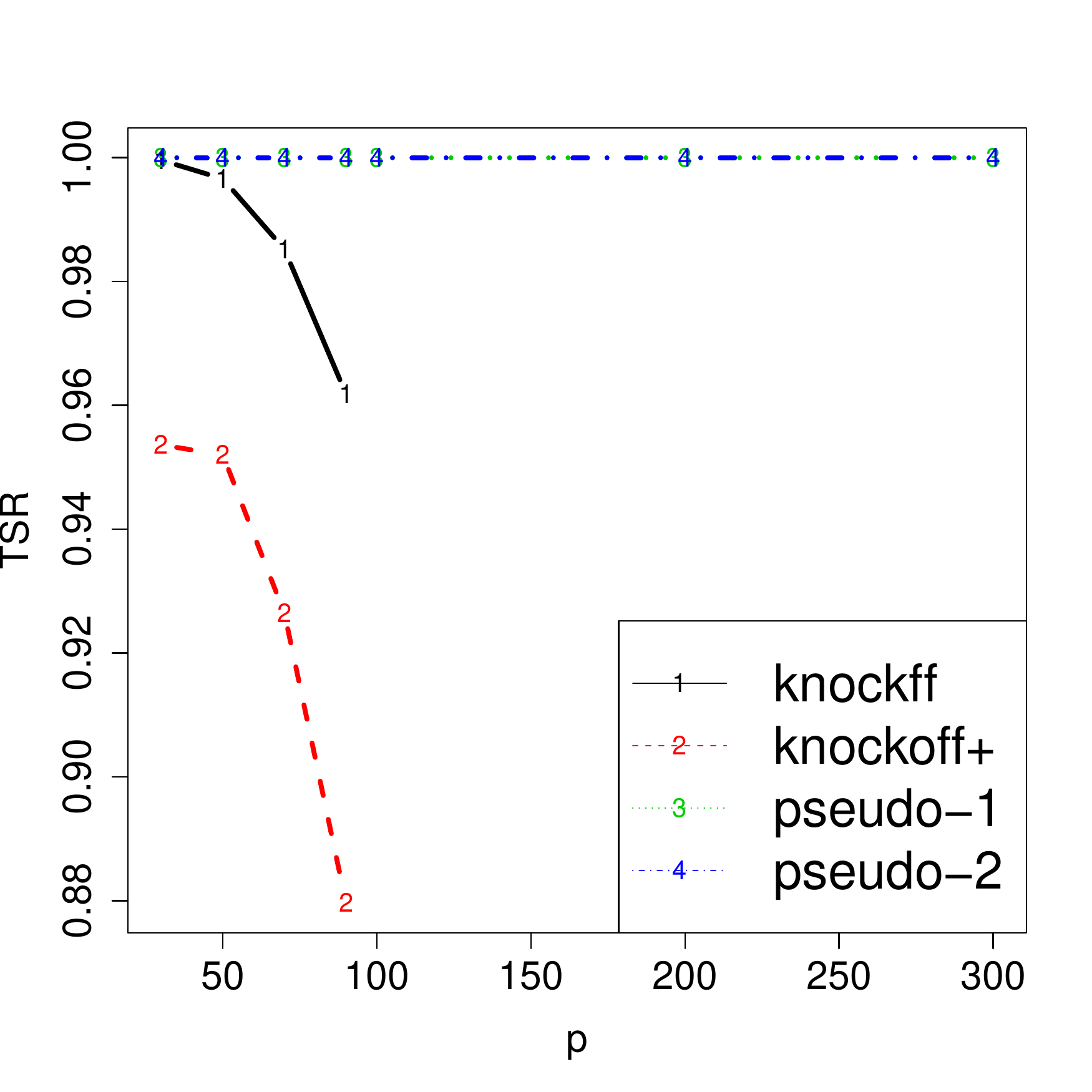}
        \caption{True selection rate vs $p$}
    \end{subfigure}
    \caption{Cox model; Performances under different dimensions at $\alpha = 0.2$. Knockff and Knockoff+ method only work if $n > 2p$.}
    \label{fig:dim_cox}
\end{figure}

\begin{figure}[!h]
    \begin{subfigure}{.5\textwidth}
        \centering
        \includegraphics[scale=0.32]{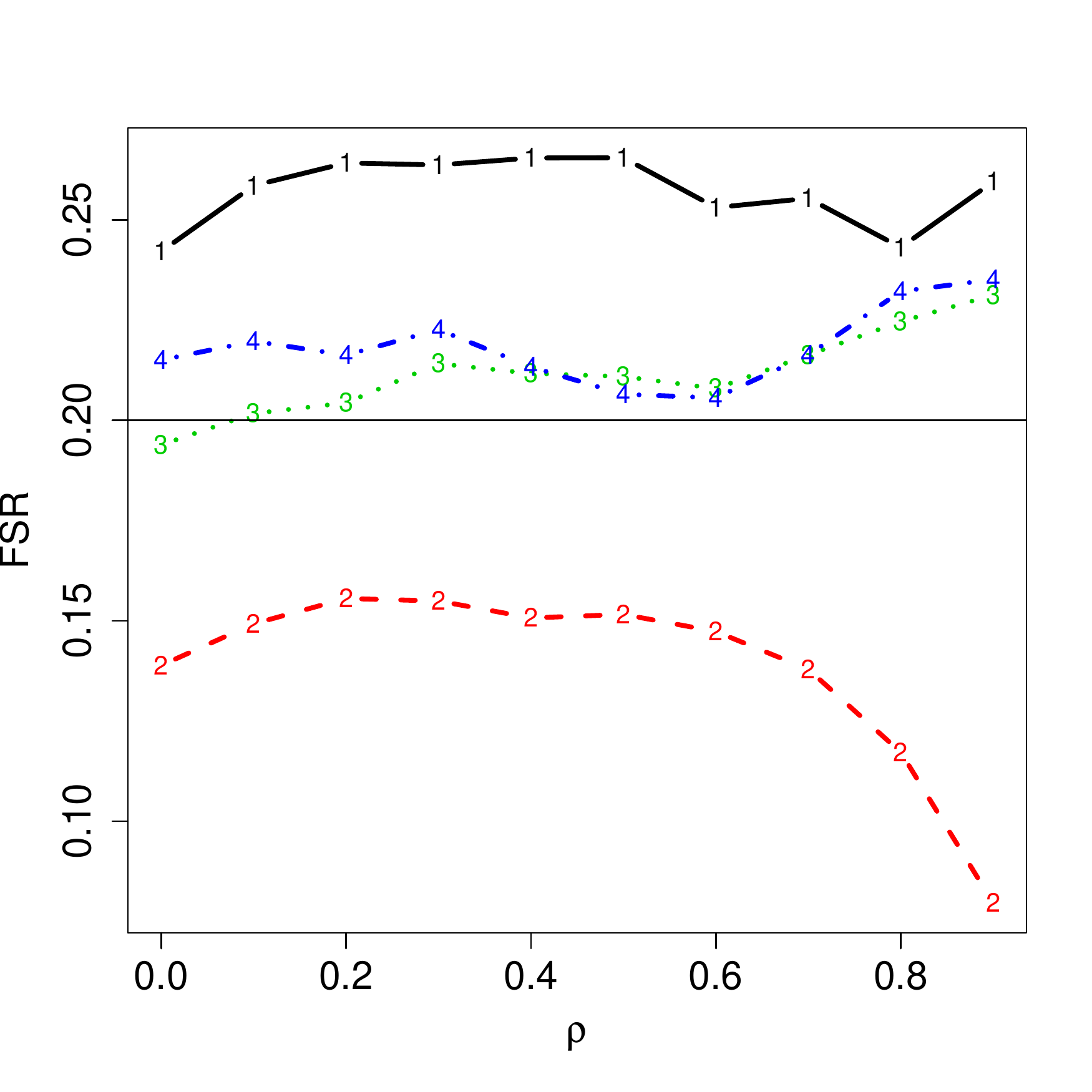}
        \caption{False selection rate vs $\rho$}
    \end{subfigure}%
    \begin{subfigure}{.5\textwidth}
        \centering
        \includegraphics[scale=0.32]{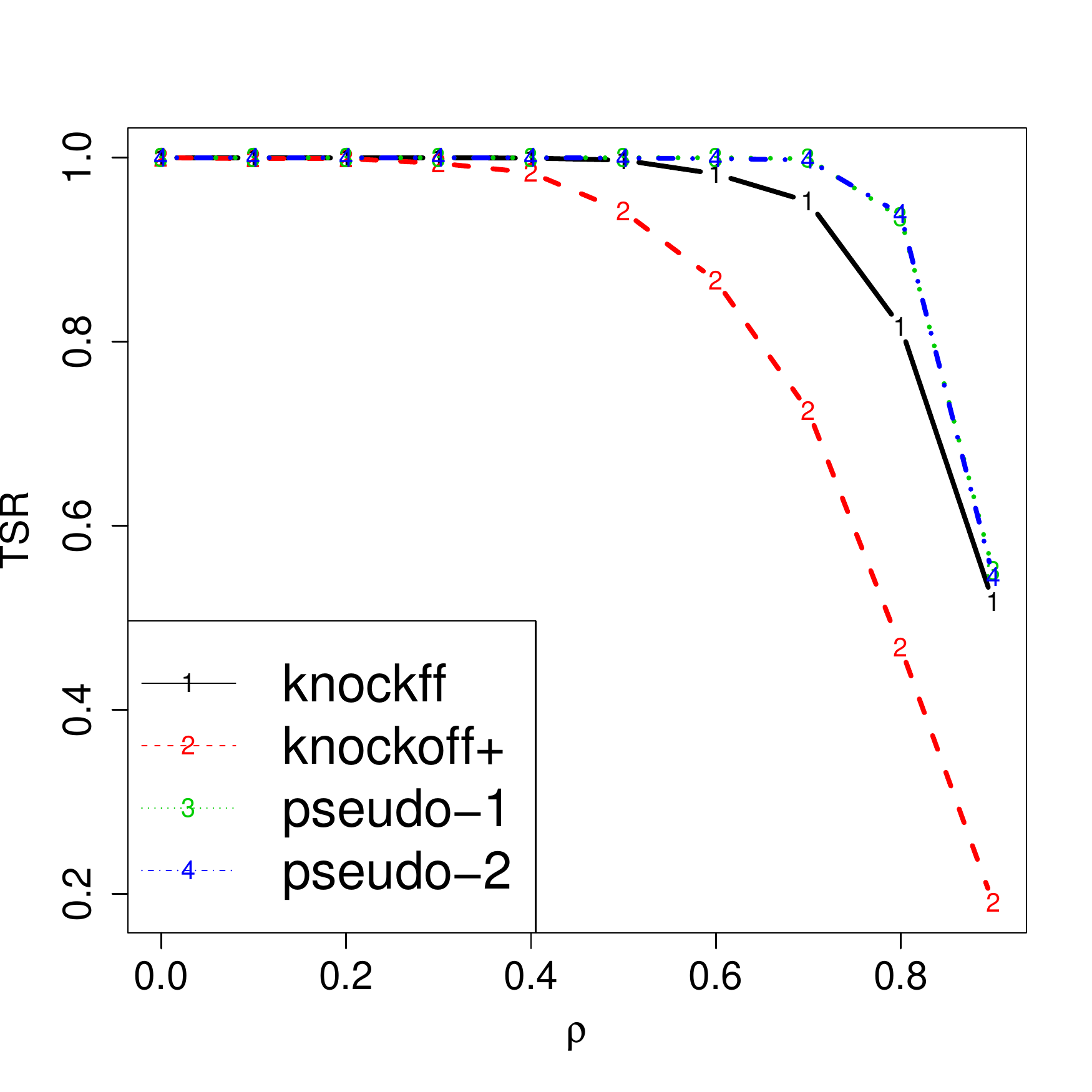}
        \caption{True selection rate vs $\rho$}
    \end{subfigure}
    \caption{Cox model; Performances under different correlations at $\alpha = 0.2$.}
    \label{fig:cor_cox}
\end{figure}

\begin{figure}[!h]
    \begin{subfigure}{.5\textwidth}
        \centering
        \includegraphics[scale=0.35]{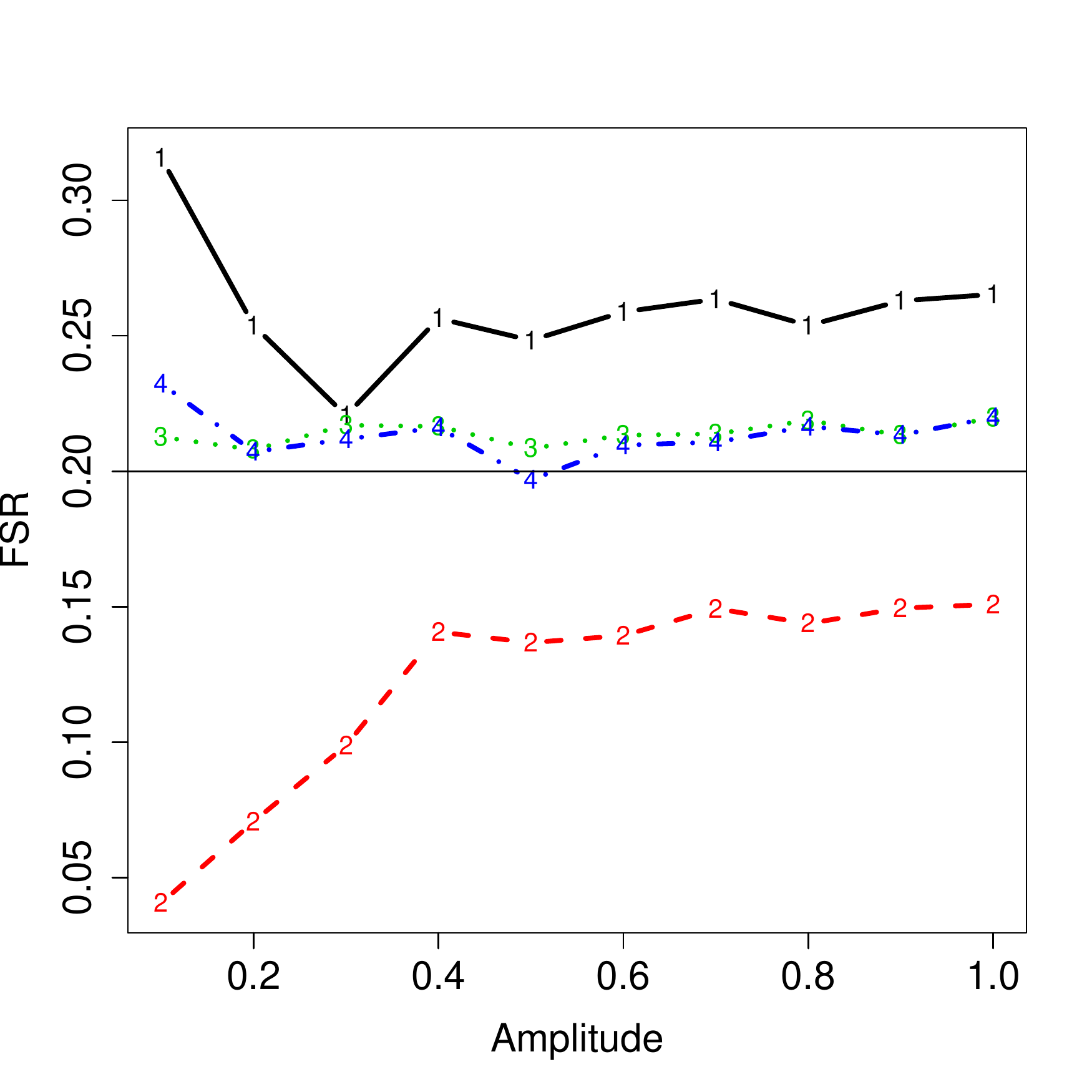}
        \caption{False selection rate vs $A$}
    \end{subfigure}%
    \begin{subfigure}{.5\textwidth}
        \centering
        \includegraphics[scale=0.35]{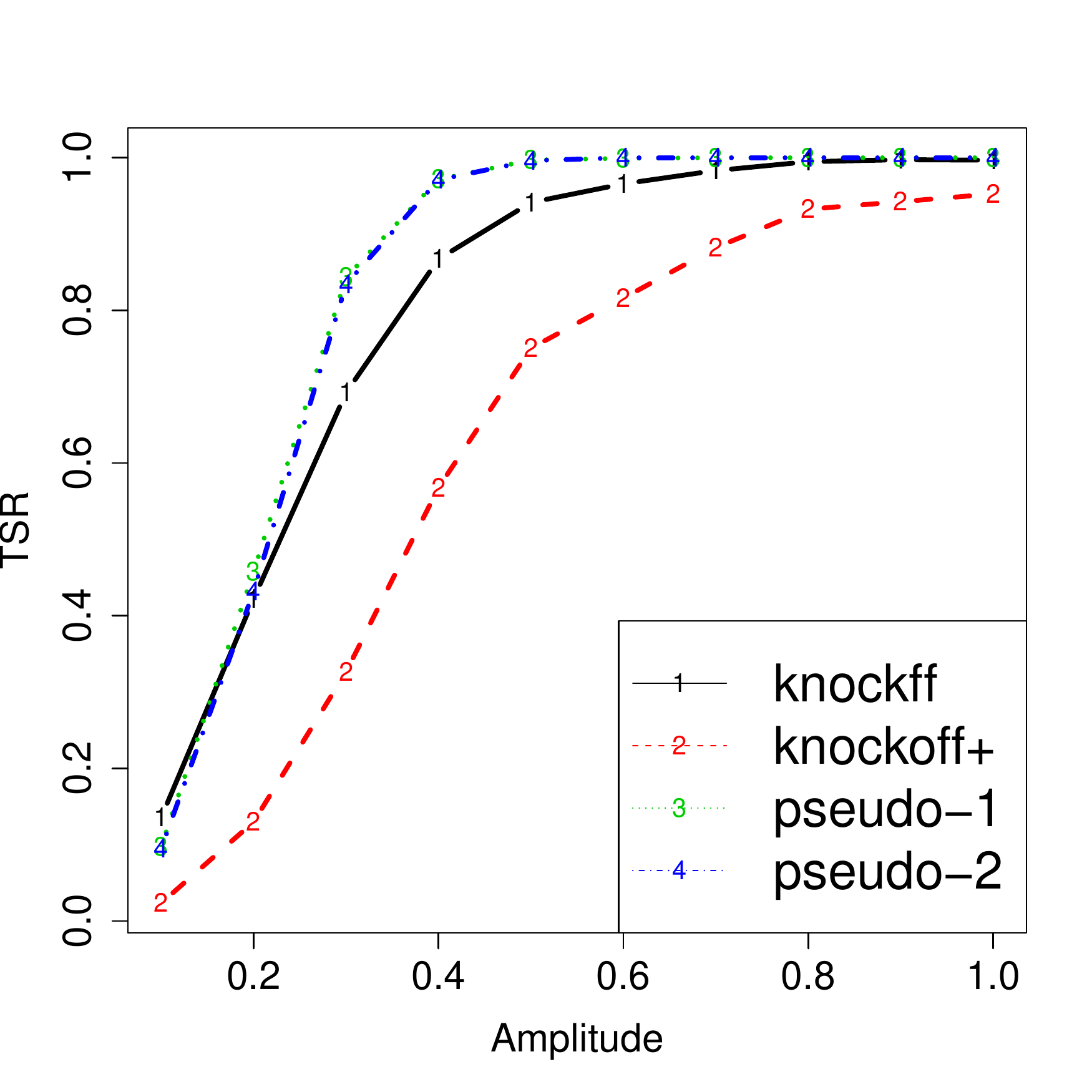}
        \caption{True selection rate vs $A$}
    \end{subfigure}
    \caption{Cox model; Performances under different coefficient amplitude at $\alpha = 0.2$.}
    \label{fig:amp_cox}
\end{figure}

\begin{figure}[!h]
    \begin{subfigure}{.5\textwidth}
        \centering
        \includegraphics[scale=0.35]{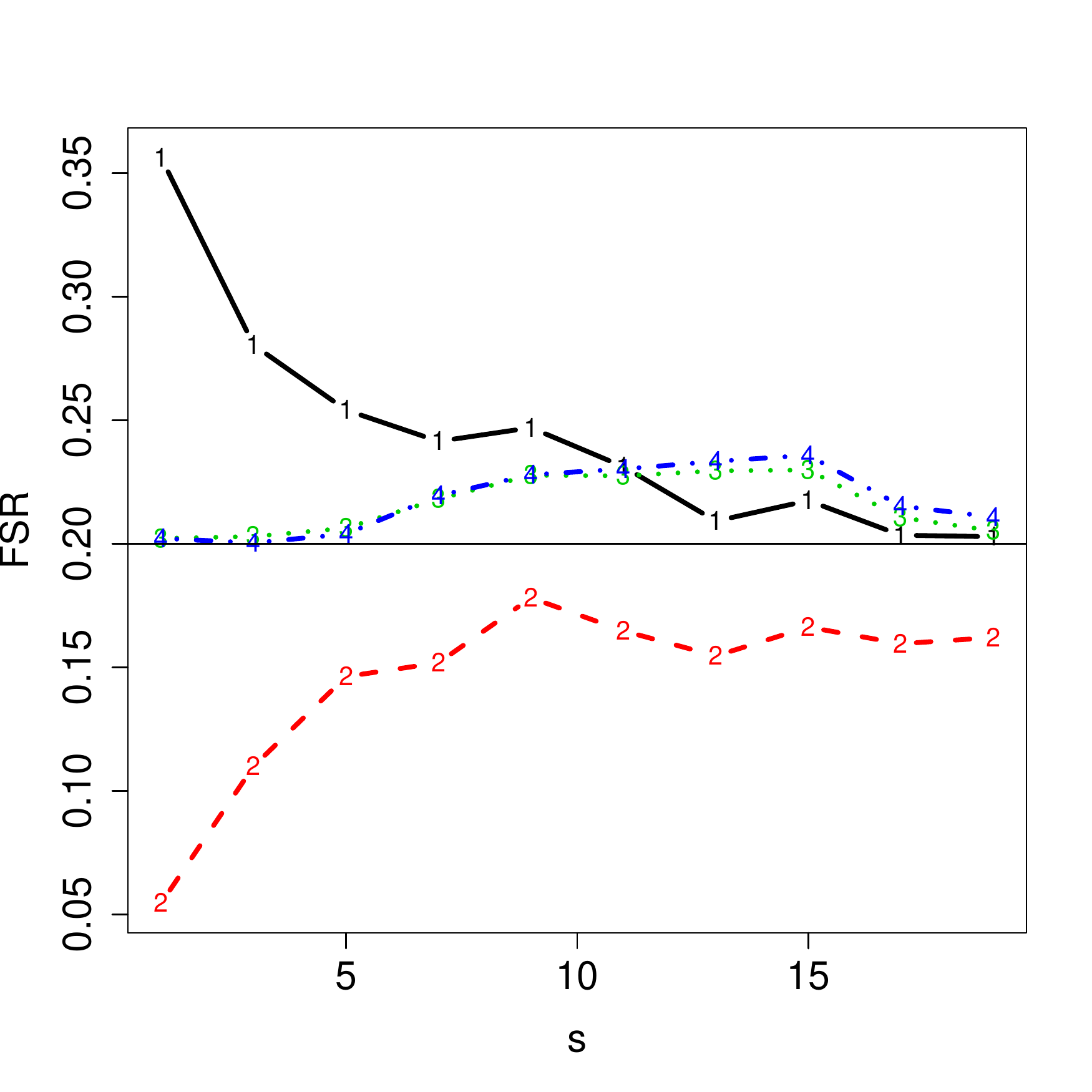}
        \caption{False selection rate vs $s$}
    \end{subfigure}%
    \begin{subfigure}{.5\textwidth}
        \centering
        \includegraphics[scale=0.35]{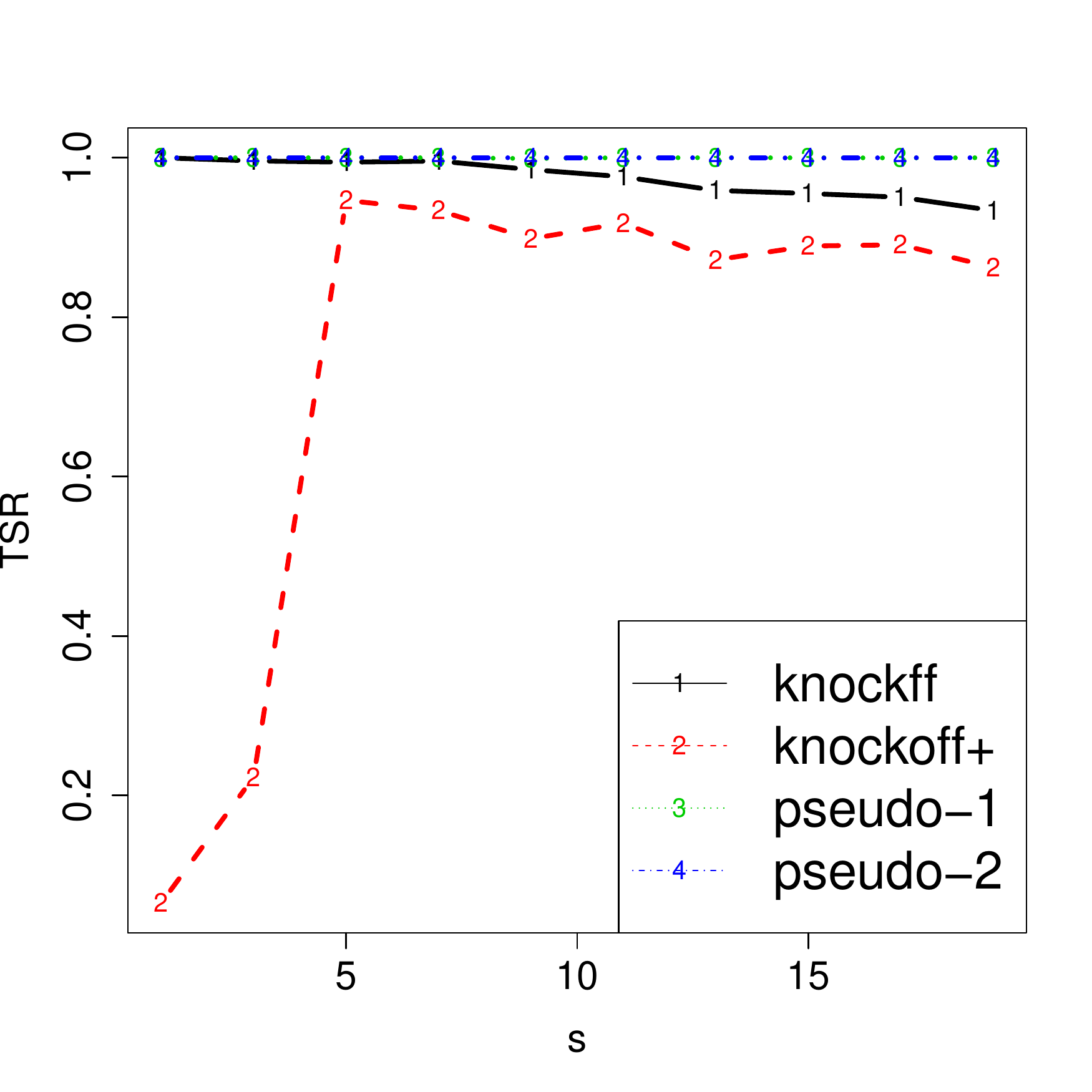}
        \caption{True selection rate vs $s$}
    \end{subfigure}
    \caption{Cox model; Performances under different number of nonzero coefficients at $\alpha = 0.2$.}
    \label{fig:spa_cox}
\end{figure}
\newpage

\section{Pseudo-variables algorithm for screening}
In this section, we introduce a general algorithm for screening based on pseudo-variables. The intuition behind this algorithm is that the pseudo-variables can be used to assess the tuning parameter selected. Therefore, one can choose a tuning parameter that controls the percentage of pseudo-variables in the selected model.  The Lasso estimator is defined as follows

\begin{equation}
\widehat{\boldsymbol{\beta}}(\lambda_n; \mathbb{Y}, \mathbb{X}) =
\arg\min_{\boldsymbol{\beta}\in\mathbb{R}^p}
\left \{\frac{1}{2} ||\mathbb{Y} - \mathbb{X}\boldsymbol{\beta}||^2
+\lambda_n \sum_{j = 1}^p  |\beta_j| \right \},
\end{equation}
where $\lambda_n$ is the tuning parameter. The detailed procedure for finding the tuning parameter is summarized as follows

\begin{itemize}
    \item Step 0: Fix a sequence of lambda $\lambda_{(1)} < \lambda_{(2)} < \cdots < \lambda_{(m)}  $
    \item Step 1:  Generate $\mathbb{X}_{\mathrm{pseudo}}$ by direct permuting the rows of $\mathbb{X}$
    \item Step 2: Fit Lasso using $\mathbb{Y}$ and $\mathbb{X}_{\mathrm{all}} = (\mathbb{X}, \mathbb{X}_{\mathrm{pseudo}})$,  calculate $\widehat{A}_{\lambda_{(i)}}^{all} = \{j: \widehat{\boldsymbol{\beta}}_j(\lambda_{(i)}; \mathbb{Y}, \mathbb{X}_{\mathrm{all}}) \neq 0 \}$, then
    \begin{align*}
    \widehat{p}_{\lambda_{(i)}} = \frac{ \# \{j \in \widehat{A}^{all}_{\lambda_{(i)}} ~\text{and $j$-th variable is pseudo}\} } { \max\left [\# \{j \in \widehat{A}^{all}_{\lambda_{(i)}} ~\text{and $j$-th variable is not pseudo}\}, 1  \right ]}
    \end{align*}
    \item Step 3: Repeat Steps 1 and 2 $B$ times, and calculate $\widehat{p}_{\lambda_{(i)}} = 1/B \sum_{b=1}^B \widehat{p}^b_{\lambda_{(i)}} $
    \item Step 4: Select tuning parameter
    $\widehat{\lambda}_n = \min \{\lambda_{(i)}:  \widehat{p}_{\lambda_{(i)}} \leq
    \alpha_n \} $, where $\alpha_n$ is a constant.
\end{itemize}
After obtaining $\widehat{\lambda}_n$, one can then fit Lasso at $\widehat{\lambda}_n$  with $\mathbb{Y}$ and $X_{\mathrm{all}} = (\mathbb{X}, \mathbb{X}_{pseudo})$. And then select those variables in the active set by excluding pseudo-variables. In the simulation studies, we fix $\alpha_n = 0.2$ and $B = 20$.

\subsection{Theoretical properties}
In this subsection, we prove that the algorithm above leads to consistent variable selection under certain conditions. Denote $\boldsymbol{\epsilon} = (\epsilon_1, \ldots, \epsilon_n)^\T.$ Denote $r_n \prec s_n$ if  $r_n = o(s_n)$,  and $r_n \succ s_n$ if $s_n = o(r_n)$.
To prove the asymptotic consistency,  we assume the following three conditions:
\begin{itemize}
    \item[(A1)] There exists positive sequences $\{r_n\}$ and $\{s_n\}$ such that Lasso method with tuning parameter $\lambda_n$ is selection consistent whenever $r_n \prec \lambda_n \prec s_n$.
\end{itemize}

Use $\lambda_n^*$ to denote such a tuning parameter with  $r_n \prec \lambda_n^* \prec s_n$. For such $\lambda_n^*$, then
\begin{align}
\label{eq1}
P(\widehat{A}_n(\lambda_n^*) = A_0) \geq 1 - k_n,
\end{align} for some $k_n \rightarrow 0$.

\begin{itemize}

    \item  [(A2)]
    For any $\tau >0$, there exists positive constant $c_0(\tau)$ and $n_0(\tau)$ such that, for $n> n_0(\tau)$,
    $
    P\left ( \left \{\inf_{\tau r_n \leq \lambda \leq \lambda_n^*} \mathds{1} \left( \widehat{A}_n(\lambda) = A_0 \right ) \right \} = 1 \right) \geq 1 - c_0(\tau),
    $
    where $c_0(\tau)$ converges to zero as $\tau \rightarrow \infty $.
\end{itemize}

\begin{itemize}

    \item [(A3)] For any pseudo variable $j$, and any $\tau >0$, there exists positive constant $c_1(\tau)$ and $n_1(\tau)$ such that, for $n > n_1(\tau)$,
    $
    P\left(\left\{\inf_{\lambda < \tau r_n} \mathds{1} \left (j \in \widehat{A}_n(\lambda_n) \right ) \right\} = 1  | \mathbb{X}, \boldsymbol{\epsilon} \right) \geq c_1(\tau) ,
    $
    almost surely.
\end{itemize}

 For assumption (A3), when conditioning on $\mathbb{X}$ and $\boldsymbol{\epsilon}$, the randomness come from $\mathbb{X}_{\mathrm{pseudo}}$. Assumptions (A1) and (A2) have been verified in \cite{sun2013consistent}. By similar argument as \cite{sun2013consistent}, it is easy to verify assumption (A3) if assuming a strong condition that $\mathbb{X}_{\mathrm{pseudo}}$ are orthogonal to $\mathbb{X}$. Under above assumptions, we have the following theorem

\begin{theorem}
    \label{thm:asycons}
    Under Assumptions 1, 2 and 3, the tuning parameter $\widehat{\lambda}_n$ selected in the above algorithm leads to consistent variable selection, i.e.,
    $\lim_{n\rightarrow \infty} \lim_{B \rightarrow \infty} P\left ( \widehat{A}_n(\widehat{\lambda}_n) = A_0 \right ) = 1$ provided $1 \succ \alpha_n \succ k_n$, where $k_n$ is defined in Eq. \ref{eq1}.
\end{theorem}

\subsection{Proof of Theorem \ref{thm:asycons}}
\begin{proof}
    Define $S_1 = \{\lambda: \lambda > \lambda_n^*\}$ and, $S_2 = \{\lambda: \tau r_n > \lambda \}$. First, we show that for $\lim_{n \rightarrow \infty}P(\widehat{\lambda}_n \in S_1 \cup S_2) \rightarrow 0$.
    For $S_1$, by definition of $\widehat{\lambda}_n$ and $\lambda_n^*$, we have
    \begin{align*}
    P(\widehat{\lambda}_n \leq \lambda_n^*) & \geq P(\widehat{p}_{\lambda^*_n} \leq \alpha_n) \\
    &= 1 - P(\widehat{p}_{\lambda^*_n} > \alpha_n) \\
    & \geq  1 - \frac{E(\widehat{p}_{\lambda^*_n})}{\alpha_n} \\
    & = 1 - \frac{E(\widehat{p}^b_{\lambda^*_n})}{\alpha_n} \\
    & \geq 1 - \frac{k_n}{\alpha_n} \rightarrow 1,
    \end{align*}
    where the last inequality holds because of $E(\widehat{p}_{\lambda^*_n}) \leq k_n$, which is implied by assumption A1.
    Therefore $\lim_{n \rightarrow \infty}P(\widehat{\lambda}_n \in S_1) = 0$.

    Then for $S_2$, by the strong law of large numbers
    $$\lim_{B \rightarrow \infty} \frac{1}{B}\sum_{b=1}^{B} \inf_{\lambda \in S_2}\widehat{p}_{\lambda}^b = E(\inf_{\lambda \in S_2}\widehat{p}_{\lambda}^b|\mathbb{X}, \boldsymbol{\epsilon}).$$ And by Assumption A3, we know that for any pseudo variable $j$, \begin{align*}
    P\left(\left\{\inf_{\lambda < \tau r_n} \mathds{1}\left (j \in \widehat{A}_n(\lambda) \right ) \right\} = 1  | \mathbb{X}, \boldsymbol{\epsilon} \right) \geq c_1(\tau).
    \end{align*}
    This implies
    $$P\left (E(\inf_{\lambda \in S_2}\widehat{p}_{\lambda}^b|\mathbb{X}, \boldsymbol{\epsilon}) \geq c_1(\tau) / p\right ) = 1.$$
    And since $\inf_{\lambda \in S_2}\widehat{p}_{\lambda} \geq B^{-1}\sum_{b=1}^{B} \inf_{\lambda \in S_2}\widehat{p}_{\lambda}^b$, $$P\left(\inf_{\lambda \in S_2} \widehat{p}_{\lambda} \geq   c_1(\tau) /p \right)
    = 1 - o(1).
    $$
    This implies $\lim_{n \rightarrow \infty}P\left(\inf_{\lambda \in S_2} \widehat{p}_{\lambda} < \alpha_n \right) = 0$. Therefore, $\lim_{n \rightarrow \infty} P(\widehat{\lambda}_n \in S_2)  = 0$.  So $\lim_{n \rightarrow \infty}P(\widehat{\lambda}_n \in S_1 \cup S_2) = 0$, i.e., $\lim_{n \rightarrow \infty} P(\tau r_n \leq \widehat{\lambda}_n \leq \lambda_n^*) = 1$.
    Then,
    \begin{align*}
    P\left (\widehat{A}_n(\widehat{\lambda}_n) = A_0 \right ) & \geq P\left (\widehat{A}_n(\widehat{\lambda}_n) = A_0, \tau r_n \leq \widehat{\lambda}_n \leq \lambda_n^* \right ) \\
    & \geq  P\left ( \left \{\inf_{\tau r_n \leq \lambda \leq \lambda_n^*} \mathds{1}(\widehat{A}_n(\lambda) = A_0) \right \} = 1 \right) + P(\tau r_n \leq \widehat{\lambda}_n \leq \lambda_n^*) - 1.
    \end{align*}
Therefore by assumption A2, $\lim_{n\rightarrow \infty} \lim_{B\rightarrow \infty}P\left (\widehat{A}_n(\widehat{\lambda}_n) = A_0 \right ) \geq 1 - c_0(\tau)$. It holds for any $\tau$. Then let $\tau \rightarrow \infty$, we have $\lim_{n\rightarrow \infty} \lim_{B\rightarrow \infty}P\left (\widehat{A}_n(\widehat{\lambda}_n) = A_0 \right ) = 1$.
\end{proof}

\end{document}